%% file: main.tex
\def\NOEDITS{}
\documentclass[11pt,a4paper]{article}
\input{preamble}

\begin{document}
\input{content}

\input{SI}

\end{document}

%% file: preamble.tex
\newif\ifrsos \ifdefined\RSOSMODE \rsostrue \fi
\newif\ifsionly \ifdefined\SIONLYMODE \sionlytrue \fi

\usepackage[a4paper, margin=2.6cm, top=2.8cm, bottom=2.6cm]{geometry}
\usepackage[T1]{fontenc}
\usepackage[utf8]{inputenc}
\usepackage{lmodern}
\usepackage{microtype}
\usepackage{amsmath, amssymb, amsthm, mathtools}
\usepackage[version=4]{mhchem}

\usepackage{graphicx}
\usepackage{float}          % enables [H] exact float placement
\usepackage{placeins}       % \FloatBarrier
\makeatletter
\AtBeginDocument{%
  \let\ge@section\section
  \renewcommand\section{\FloatBarrier\ge@section}%
  \let\ge@subsection\subsection
  \renewcommand\subsection{\FloatBarrier\ge@subsection}%
  \let\subsectionnobarrier\ge@subsection   % for a heading that must not flush pending floats%
}
\makeatother
\usepackage{upgreek}
\usepackage{xcolor}
\usepackage{fontawesome5}   % icons in the arXiv title block
\usepackage{tikz}

\usetikzlibrary{calc, decorations.pathmorphing, decorations.markings,arrows.meta,
                shapes.geometric, decorations.pathreplacing, positioning}

\usepackage{booktabs}
\usepackage{multirow}
\usepackage{array}
\usepackage{threeparttable}

\usepackage{etoolbox}
\ifrsos\ifsionly\def\needSI{}\fi\else\def\needSI{}\fi
\ifdefined\needSI
\usepackage[resetlabels]{multibib}    % separate Supplementary References (SI), 2026-10-21
\newcites{SI}{Supplementary References}
\fi
\ifrsos\usepackage{xr-hyper}\fi
\usepackage[hidelinks, colorlinks=false, breaklinks=true]{hyperref}
\usepackage{orcidlink}       % after hyperref (it loads hyperref itself)
\usepackage{url}
\usepackage{xurl}            % URLs may break at any character (no overfull lines)
\usepackage{cite}

\usepackage{titlesec}
\usepackage{titletoc}        % partial table of contents for the SI
\usepackage[most]{tcolorbox}

\usepackage[font=small, labelfont=bf, labelsep=period]{caption}

\definecolor{kindigo}{HTML}{2E3457}     % deep indigo
\definecolor{kochre}{HTML}{C9A227}      % muted ochre
\definecolor{kterra}{HTML}{B85042}      % terracotta
\definecolor{kcream}{HTML}{F4EFE6}      % parchment
\definecolor{kink}{HTML}{1A1A1A}        % soft black
\definecolor{kgrey}{HTML}{6B6B6B}       % stone gray
\definecolor{kteal}{HTML}{3E7E7E}       % faded teal accent
\newif\ifshowedits
\showeditstrue    % revision marks in blue in every version (working phase)
\ifdefined\NOEDITS \showeditsfalse \fi   % final/submission PDFs: put \def\NOEDITS{} first in the wrapper
\ifshowedits
  \definecolor{knew}{HTML}{1F5FBF}
\else
  \colorlet{knew}{black}
\fi        % revision blue (Pass A--C edits)

\newcommand{\Z}{\mathbb{Z}}

\newcommand{\Stab}{\operatorname{Stab}}
\newcommand{\Orb}{\operatorname{Orb}}
\newcommand{\Fix}{\operatorname{Fix}}
\newcommand{\Hor}{H_{\text{orbit}}}
\newcommand{\Hmic}{H_{\text{micro}}}
\newcommand{\Hmac}{H_{\text{macro}}}

\newtheorem{theorem}{Theorem}

\newtheorem{definition}{Definition}

\newcommand{\subglyph}{%
  \tikz[baseline=-0.5ex]{\fill[kochre] (0,0) rectangle (0.16,0.16);}%
}

\titleformat{\section}
  {\normalfont\Large\bfseries\color{kindigo}}
  {\thesection}
  {0.6em}
  {}
  [\vspace{0.15em}\noindent\textcolor{kindigo!30}{\rule{\linewidth}{0.4pt}}]

\titleformat{\subsection}
  {\normalfont\large\bfseries\color{kink}}
  {\subglyph\hspace{0.4em}\thesubsection}
  {0.5em}
  {}

\titleformat{\subsubsection}
  {\normalfont\normalsize\bfseries\color{kgrey}}
  {\thesubsubsection}
  {0.5em}
  {}

\titlespacing*{\section}{0pt}{1.4em}{0.6em}
\titlespacing*{\subsection}{0pt}{1.0em}{0.4em}

\newtcolorbox{kcallout}[1][]{%
  enhanced, breakable,
  colback=kcream, colframe=kindigo,
  boxrule=0.5pt, arc=2pt, outer arc=2pt,
  left=10pt, right=10pt, top=8pt, bottom=8pt,
  fonttitle=\small\bfseries,
  coltitle=white,
  colbacktitle=kindigo,
  attach boxed title to top left={xshift=10pt, yshift=-2.5mm},
  boxed title style={
    colframe=kindigo,
    colback=kindigo,
    boxrule=0.4pt,
    arc=2pt,
    outer arc=2pt,
    left=4pt, right=4pt, top=2pt, bottom=2pt,
  },
  fontupper=\small,
  title={#1}%
}

\title{\vspace{-1.0em}
  {\color{kindigo}\Large\bfseries Symmetric Structure and Motivic Transformation in Melodies, Phrases, and Themes}
}

\author{%
  Alfredo Gonz\'alez-Espinoza\,\orcidlink{0000-0003-2361-1827}\\[2pt]
  \small{\color[HTML]{C41230}\faUniversity}~University Libraries, {\color[HTML]{C41230}C}arnegie {\color[HTML]{C41230}M}ellon {\color[HTML]{C41230}U}niversity\\
  \small{\color{kterra}\faEnvelope[regular]}~\texttt{\href{mailto:[email protected]}{agonzal3@andrew.cmu.edu}}\\
  \small{\color[HTML]{A6CE39}\faOrcid}~\href{https://orcid.org/0000-0003-2361-1827}{0000-0003-2361-1827}
}

\date{\small\itshape Preprint, 29 October 2026}

\ifrsos
  \usepackage{lineno}
  \titleformat{\section}{\normalfont\Large\bfseries}{\thesection}{0.6em}{}
  \titleformat{\subsection}{\normalfont\large\bfseries}{\thesubsection}{0.5em}{}
  \titleformat{\subsubsection}{\normalfont\normalsize\bfseries}{\thesubsubsection}{0.5em}{}
\author{%
    Alfredo Gonz\'alez-Espinoza\\
    \small University Libraries, Carnegie Mellon University\\
    \small \texttt{\href{mailto:[email protected]}{agonzal3@andrew.cmu.edu}}\\
    \small ORCID \href{https://orcid.org/0000-0003-2361-1827}{0000-0003-2361-1827}
  }
  \title{Symmetric Structure and Motivic Transformation in Melodies, Phrases, and Themes}
  \date{}
  \ifsionly
  \else
  \fi
\fi

%% file: content.tex
\maketitle
\ifrsos\linenumbers\fi

\noindent\rule{\linewidth}{0.4pt}\par\smallskip

\noindent\textbf{Abstract.}~Transposition, inversion, and retrograde are fundamental operations in music theory, from
Schoenberg's twelve-tone technique to Forte's set classes and Lewin's transformational theory, yet
their use in melodies has rarely been quantified. We read a melodic line through the orbits and
stabilizers of its windows of $m = 3$--$6$ notes under the group generated by these operations and
cyclic shift. The window entropy splits into an orbit entropy (motivic vocabulary) and a
within-orbit entropy normalized by orbit size (deployment), and the group elements relating windows
of one orbit give the motivic transformations. In $852$ keyboard pieces spanning four centuries, at
least $97\%$ of melodies have lower orbit entropy than surrogates keeping their pitch content,
spectrum, or one-step transitions ($m \ge 4$). Periods differ in vocabulary and symmetry types, forms
in deployment. Recurring motifs are transposed $1.6$ times more often than in random permutations of their
forms, by a fifth or fourth $2.8$ times more often, a preference that weakens from the
Baroque to the 20th century and from Bach's fugues to Shostakovich's. Folk phrases and instrumental
themes are rarely symmetric ($2$--$10\%$) but $1.6$--$5$ times more often than note rearrangements.
The framework extends to any symbolic sequence with a symmetry group.

\par\smallskip\noindent\rule{\linewidth}{0.4pt}

\vspace{0.3em}

\section{Introduction}
\label{sec:intro}
\label{sec:background}

Symmetry is one of the oldest ways we make sense of structure. We recognize that an object can
be transformed (reflected, rotated, shifted) and still be, in the ways that {\em matter}, the
same object. An animal's body is almost mirror-symmetric (invariant under reflection from its bilateral symmetry), and a snowflake repeats every $60^\circ$ because the ice beneath it is a hexagonal crystal.
In physics, symmetry is also a principle that constrains the equations themselves. We identify it in conservation laws: equations that keep their form under {\em shifts} or translations in time and space or under rotations conserve
energy, momentum, and angular momentum \cite{noether1918invariante}, and the symmetries of a law
restrict the phenomena it can produce \cite{curie1894symetrie,wigner1967symmetries,gross1996role}.
Even the concept of equilibrium can be read as a symmetry, a symmetry of time. The microscopic laws of motion
run equally well backward, so at equilibrium every process occurs as often as its reverse, a
molecule forming as often as it breaks apart, and nothing changes on average (or statistically)
\cite{onsager1931reciprocal}. In this case, the concept of time acquires a direction, {\em an arrow}, only when this balance is
broken or moved away from equilibrium. Group theory gives these observations a common language
\cite{weyl1952symmetry}. Taken together, the operations that leave an object unchanged form a
group, and organizing objects by their group arranges them by structure, and not necessarily by content. For example, chemistry
classifies molecules in this way, by their point group, independently of the atoms they contain.
Water and sulfur dioxide, compositionally different, fall in the same class ($C_{2v}$) and therefore share the selection rules for infrared and Raman
activity, the pattern of allowed electronic transitions, and the orbital degeneracies, all defined
by the group alone \cite{cotton1990group}.

Much of this recognition is related to our perceptual capacities. People detect mirror symmetry in an image rapidly, and more
readily than repetition or rotation, and tend to judge symmetric faces and patterns as {\em more
attractive} \cite{wagemans1997symmetry,treder2010symmetry,rhodes1998facial}, a bias shared by
pollinators that prefer symmetric flowers \cite{moller1995bumblebee}. It is no surprise that we
include signatures of symmetry in what we make. The tilings of the Alhambra and the prints of Escher
generate endless surface from a handful of operations, and every pattern that tiles the plane by
repeating in two independent directions falls into one of seventeen symmetry classes,
the wallpaper groups \cite{grunbaum1987tilings}. How often each class is used or appears in a collection of works can give clues of where it comes from, as distributions over the seventeen groups separate Hellenistic,
Roman, Armenian, Byzantine, Arab, and Seljuk ornament \cite{erbudak2019sanmarco,erbudak2020ornaments},
in a long tradition of classifying cultural patterns by symmetry \cite{washburn1988symmetries}.
Music is a peculiar case, since it can apply symmetry in two dimensions at once as well, but these dimensions are of different nature, pitch and time
\cite{hart2009musicalplane}. Hearing can feel less direct than seeing; listeners recognize a
transposed melody at once but recognize its inversion or retrograde poorly
\cite{dowling1972recognition}, so musical symmetry could be a matter of construction as well as
 perception. The time reversal concept of physics can be read in musical melodies as retrograde, and its
breaking, the arrow of time, has been previously measured in scores with network theory methods~\cite{Gonzalez-Espinoza2020}, finding the existence of statistical irreversibility (time asymmetry) as a universal feature in melodies. Like a
molecule, a melody can carry a symmetry internally. A phrase that reads the same backward, or a
motif that mirrors about a pitch axis, is left unchanged by an operation, and motifs can be
classified by their symmetry independently of their notes, as molecules are by their point group
(Figure~\ref{fig:stabilizer-analogy}). The distribution of these symmetry types across a piece is
the musical counterpart of the ornament fingerprint, with the symmetry classes of short motifs instead of plane groups.

\begin{figure}[h]
\centering
\begin{tikzpicture}[
    >=stealth,
    atom/.style={circle, draw=kindigo, line width=0.6pt,
                 minimum size=0.7cm, font=\small\bfseries, inner sep=0pt},
    O/.style={atom, fill=kterra!20, text=kterra},
    H/.style={atom, fill=kindigo!10, text=kindigo},
    bond/.style={line width=1.2pt, kindigo},
    axis/.style={dashed, kochre, line width=0.7pt},
    note/.style={circle, fill=kindigo, minimum size=0.20cm, inner sep=0pt},
    contour/.style={kindigo, line width=0.9pt},
    panellabel/.style={font=\small\bfseries, kindigo},
    grouplabel/.style={font=\small\itshape, kterra},
    orderlabel/.style={font=\footnotesize, kgrey}
]

\begin{scope}[shift={(-4.4, 1.0)}]
  \draw[axis] (0, -1.3) -- (0, 1.3);
  \node[font=\footnotesize, kochre] at (0.28, 1.25) {$\sigma$};
  \draw[->, kochre, line width=0.6pt] (24:1.22) arc (24:-24:1.22);
  \node[font=\footnotesize, kochre, right] at (0:1.3) {$60^\circ$};

  \foreach \a in {90,150,...,450} {
    \draw[bond] (0,0) -- (\a:1.0);
    \draw[bond, line width=0.9pt] (\a:0.5) -- ++(\a+45:0.28);
    \draw[bond, line width=0.9pt] (\a:0.5) -- ++(\a-45:0.28);
    \draw[bond, line width=0.9pt] (\a:0.78) -- ++(\a+45:0.16);
    \draw[bond, line width=0.9pt] (\a:0.78) -- ++(\a-45:0.16);
  }
  \draw[kindigo, line width=0.6pt, fill=kindigo!10]
    (90:0.2) -- (150:0.2) -- (210:0.2) -- (270:0.2) -- (330:0.2) -- (30:0.2) -- cycle;

  \node[panellabel] at (0, -1.7)  {snowflake};
  \node[grouplabel] at (0, -2.15) {Stab $= D_{6h}$};
  \node[orderlabel] at (0, -2.55) {24 operations};
\end{scope}

\begin{scope}[shift={(0,0)}]
  \draw[axis] (0, -0.3) -- (0, 2.3);
  \node[font=\footnotesize, kochre] at (0.28, 2.2) {$C_2$};

  \node[O] (O)  at (0, 1.4) {O};
  \node[H] (H1) at (-0.79, 0.79) {H};
  \node[H] (H2) at ( 0.79, 0.79) {H};

  \draw[bond] (O) -- (H1);
  \draw[bond] (O) -- (H2);

  \node[panellabel] at (0, -0.7)  {\ce{H2O}};
  \node[grouplabel] at (0, -1.15) {Stab $= C_{2v}$};
  \node[orderlabel] at (0, -1.55) {4 operations};
\end{scope}

\begin{scope}[shift={(3.2, 0)}]
  \draw[axis] (1.5, -0.3) -- (1.5, 2.3);
  \node[font=\footnotesize, kochre] at (1.78, 2.2) {$F$};

  \draw[->, kgrey!60, line width=0.4pt] (-0.1, 0) -- (3.2, 0)
        node[right, font=\footnotesize, kgrey] {time};
  \draw[->, kgrey!60, line width=0.4pt] (0, -0.1) -- (0, 2.0)
        node[above, font=\footnotesize, kgrey] {pitch};

  \node[note] (n1) at (0.3, 0.4) {};
  \node[note] (n2) at (0.9, 1.0) {};
  \node[note] (n3) at (1.5, 1.5) {};
  \node[note] (n4) at (2.1, 1.0) {};
  \node[note] (n5) at (2.7, 0.4) {};

  \draw[contour] (n1) -- (n2) -- (n3) -- (n4) -- (n5);

  \node[panellabel] at (1.5, -0.7)  {$(60,\,64,\,67,\,64,\,60)$};
  \node[grouplabel] at (1.5, -1.15) {Stab $= \langle F \rangle$};
  \node[orderlabel] at (1.5, -1.55) {2 operations};
\end{scope}

\end{tikzpicture}
\caption{Internal symmetry classifies an object independently of its
identity. \textbf{Left:} a snowflake is left unchanged by rotations
in steps of $60^\circ$ and by reflections (one mirror $\sigma$ shown),
24 operations in all, its point group $D_{6h}$ inherited from the
hexagonal ice lattice. \textbf{Middle:} the water molecule \ce{H2O} is left unchanged
by a rotation (the $C_2$ axis shown) and two reflection planes,
its point group $C_{2v}$. This classification holds whatever atoms
occupy the vertices. \textbf{Right:} the motif $(60, 64, 67, 64, 60)$ (MIDI notation)
reads the same forward and backward, so it is left unchanged by
retrograde ($F$) and belongs to the musical counterpart of that
symmetry class. The set of operations fixing
an object (its \emph{stabilizer}) sorts objects by internal
symmetry, and objects sharing a class share a structural profile:
spectroscopic selection rules in chemistry, transformation invariances
in music. Within one group, the more operations fix an object, the fewer
distinct forms it has (the orbit--stabilizer theorem).}
\label{fig:stabilizer-analogy}
\end{figure}
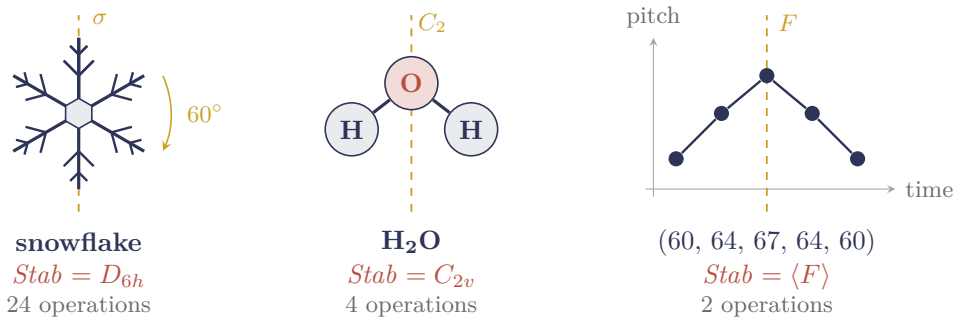

Musicians and music theorists have long borrowed mathematical tools and methods, both to generate music, from the
eighteenth-century \emph{Musikalisches W\"urfelspiel} (musical dice game) to the first computer-composed score, the
\emph{Illiac Suite} \cite{hiller1959experimental}, and Koenig's \textsc{project}~1
\cite{koenig1970project1}, and to analyze it. A {\em motivic idea} is rarely stated only once in a melody or a piece. It
can be restated a step higher, turned upside down, played back to front, or displaced in time, so
that we recognize ``the same idea'' through a {\em disguise} of transposition, inversion,
retrograde, and shift. Two sets of notes that look different on the page can be one motif seen
from another angle, which makes melodic material awkward to compare directly. In music theory
these operations mostly relate distinct objects, and the relation is drawn either inside one
melodic sequence or between sequences. Inside a sequence, developing variation and thematic
transformation grow material from a basic shape \cite{schoenberg1950style,reti1951thematic}, and
Schoenberg also formalized transposition, inversion, and retrograde as the operations of
twelve-tone technique. Between sequences, the operations relate voices, in imitation, canon, and
the subject and answer of a fugue \cite{mann1958study}, or pieces, in tune families
\cite{bayard1950}, in Forte's set classes, which gather the transposed and inverted forms of a chord
\cite{forte1973structure}, and in Lewin's transformational theory, which recasts the operations
as group actions on a generalized space, arrows between objects, and takes the interval between
two objects to be the operation that carries one onto the other \cite{lewin1987gmit}. Mathematical music
theory offers a categorical synthesis of these views \cite{mazzola2002topos}. Self-symmetry, an
object left unchanged by an operation, is the least common documented case, as in Messiaen's modes of
limited transposition \cite{messiaen1944technique} and the palindromes of Machaut and Haydn.
Group theory holds both readings in one object, since the orbit of a motif contains its related
forms and its stabilizer the operations that leave it unchanged. To capture both readings from a piece we
need a statistical description of how a melodic line reuses its material, a distribution of
ideas across a piece rather than a single symmetry label for the whole melody. The orbit
statistics ($\Hor$, $\eta$) and the operations that relate recurring windows measure the first reading, and the distribution of symmetry types
the second.

In the literature, the mathematics of these operations is mostly enumerative. Reiner introduced
P\'olya counting to music theory \cite{reiner1985enumeration}, Fripertinger extended it to
intervals, chords, rows, motifs, and mosaics \cite{fripertinger1992enumeration}, and Fripertinger
and Lackner classified the twelve-tone rows under the direct product of two dihedral groups,
reporting for each row its orbit, normal form, and stabilizer class \cite{fripertinger2015tone}.
Their group is the same as ours at $m = 12$, and the stabilizer class they compute exactly for a cataloged
row is the quantity our windowed spectrum estimates statistically for a piece. Close to the
present work, Ibragimova and Nehaniv show that transposition, inversion, retrograde, and cyclic
shift generate the finite group $G = D_{12} \times D_m$, and give closed-form recipes for
constructing and counting the melodies invariant under it \cite{ibragimova2025algebraic}. We
adopt their orbit count, with $m$ the size of a window rather than the length of a melody. Their
treatment enumerates how much symmetry is \emph{possible}. We ask how often the operations
relate the motifs inside real melodies, as a view of how composers develop motivic ideas, and
whether themes and phrases chosen by people are symmetric more often than chance.

Earlier studies have described melodic structure computationally. Some estimate the entropy of
pitch sequences with $n$-gram models \cite{pearce2005construction}, others summarize the shape of
a melody through its contour \cite{schmuckler2010melody} or extract recurring patterns and cells
\cite{conklin2002representation}. These descriptions compare melodic material note by note or
by shape, and none treats two passages as the same idea when they differ by transposition,
inversion, retrograde, or cyclic shift, the four operations at once.

Where those methods compare material note by note, we use symmetry as the organizing concept and borrow a pattern common across the quantitative sciences. Objects are grouped into
equivalence classes under a relevant symmetry, each class is labeled by one canonical
representative, and the distribution over representatives is measured with Shannon's entropy
\cite{shannon1948mathematical}, so that a symmetry that would otherwise inflate the count of
distinct objects makes them equivalent instead. Genomics counts canonical $k$-mers under reverse
complement \cite{wittler2023canonical,marccais2024knonical}, network science counts motifs up to
graph isomorphism \cite{milo2002network,wang2020graphmotif}, cheminformatics uses canonical SMILES
\cite{weininger1989smiles,godden2009shannon}, and crystallography measures the information of
atoms distributed over Wyckoff orbits \cite{krivovichev2012topological}. This last case is the
closest to the one presented here and has been computed across thousands of structures and followed through
geological time \cite{krivovichev2014which,krivovichev2018minerals}. Orbit entropy also has a
history in graph theory, from Rashevsky and Trucco to Mowshowitz
\cite{rashevsky1955life,trucco1956note,mowshowitz1968entropy1,mowshowitz1968entropy4,dehmer2011history}.
There, one static graph under its own automorphism group produces one number, computed over a
single partition of its vertices, and non-isomorphic graphs with the same orbit structure receive
the same value \cite{dehmer2011history}. Here, the action of sliding a window draws a distribution of
sub-objects from one sequence under a group fixed in advance, so two pieces visiting the same
orbits are still distinguished by how often they visit each.

The vocabulary of this normalization comes from statistical mechanics, where many microscopic
configurations collapse onto one macroscopic label. Windows play the role of microstates and
orbits of macrostates, and the analogy supplies the partition-and-count logic and the
equipartition ceiling $\log_2 K(m)$ against which a measured entropy is read. The analogy does not go deeper, since consecutive windows share $m - 1$ notes, so they are autocorrelated draws instead of an
equilibrium ensemble, the orbit count grows too fast with $m$ for a thermodynamic limit, and the
equal probability of the microstates within a macrostate, postulated in physics, is something
that musical practice is expected to violate, a violation that $\eta$ measures. Windows carry no
energy, and the entropy is Shannon's, in bits, an informational measure of how many motivic
ideas a piece draws on and how evenly.

These ideas are drawn from a single melodic line extracted from the music, which we read at two lengths. Motifs of a few notes (from three to six), the scale of short figures and of small cells such as Schoenberg's \emph{Grundgestalt} \cite{schoenberg1950style} and R\'eti's thematic process \cite{reti1951thematic}, and of minimalist cell-based writing, are read statistically, through sliding windows. Only at that scale does one piece repeat enough material to estimate a distribution, since the number of possible patterns grows explosively with length. Phrases and themes, the scale of fugal subjects, twelve-tone rows, and Renaissance points of imitation, are read exactly, as whole sequences that people have already delimited. Relations between sequences (imitation and canon between voices, mirror fugues, Bach's crab canon, tune families between pieces) are out of scope and can be material for future studies.

To measure how melodies reuse their ideas under these operations, we make three contributions. The first is a set of statistical symmetry measures for melodic
sequences on the orbit space of $G = D_{12} \times D_m$, namely the orbit entropy $\Hor$, the
deployment ratio $\eta \in [0, 1]$, and the distribution of symmetry types, together with a census of the operations that relate recurring windows. The second is an
implementation, the \texttt{MelodyExtraction.jl} and \texttt{MelodicSymmetry.jl} packages,
validated against hand-computed examples and published results. The third is a pair of empirical tests,
a windowed analysis of a solo keyboard corpus balanced by period, composer, and length against
three null models and a within-orbit shuffle that keeps its vocabulary, and an exact census of phrases and themes curated by people (MTC-ANN, Essen,
MTD) against rearrangements of each segment's own notes. Together they let us ask the following questions: Do real melodies carry $G$-symmetric structure beyond their pitch content, spectrum
and one-step transitions? Where does this structure vary across periods, forms, or composers
writing in the same form? Which operations relate the recurring motifs of a line, beyond what its vocabulary implies? How often are hand-picked themes and phrases exactly symmetric, compared
with chance, and by which symmetry?

\section{Methods}
\label{sec:framework}

\subsection{Symmetry group, orbits, and stabilizers}
\label{sec:group}

We treat two motifs related by transposition, inversion, retrograde, or cyclic shift as the same {\em musical idea}. Formally, this identification is an equivalence relation induced
by a group action. The orbit of a motif is the set of all sequences that these operations
produce from it, so each orbit is one equivalence class, a single musical idea together with all
its transformed forms. A group of transformations is a set of invertible operations on a set $X$ that contains
the identity $e$ and is closed under composition and inversion. Often a few generators are enough to represent all the possible transformations. One rotation and one reflection generate the $2n$ symmetries of a regular
$n$-gon, the dihedral group $D_n$. We use two dihedral groups in this work, one acting on pitch and one
on time (Figure~\ref{fig:dihedral}).

\begin{figure}[h]
\centering
\begin{minipage}[c]{0.49\linewidth}\centering
{\small\textbf{(a)}}\\[2pt]
\begin{tikzpicture}[scale=1.1, every node/.style={font=\scriptsize}]
  \foreach \i [evaluate=\i as \a using 90 - 30*\i] in {0,...,11}
    \node[circle, draw=kink, line width=0.4pt, fill=white,
          inner sep=0.3pt, minimum size=1.4em] (v\i) at (\a:1.7) {$\i$};
  \draw[kink!35] (v0) -- (v1) -- (v2) -- (v3) -- (v4) -- (v5) -- (v6)
                  -- (v7) -- (v8) -- (v9) -- (v10) -- (v11) -- (v0);
  \draw[->, thick, kindigo] (v0) to[bend left=35]
        node[midway, above right=1pt, kindigo] {$r$} (v1);
  \draw[kterra, dashed, thick] ($(v0)+(0,0.3)$) -- ($(v6)+(0,-0.3)$);
  \node[kterra] at (0, 2.3) {axis of $f$};
  \draw[->, kteal, dotted, thick] (v1) to[bend left=22] (v11);
  \draw[->, kteal, dotted, thick] (v3) to[bend left=14] (v9);
  \draw[->, kteal, dotted, thick] (v5) to[bend left=8]  (v7);
\end{tikzpicture}
\end{minipage}\hfill
\begin{minipage}[c]{0.49\linewidth}\centering
{\small\textbf{(b)}}\\[2pt]
\begin{tikzpicture}[scale=1.2, every node/.style={font=\scriptsize}]
  \foreach \i [evaluate=\i as \a using 90 - 72*(\i-1)] in {1,...,5}
    \node[circle, draw=kink, line width=0.4pt, fill=white,
          inner sep=0.4pt, minimum size=1.5em] (p\i) at (\a:1.7) {$\i$};
  \draw[kink!35] (p1) -- (p2) -- (p3) -- (p4) -- (p5) -- (p1);
  \draw[->, thick, kindigo] (p1) to[bend left=35]
        node[midway, above right=1pt, kindigo] {$R$} (p2);
  \coordinate (mid15) at ($(p1)!0.5!(p5)$);
  \draw[kterra, dashed, thick]
        ($(p3)+(0.3, 0.05)$) -- ($(mid15)+(-0.3, -0.05)$);
  \node[kterra] at (-2.0, 1.3) {axis of $F$};
  \draw[->, kteal, dotted, thick] (p1) to[bend right=30] (p5);
  \draw[->, kteal, dotted, thick] (p2) to[bend right=22] (p4);
\end{tikzpicture}
\end{minipage}
\caption{Dihedral groups used in this work. Solid arcs are rotations, dashed lines reflection axes, dotted arcs the swaps they induce. \textbf{(a)} $D_{12}$ on the chromatic circle: $r$ transposes by a semitone, $f$ inverts about the $0$--$6$ axis, and $|D_{12}| = 24$. \textbf{(b)} $D_m$ at $m = 5$: $R$ shifts positions cyclically, $F$ reverses them (retrograde). For even $m$ the axis passes between vertices, and $|D_m| = 2m$.}
\label{fig:dihedral}
\end{figure}
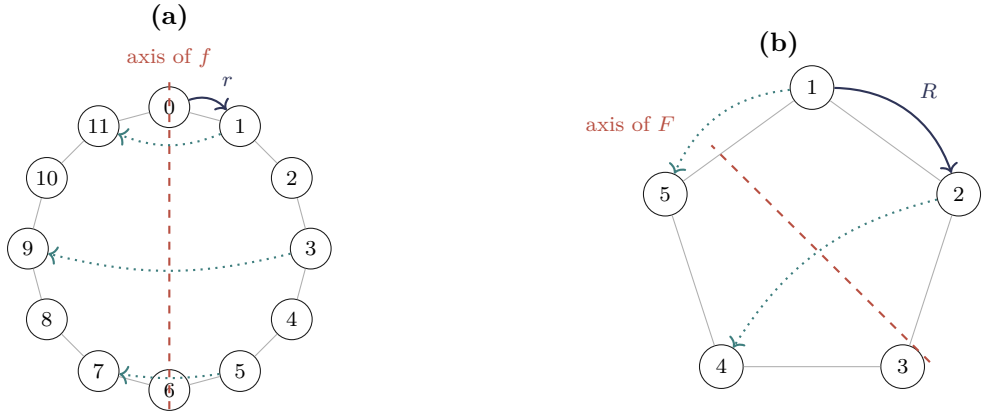

A melodic sequence is a line of $L$ pitch classes, $S = (s_1, \dots, s_L)$, with
$s_t = p_t \bmod 12 \in \Z_{12}$ the pitch class of its $t$-th MIDI pitch and $C = 0$. For a
keyboard piece, $S$ is the line extracted from the score (Section~\ref{sec:extraction}); in the
census, it is a phrase or theme as delimited in its corpus. A sequence $S$ is read at one of two
scales depending on the type of data (piece or segment). Its windows of length $m$, $w_i = (s_i, \dots, s_{i+m-1})$ for $i = 1, \dots, N$ with
$N = L - m + 1$, are elements of $X = \Z_{12}^m$ and give a distribution of motifs (a melodic line in the case of pieces); the whole
sequence is the case $m = L$, a single window, read under $D_{12} \times D_L$ (a single sequence of length $L$, phrase or theme). A window, or any
element of $X$, is written $x = (x_1, \dots, x_m)$.
The tonal group $D_{12}$ acts pointwise on values and the positional group $D_m$ on
indices, with generators
\begin{equation}
\label{eq:generators}
r \cdot p = p + 1, \quad f \cdot p = -p \pmod{12};
\qquad
R : i \mapsto (i \bmod m) + 1, \quad F : i \mapsto m + 1 - i .
\end{equation}
Musically, $r$ transposes by a semitone, $f$ inverts about $C$, $R$ shifts the window
cyclically, and $F$ reverses it (retrograde), with $|D_{12}| = 24$ and $|D_m| = 2m$.

The combined group $G = D_{12} \times D_m$ has elements $g = (a, k, b, j)$ acting as
$f^a r^k$ on values and $F^b R^j$ on positions,
\begin{equation}
(x \ast g)_i = (g_1 \cdot x)_{g_2^{-1}(i)},
\qquad g_1 = f^a r^k, \quad g_2 = F^b R^j,
\end{equation}
so that $|G| = 48m$. The two components act on different coordinates, what is played and
when, and therefore commute (Section~\ref{app:commute}). Conventions for composite
elements are given in Section~\ref{app:conventions}.

\begin{definition}
For $x \in \Z_{12}^m$:
\begin{align}
\Orb(x)  &= \{\, g \cdot x : g \in G \,\} \subseteq \Z_{12}^m, \\
\Stab(x) &= \{\, g \in G : g \cdot x = x \,\} \le G, \\
\hat{o}  &= \min_{g \in G}\, g \cdot x \quad (\text{lex-minimum}).
\end{align}
\end{definition}

Two windows share an orbit if and only if one is a transposition, inversion, retrograde,
or cyclic shift of the other, or a composition of these, so orbits partition
$\Z_{12}^m$ into classes of $G$-related motifs. The stabilizer captures the internal
symmetries of $x$. It always contains $e$, and a palindrome such as $(0, 4, 0)$ is also
fixed by $F$. The orbit--stabilizer theorem, $|\Orb(x)|\,|\Stab(x)| = |G|$, relates the two. A symmetric motif has few distinct surface forms, and an asymmetric one has the maximum, $|G|$. The lex-minimum $\hat o$ labels each orbit uniquely, so two windows are
$G$-equivalent if and only if they share $\hat o$ (Figure~\ref{fig:orbit}; computation in
Section~\ref{app:conventions}).

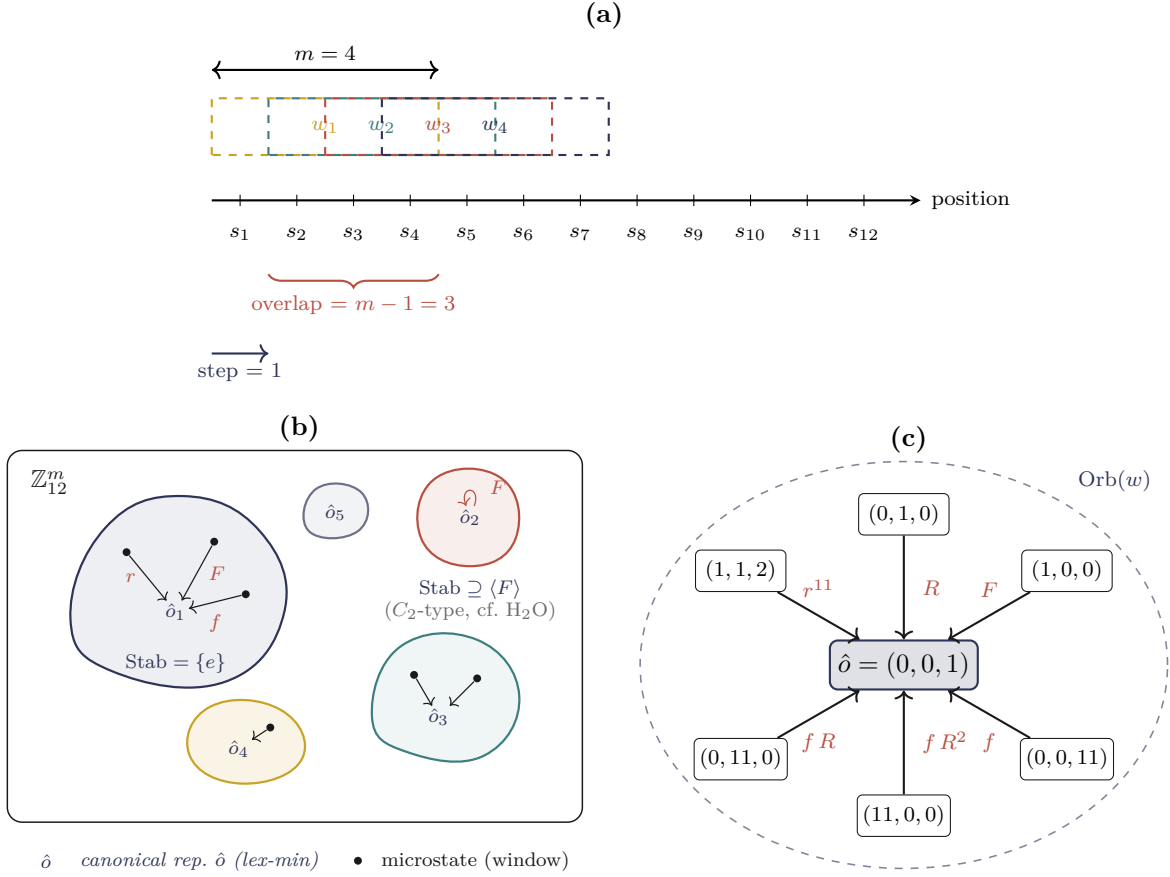
\begin{figure}[h]
\centering
{\small\textbf{(a)}}\\[2pt]
\begin{tikzpicture}[scale=0.75, every node/.style={font=\scriptsize}]
  \draw[thick,-stealth] (-0.5,0) -- (12,0) node[right] {position};
  \foreach \x [evaluate=\x as \i using int(\x+1)] in {0,...,11} {
      \draw (\x,-0.1) -- (\x,0.1);
      \node[below] at (\x,-0.3) {$s_{\i}$};
  }
  \draw[thick, kochre,  dashed] (-0.5,0.8) rectangle (3.5,1.8) node[pos=0.5] {$w_1$};
  \draw[thick, kteal,   dashed] ( 0.5,0.8) rectangle (4.5,1.8) node[pos=0.5] {$w_2$};
  \draw[thick, kterra,  dashed] ( 1.5,0.8) rectangle (5.5,1.8) node[pos=0.5] {$w_3$};
  \draw[thick, kindigo, dashed] ( 2.5,0.8) rectangle (6.5,1.8) node[pos=0.5] {$w_4$};
  \draw[<->,thick] (-0.5,2.3) -- (3.5,2.3);
  \node[above] at (1.5,2.3) {$m = 4$};
  \draw[decorate,decoration={brace,amplitude=5pt,mirror},kterra,thick]
      (0.5,-1.3) -- (3.5,-1.3);
  \node[below,kterra] at (2.0,-1.5) {overlap = $m - 1 = 3$};
  \draw[->,thick,kindigo] (-0.5,-2.7) -- (0.5,-2.7);
  \node[below,kindigo] at (0,-2.7) {step = 1};
\end{tikzpicture}
\\[8pt]
\begin{minipage}[c]{0.49\linewidth}\centering
{\small\textbf{(b)}}\\[2pt]
\resizebox{\linewidth}{!}{%
\begin{tikzpicture}[
    every node/.style={font=\scriptsize},
    ambient/.style={draw=kink, line width=0.6pt, rounded corners=6pt},
    micro/.style={circle, fill=kink, inner sep=0pt, minimum size=3.2pt},
    canon/.style={font=\scriptsize\bfseries, text=kindigo, inner sep=1pt},
    op/.style={kterra, font=\scriptsize\itshape},
    clabel/.style={kindigo, font=\scriptsize\itshape}
  ]
  \draw[ambient] (0,0) rectangle (8.2,5.3);
  \node[kink, anchor=north west, font=\small] at (0.2,5.15) {$\Z_{12}^m$};

  \draw[line width=0.9pt, draw=kindigo, fill=kindigo!8]
    plot[smooth cycle, tension=0.75] coordinates
    {(0.9,2.9) (1.4,4.2) (2.5,4.65) (3.6,4.15) (3.95,3.0) (3.1,2.0) (1.5,2.15)};
  \node[micro] (a1) at (1.7,3.85) {};
  \node[micro] (a2) at (2.95,4.0) {};
  \node[micro] (a3) at (3.4,3.25) {};
  \node[canon] (ao) at (2.4,3.0) {$\hat o_1$};
  \draw[->, kink, line width=0.5pt] (a1) -- (ao) node[op, midway, left=1pt] {$r$};
  \draw[->, kink, line width=0.5pt] (a2) -- (ao) node[op, midway, right=1pt] {$F$};
  \draw[->, kink, line width=0.5pt] (a3) -- (ao) node[op, midway, below=1pt] {$f$};
  \node[clabel, anchor=north] at (2.4,2.55) {$\Stab=\{e\}$};

  \draw[line width=0.9pt, draw=kindigo!70, fill=kindigo!6]
    plot[smooth cycle, tension=0.85] coordinates
    {(4.3,4.15) (4.4,4.75) (5.05,4.7) (5.0,4.15)};
  \node[canon] (eo) at (4.68,4.42) {$\hat o_5$};

  \draw[line width=0.9pt, draw=kterra, fill=kterra!10]
    plot[smooth cycle, tension=0.8] coordinates
    {(5.9,4.05) (6.1,4.9) (7.05,4.9) (7.25,4.05) (6.6,3.65)};
  \node[canon] (bo) at (6.6,4.35) {$\hat o_2$};
  \draw[->, kterra, line width=0.5pt]
    (bo) to[out=70, in=110, looseness=7] (bo);
  \node[op] at (7.0,4.78) {$F$};
  \node[clabel, anchor=north] at (6.6,3.6) {$\Stab\supseteq\langle F\rangle$};
  \node[kgrey, font=\scriptsize, anchor=north] at (6.6,3.28)
    {($C_2$-type, cf.\ \ce{H2O})};

  \draw[line width=0.9pt, draw=kteal, fill=kteal!8]
    plot[smooth cycle, tension=0.75] coordinates
    {(5.2,1.5) (5.6,2.5) (6.7,2.6) (7.3,1.8) (6.8,0.95) (5.7,1.0)};
  \node[micro] (c1) at (5.8,2.1) {};
  \node[micro] (c2) at (6.7,2.05) {};
  \node[canon] (co) at (6.15,1.5) {$\hat o_3$};
  \draw[->, kink, line width=0.5pt] (c1) -- (co);
  \draw[->, kink, line width=0.5pt] (c2) -- (co);

  \draw[line width=0.9pt, draw=kochre, fill=kochre!12]
    plot[smooth cycle, tension=0.8] coordinates
    {(2.6,0.95) (2.9,1.65) (3.9,1.6) (4.2,0.9) (3.35,0.55)};
  \node[micro] (d1) at (3.75,1.35) {};
  \node[canon] (do) at (3.3,1.05) {$\hat o_4$};
  \draw[->, kink, line width=0.5pt] (d1) -- (do);

  \node[canon] at (0.55,-0.55) {$\hat o$};
  \node[clabel, anchor=west] at (0.9,-0.55) {canonical rep.\ $\hat o$ (lex-min)};
  \node[micro] at (5.0,-0.55) {};
  \node[kink, font=\scriptsize, anchor=west] at (5.2,-0.55) {microstate (window)};
\end{tikzpicture}
}
\end{minipage}\hfill
\begin{minipage}[c]{0.49\linewidth}\centering
{\small\textbf{(c)}}\\[2pt]
\resizebox{0.92\linewidth}{!}{%
\begin{tikzpicture}[
    every node/.style={font=\scriptsize},
    micro/.style={rounded corners=2pt, draw=kink, line width=0.4pt,
                  fill=white, inner sep=2pt, minimum width=1.2cm,
                  minimum height=0.55cm},
    canon/.style={rounded corners=3pt, draw=kindigo, line width=0.8pt,
                  fill=kindigo!15, inner sep=3pt, minimum width=1.7cm,
                  minimum height=0.65cm, font=\small\bfseries},
    grp/.style={kterra, font=\scriptsize\itshape}
  ]
  \node[canon] (c) at (0, 0) {$\hat{o} = (0, 0, 1)$};
  \node[micro] (s1) at ( 90:2.0) {$(0, 1, 0)$};
  \node[micro] (s2) at ( 30:2.5) {$(1, 0, 0)$};
  \node[micro] (s3) at (-30:2.5) {$(0, 0, 11)$};
  \node[micro] (s4) at (270:2.0) {$(11, 0, 0)$};
  \node[micro] (s5) at (210:2.5) {$(0, 11, 0)$};
  \node[micro] (s6) at (150:2.5) {$(1, 1, 2)$};

  \draw[->, kink, thick] (s1) -- (c) node[grp, midway, right=3pt] {$R$};
  \draw[->, kink, thick] (s2) -- (c) node[grp, midway, above=3pt] {$F$};
  \draw[->, kink, thick] (s3) -- (c) node[grp, midway, below=3pt] {$f$};
  \draw[->, kink, thick] (s4) -- (c) node[grp, midway, right=3pt] {$f\,R^{2}$};
  \draw[->, kink, thick] (s5) -- (c) node[grp, midway, below=3pt] {$f\,R$};
  \draw[->, kink, thick] (s6) -- (c) node[grp, midway, above=3pt] {$r^{11}$};

  \draw[kindigo!60, dashed, line width=0.6pt]
        (0, 0) ellipse (3.5cm and 2.7cm);
  \node[kindigo, font=\scriptsize\itshape]
        at (2.8, 2.45) {$\Orb(w)$};
\end{tikzpicture}
}
\end{minipage}
\caption{Windows and orbits. \textbf{(a)} A sequence $S$ is read through a sliding window of length $m$ (here $m = 4$), giving $N = L - m + 1$ overlapping windows. \textbf{(b)} $G$ partitions $\Z_{12}^m$ into orbits, each with area proportional to $|\Orb| = |G|/|\Stab|$. Arrows map windows to the canonical representative $\hat o$. Asymmetric motifs ($\Stab = \{e\}$) fill large orbits, while the palindrome $\hat o_2$ is fixed by $F$, as \ce{H2O} is fixed by its $C_2$ axis. \textbf{(c)} One orbit at $m = 3$, each window labeled by the element that maps it to $\hat o = (0,0,1)$.}
\label{fig:orbit}
\end{figure}

The number of orbits for windows of size $m$ is obtained with the Cauchy--Frobenius lemma
\cite{cauchy1845memoire, frobenius1887doppelmodul},
\begin{equation}
\label{eq:orbit-counting}
K(m) \;=\; \frac{1}{|G|} \sum_{g \in G} |\Fix(g)|,
\end{equation}
where $|\Fix(g)|$ has a closed form over the cycles of $g_2$
(Section~\ref{app:decomposition}, following \cite{ibragimova2025algebraic}). This gives
$K(m) = 19$, $163$, $1{,}110$ and $10{,}962$ for $m = 3$--$6$
(Table~\ref{tab:Km}). $K(m)$ is thus the number of distinct motifs of length $m$ once transposition, inversion, retrograde, and cyclic shift are identified.

\subsectionnobarrier{Orbit entropy and deployment ratio}
\label{sec:measures}
\label{sec:orbit-analysis}

With the framework presented in this work, two of the questions that can be asked of the motifs of a melody are: how many distinct musical ideas the melody uses and in how many forms it states each one. The $N$ windows $w_1, \dots, w_N$ of a sequence $S$, read without wrap-around
(Figure~\ref{fig:orbit}a), are each
labeled by its orbit. With $n(o)$ windows in orbit $o$ and $\hat P(o)$ an estimate of its
probability, the orbit entropy is
\begin{equation}
\label{eq:hor-general}
\Hor = -\sum_{o \in X/G} \hat P(o) \log_2 \hat P(o),
\qquad 0 \le \Hor \le \log_2 \min\bigl(N, K(m)\bigr),
\end{equation}
which is $0$ when every window falls in one orbit and reaches the upper bound when windows
spread evenly over orbits (orbital equipartition). At $m = 3$ the positional group
$D_3 \cong S_3$ realizes every reordering of three positions, so orbits of three distinct
pitch classes reduce to the trichord set classes under transposition and inversion. From
$m \ge 4$, orbits retain melodic order and describe shape more than harmony. Both bounds on the window size come from counting. Below three notes every window is symmetric, since any pair of notes is its own
retrograde inversion, and the share of windows in asymmetric orbits rises from $58\%$ at $m = 3$
to $95\%$ at $m = 6$. Beyond six notes nearly every window would be asymmetric, leaving the
symmetry types little to separate, and the orbits ($107{,}509$ at $m = 7$) would far outnumber
the few hundred windows of a typical piece, so that orbit entropy could no longer be estimated
reliably.

\begin{kcallout}[The C major triad]
\noindent $x = (0,4,7)$ is fixed only by $e$, so its orbit holds the maximal
$|G| = 144$ forms at $m = 3$: transposed ($r^5$: F major, $(5,9,0)$), rotated
($R$: $(7,0,4)$), reversed ($F$: $(7,4,0)$) and inverted, which turns major into minor.
Its canonical representative is the C minor triad, $\hat o = (0,3,7) = f r^5 F \cdot x$, so every
major and minor triad, in any key, order, or direction, is one macrostate.
\end{kcallout}

Each window is both an orbit and one surface form within it, so the surface entropy splits
by the chain rule (Section~\ref{app:decomposition}),
\begin{align}
\label{eq:decomp}
\Hmic &= \Hmac + \langle H_{\text{intra-orbit}} \rangle, \\
\langle H_{\text{intra-orbit}} \rangle
 &= -\sum_{o} P(o) \sum_{w \in \Orb(o)} P_{\text{within}}(w \mid o) \log_2 P_{\text{within}}(w \mid o), \nonumber
\end{align}
with $\Hmac = \Hor$ and $P_{\text{within}}(w \mid o) = n(o, w)/n(o)$ the share of orbit
$o$'s windows that take surface form $w$. The first term measures the motivic
\emph{vocabulary}, how many ideas a line draws on, and the second its \emph{deployment},
how many surface forms each idea takes. The same windows give the standard entropies of a symbolic sequence as other readings of one
count. Counted without symmetry they give the block, or $m$-gram, entropy $\Hmic$, counted up to
transposition the entropy of interval $m$-grams, and counted by the rank order of their pitches
the permutation entropy \cite{bandt2002permutation}. Orbit entropy is the reading under the full
group (Section~\ref{app:entropies}).

The forms available to an orbit are bounded by its size. The capacity
\begin{equation}
\label{eq:capacity}
S(\hat o) = \log_2 |\Orb(\hat o)| = \log_2 \frac{|G|}{|\Stab(\hat o)|}
\end{equation}
is the within-orbit entropy reached when every form is used equally. It plays the role of
the symmetry number $\sigma$ in chemistry, which lowers a molecule's entropy by
$k_B \ln \sigma$ \cite{mcquarrie2000statistical}. Symmetric motifs have fewer
distinguishable forms, and asymmetric ones reach the maximum $\log_2 |G|$. The average
$\langle S \rangle = \sum_o P(o) \log_2 |\Orb(o)|$ bounds the intra-orbit term
(Section~\ref{app:decomposition}), and the deployment ratio
\begin{equation}
\label{eq:eta}
\eta = \frac{\langle H_{\text{intra-orbit}} \rangle}{\langle S \rangle} \in [0, 1]
\end{equation}
is the fraction of that capacity a melodic line uses. It tends to $0$ when each idea appears in a single form (rigid quotation) and to $1$ when each is spread evenly over all its transformed forms.

Three idealized regimes separate the two axes (Table~\ref{tab:decomposition}). A strict ostinato uses one form of one idea
($\Hor = 0$, $\langle H_{\text{intra-orbit}} \rangle = 0$), a strict canon many forms of
one idea ($\Hor = 0$, intra-orbit term large), and a line of many ideas, each stated once,
the reverse. An ostinato whose period matches $m$ reaches $\Hor = 0$ because its windows
are cyclic shifts of one cell. A crystal behaves alike in space. Its atoms fall into a few
crystallographic orbits, and its structural complexity is the Shannon information of their
distribution \cite{krivovichev2012topological, krivovichev2016structural}, with the same polarity. In both, a more symmetric object carries less information. Real melodic lines sit between this
crystal limit and equipartition (Figure~\ref{fig:proof-spectrum}a).

\begin{table}[h]
\centering
\caption{Partition of the entropy budget under three compositional
regimes.}
\label{tab:decomposition}
\small
\begin{threeparttable}
\begin{tabular}{lccc}
\toprule
\textbf{Regime} & $\Hmac$ & $\langle H_{\text{intra-orbit}} \rangle$ & $\Hmic$ \\
\midrule
Strict ostinato\tnote{a}    & $0$              & $0$                          & $0$            \\
Strict canon\tnote{b}       & $0$              & up to $S$ of that orbit      & matches intra  \\
Maximum motif variety\tnote{c}   & up to $\log_2 N$ & $0$                          & matches macro  \\
\bottomrule
\end{tabular}
\begin{tablenotes}[para,flushleft]\footnotesize
\item[a] One motive, one surface form. \item[b] One motive, many surface forms. \item[c] Many motives, one surface form each.
\end{tablenotes}
\end{threeparttable}
\end{table}

For synthetic and short annotated lines, which are used to explore known regimes, $\hat P$ is
the plug-in frequency $n(o)/N$, the maximum-likelihood estimate. Corpus lines undersample
the orbit space ($N \ll K(m)$ for $m \gtrsim 6$), so their $\Hor$ is estimated with the
Nemenman--Shafee--Bialek (NSB) estimator \cite{nsb2002}, which corrects small-sample bias
(Section~\ref{app:nsb}). The two estimators are never mixed within a comparison, and
Laplace and Miller--Madow \cite{miller1955note} corrections serve as checks. $\eta$ and its
ingredients use plug-in frequencies throughout.

\subsection{Symmetry types}
\label{sec:stab-types}
\label{sec:genome}

The orbit contains information about which motif a window is, and its stabilizer says what symmetry the motif carries,
e.g. the self-loop of Figure~\ref{fig:orbit}b. Read up to transposition, the kind of operation in
$\Stab(\hat o)$ organizes windows into six classes (Table~\ref{tab:window-types}), asymmetric
(fixed only by $e$), retrograde, retrograde-inversion (RI), inversion, translation, and
multiple, for windows fixed by more than one kind.

\begin{table}[h]
\centering
\caption{Window stabilizer types and the local shape each detects. A
window is symmetric if it is fixed by \emph{some} reflection $\rho$ of the
$D_m$ position cycle (its central axis is only the simplest case) or by a
rotation; conditions hold up to a transposition $k$. The \textbf{Sketch}
column gives an illustrative example of one representative pitch--position contour per type, all drawn
in a common frame (dots are successive notes). By analogy with molecular point groups, the asymmetric class corresponds to $C_1$ and the retrograde class to a $C_{2v}$-type mirror. At $m = 3$ a single axis pairs one couple of positions, so repeated-note and neighbor figures ($[x,x,y]$, $[x,y,x]$) are retrograde-type; at $m = 5$ two couples must agree.}
\label{tab:window-types}
\small
\setlength{\tabcolsep}{5pt}
\newcommand{\stabframe}[1]{%
  \tikz[baseline=1.6mm,x=1.5mm,y=1.5mm]{%
    \useasboundingbox (-0.3,-0.3) rectangle (4.3,4.3);#1}}
\begin{tabular}{l l >{\raggedright\arraybackslash}p{2.0cm} >{\raggedright\arraybackslash}p{3.4cm} c}
\toprule
\textbf{Type} & \textbf{$\Stab$ contains} & \textbf{Condition (some $\rho$)} & \textbf{Local shape detected} & \textbf{Sketch} \\
\midrule
asymmetric   & $\{e\}$              & fixed by no nontrivial element         & directed, non-returning motion &
  \stabframe{\draw[kindigo,line width=0.5pt](0,0)--(1.33,3)--(2.67,1.5)--(4,4);
    \foreach \p in {(0,0),(1.33,3),(2.67,1.5),(4,4)}\fill[kink]\p circle(0.34mm);} \\[3pt]
retrograde   & $\langle F R^j \rangle$  & $x_i = x_{\rho(i)} + k$            & literal return (arch, neighbor, turn) &
  \stabframe{\draw[kindigo,line width=0.5pt](0,0.5)--(2,4)--(4,0.5);
    \foreach \p in {(0,0.5),(2,4),(4,0.5)}\fill[kink]\p circle(0.34mm);} \\[3pt]
RI           & $\langle f F R^j \rangle$ & $x_i + x_{\rho(i)} = \text{const}$ & inversional balance (scalar / mirror) &
  \stabframe{\draw[kindigo,line width=0.5pt](0,0)--(1.33,0.5)--(2.67,3.5)--(4,4);
    \foreach \p in {(0,0),(1.33,0.5),(2.67,3.5),(4,4)}\fill[kink]\p circle(0.34mm);} \\[3pt]
inversion    & $\langle f \rangle$  & $2x_i \equiv \text{const}\ \forall i$  & tritone-only; degenerate, $\approx$ absent &
  \stabframe{\draw[kgrey,dashed,line width=0.5pt](0,0)--(1.33,4)--(2.67,0)--(4,4);
    \foreach \p in {(0,0),(1.33,4),(2.67,0),(4,4)}\fill[kgrey]\p circle(0.34mm);} \\[3pt]
translation  & $\langle R^j \rangle$& $x$ periodic, period $\mid m$          & internal repetition (ostinato cell) &
  \stabframe{\draw[kindigo,line width=0.5pt](0,0)--(0.8,4)--(1.6,2)--(2.4,0)--(3.2,4)--(4,2);
    \foreach \p in {(0,0),(0.8,4),(1.6,2),(2.4,0),(3.2,4),(4,2)}\fill[kink]\p circle(0.34mm);} \\[3pt]
multiple & two or more of the above & several conditions & repeated notes, oscillation &
  \stabframe{\draw[kindigo,line width=0.5pt](0,1)--(1.33,3)--(2.67,1)--(4,3);
    \foreach \p in {(0,1),(1.33,3),(2.67,1),(4,3)}\fill[kink]\p circle(0.34mm);} \\
\bottomrule
\end{tabular}
\end{table}

Every window that has a constant-step structure or shape, such as a chromatic run, is RI-type at every $m$:
retrograde-inversion reverses the succession of intervals, which leaves a constant succession
unchanged (Section~\ref{app:stab-levels}). The class is the most coarse-grained of four descriptions of a stabilizer; the exact subgroup, its conjugacy class, and its abstract family are compared in
the same appendix.

Counting classes over the $N$ windows of a piece gives its class profile $P(c)$ and the
stabilizer entropy
\begin{equation}
\label{eq:hstab}
H_{\text{stab}} = -\sum_{c} P(c) \log_2 P(c),
\end{equation}
a plug-in estimate over six symbols, well sampled at any $N$ in the corpus. Because the class
of a window is a function of its orbit, $H_{\text{stab}}$ coarse-grains the orbit
distribution. Averaging the profiles within a period, each piece weighted equally, gives that
period's class composition, shown as a Muller plot, with class fractions stacked to one and joined
across periods in historical order, one panel per $m$ (Figure~\ref{fig:genome-periods}).
Because the symmetric share falls with $m$ by counting alone, periods are compared at fixed
$m$ and against the share over all $12^m$ windows (Table~\ref{tab:si-baseline}). The number of distinct stabilizer subgroups each period accumulates is shown in Figure~\ref{fig:genome-periods}b, with the estimators in Section~\ref{app:richness}.

\subsection{Relations between windows}
\label{sec:relations}

The orbit identifies which idea a window states, and its stabilizer what symmetry the idea carries.
A third question is how two statements of one idea are related. If windows $w_p$ and $w_q$
share an orbit, the operations that carry one onto the other form the coset
\begin{equation}
\label{eq:relation}
\{\, h \in G : h \cdot w_p = w_q \,\} = h_0 \Stab(w_p), \qquad h_0 = g_q^{-1} g_p ,
\end{equation}
where $g_p \cdot w_p = \hat o$ brings each window to the canonical form. An asymmetric window
gives one relation, and a symmetric one gives $|\Stab(w_p)|$ equally valid relations, each
counted with weight $1/|\Stab(w_p)|$. A relation $h = (a, k, b, j)$ is typed by its inversion
and retrograde components, which gives five types, repetition ($e$), transposition ($r^k$,
$k \ne 0$), inversion ($f r^k$), retrograde ($F r^k$), and retrograde-inversion ($f F r^k$).
Relations that include a cyclic shift ($j \ne 0$) relate windows only as rotations of each
other. They are reported separately as the rotated share and left out of the types.

Every pair of windows in the same orbit whose starting positions differ by $\Delta = q - p \ge m$
notes is counted, so overlapping windows are never compared. The distance $\Delta$ is measured
along the line and is distinct from the cyclic shift $R^j$, which acts within a window. The relation profile of a piece is the share of each
type among its unrotated relations, its entropy $H_{\text{rel}} \le \log_2 5$ measures how
evenly the types are used, and the intervals $k$ of its transpositions form the transposition
spectrum. Splitting the entropy of the operations into the entropy of the types and the
entropy of the interval within each type separates which kind of operation a line uses from the
intervals it uses it at (Section~\ref{app:relations}). Profiles are taken over pairs of
asymmetric windows, where the relation is unique.
Over all windows, weight passes between types through the stabilizer, and every constant-step
run, which is RI-symmetric (Section~\ref{sec:stab-types}), would count each of its
transpositions also as a retrograde-inversion. The relations of a line are the autocorrelation
of its forms over the group \cite{diaconis1988group}, so $\eta$ fixes how widely they spread
but not which operations they are (Section~\ref{app:relations}).

\subsection{Validation with control sequences}
\label{sec:validation}
\label{sec:res-controls}
\label{sec:assumptions}

Before interpreting results from the corpus, we test the measures against sequences whose structure is known
(Table~\ref{tab:synthetic}). They are a period-3 ostinato, the chromatic scale, one motif in its
retrograde, inversion, transposition, and RI images, and uniform-random pitch ($n = 500$),
with four tonal melodies (Twinkle, Ode to Joy, Fr\`ere Jacques, Happy Birthday) for
placement (Figure~\ref{fig:examples-gallery}). These short lines use the plug-in estimator and are read as regime indicators to test our readings of the measures.

\begin{figure}[h]
\centering
\includegraphics[width=\linewidth]{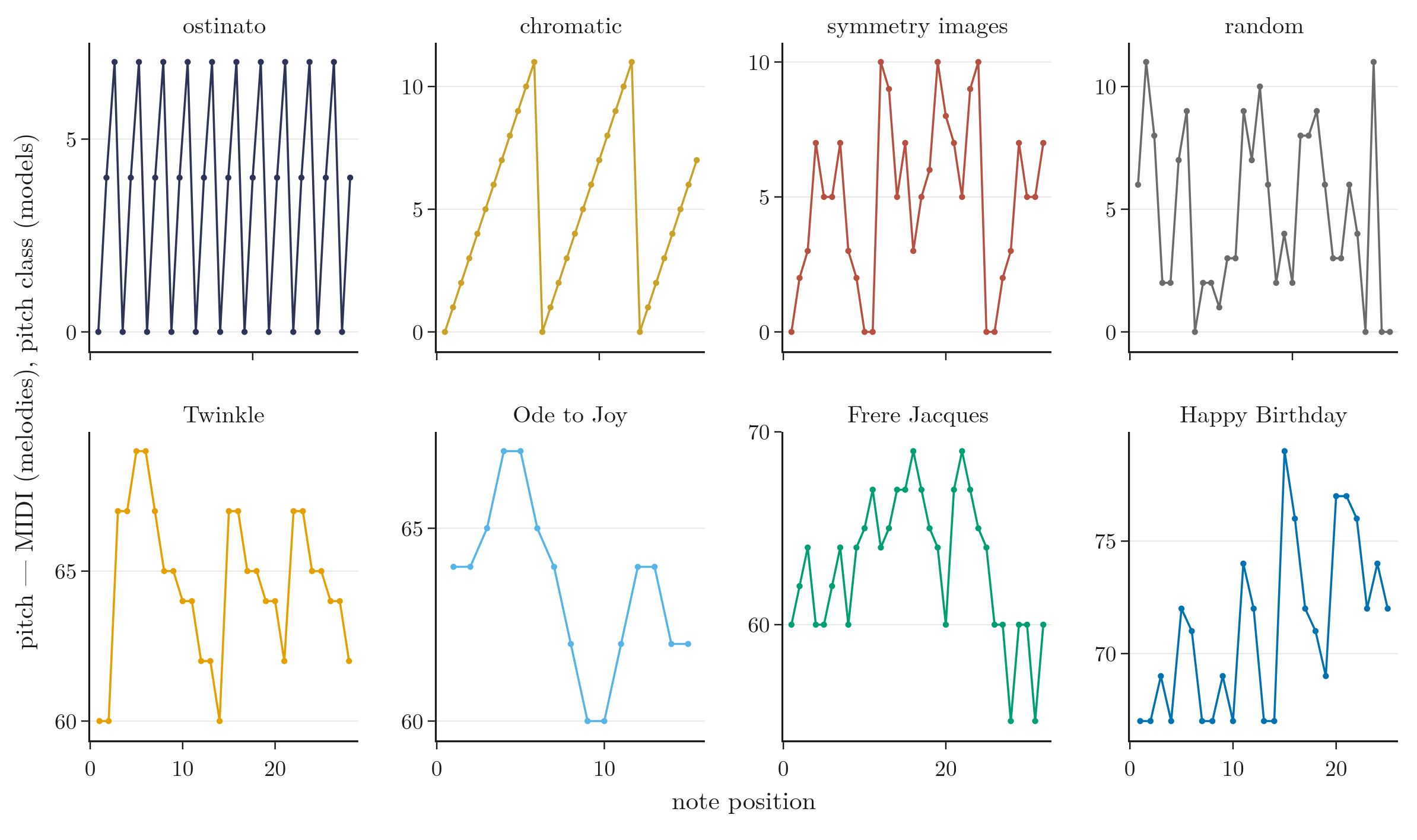}
\caption{The eight example objects as pitch--position contours: four synthetic
controls (ostinato, chromatic, symmetry images, random) and four tonal melodies
(Twinkle, Ode to Joy, Fr\`ere Jacques, Happy Birthday).}
\label{fig:examples-gallery}
\end{figure}

\begin{table}[h]
\centering
\caption{Synthetic controls at $m = 3$ and $m = 5$ (MLE estimator, ceilings
$\log_2 K(3) = 4.25$, $\log_2 K(5) = 10.12$~bits).}
\label{tab:synthetic}
\small
\begin{tabular}{l cc >{\raggedright\arraybackslash}p{4.6cm}}
\toprule
\textbf{Control} & $\Hor$ (m3 / m5) & $\eta$ (m3 / m5) & \textbf{Expected behavior} \\
\midrule
Ostinato \texttt{[0,4,7]}        & $0.00$ / $1.58$ & $0.22$ / $0.00$ & crystal endpoint; window-dependence \\
Chromatic \texttt{0:11}          & $0.00$ / $0.00$ & $0.58$ / $0.52$ & orbit entropy $\neq$ pitch entropy \\
Symmetry images (R,\,I,\,T,\,RI) & $1.99$ / $3.06$ & $0.39$ / $0.20$ & one idea, many surface forms \\
Random \texttt{0:11}, $n{=}500$  & $4.06$ / $8.38$ & $0.69$ / $0.07$ & motivic-vocabulary ceiling \\
\bottomrule
\end{tabular}
\end{table}

The control sequences reflect the predicted endpoints (Figure~\ref{fig:proof-spectrum}; Table~\ref{tab:decomposition}). The ostinato
sits at the crystal limit and random pitch near the ceiling. The chromatic scale returns
$\Hor = 0$ at every $m$ despite having high pitch-class entropy (diversity of notes), because each window is the
previous one transposed by a semitone ($w_{i+1} = r \cdot w_i$) and the line folds onto one
orbit. The
symmetry-images control has a small vocabulary with moderate deployment ($\eta = 0.39$ at
$m = 3$, $0.20$ at $m = 5$), while random pitch has the largest vocabulary, deployed widely at
$m = 3$ ($\eta = 0.69$) but thinly from $m = 5$ ($0.07$), where its windows rarely repeat.

\begin{figure}[tbp]
\centering
{\small\textbf{(a)}}\\[2pt]
\resizebox{0.75\linewidth}{!}{%
\begin{tikzpicture}[scale=0.95, every node/.style={font=\scriptsize}]
  \draw[->, thick, kink] (-0.3, 0) -- (10.6, 0)
        node[right, kink] {$\Hor$};
  \draw[thick, kink] (0, -0.12) -- (0, 0.12);
  \draw[thick, kink] (10, -0.12) -- (10, 0.12);
  \node[below=0.05cm, kink] at (0, -0.15) {$0$};
  \node[below=0.05cm, kink] at (10, -0.15) {$\log_2 K(m)$};
  \node[kindigo, font=\scriptsize\bfseries] at (0, -0.85)
        {crystal / ostinato};
  \node[kindigo, font=\scriptsize\bfseries] at (10, -0.85)
        {equipartition / all motifs explored uniformly};
  \fill[kterra] (3.4, 0) circle (2.5pt);
  \node[above=0.1cm, kterra, font=\scriptsize\itshape] at (3.4, 0.1)
        {real pieces};
  \draw[decorate, decoration={brace, amplitude=4pt}, kteal, thick]
        (3.5, 0.65) -- (10, 0.65);
  \node[above=0.15cm, kteal, font=\scriptsize\itshape] at (6.75, 0.7)
        {motivic compressibility};
\end{tikzpicture}}
\\[6pt]
{\small\textbf{(b)}}\\[2pt]
\includegraphics[width=0.85\linewidth]{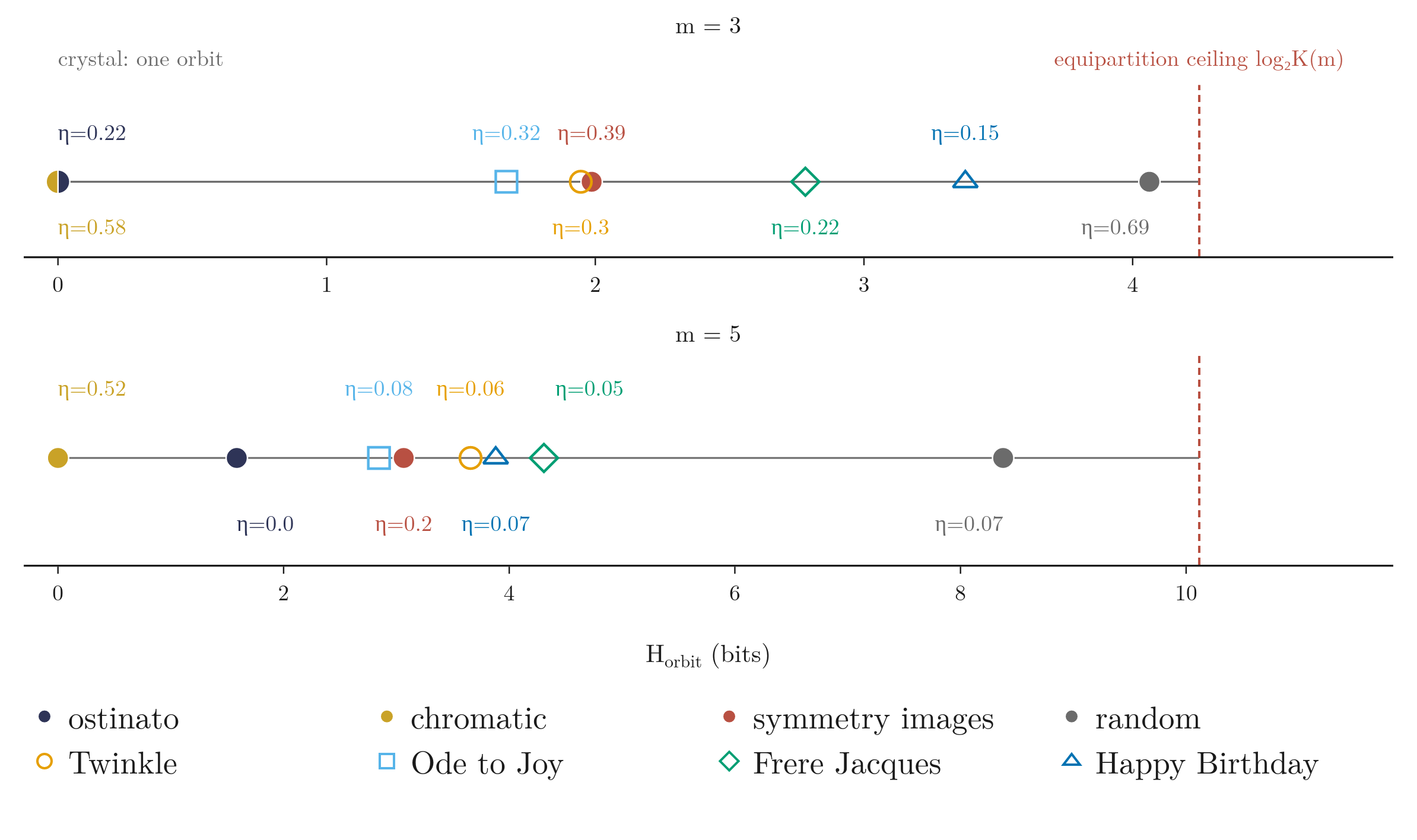}
\caption{\textbf{(a)} Motivic-entropy spectrum: $\Hor \to 0$ when a melodic line stays in one orbit (crystal, ostinato, strict canon) and $\Hor \to \log_2 K(m)$ at orbital equipartition. The gap to the ceiling is the motivic compressibility. \textbf{(b)} Measured $\Hor$: synthetic controls span the interval. Annotated melodies are shown for placement only, since their short length undersamples the distribution.}
\label{fig:proof-spectrum}
\end{figure}
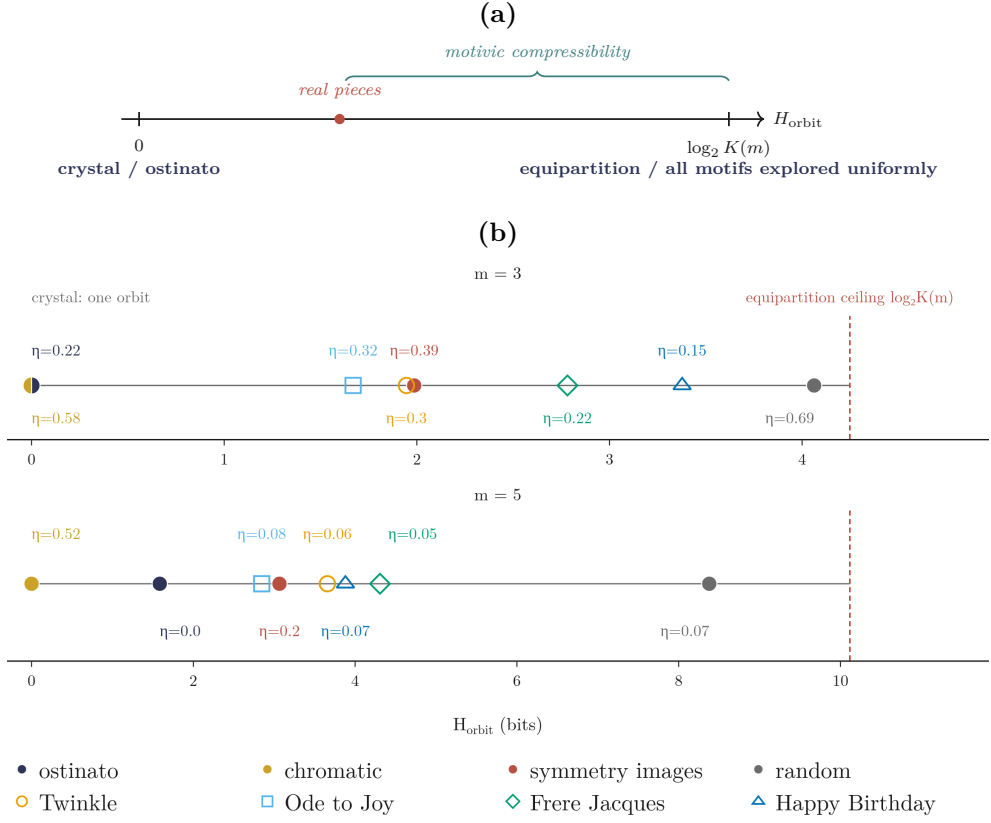

The crystal limit depends on window size (Figure~\ref{fig:proof-sweep}). A periodic line
collapses to one orbit only when $m$ is commensurate with its period $P$. At $m = P$ every
window is a cyclic shift of one cell, while other window sizes straddle cell boundaries, the
orbit-space counterpart of a peak in the autocorrelation at lag $P$. The phase moves the
$\log_2 P$ bits of the cell between the two measures. At $m = P$ they appear as surface
variety within one orbit ($\Hor = 0$, $\eta = 0.22$), otherwise as separate orbits with one
form each ($\Hor = 1.58$, $\eta = 0$). On a corpus read at fixed $m$, a literal ostinato is
therefore flagged by low $\Hor$ or low $\eta$, never by both being high.

\begin{figure}[tbp]
\centering
\includegraphics[width=0.8\linewidth]{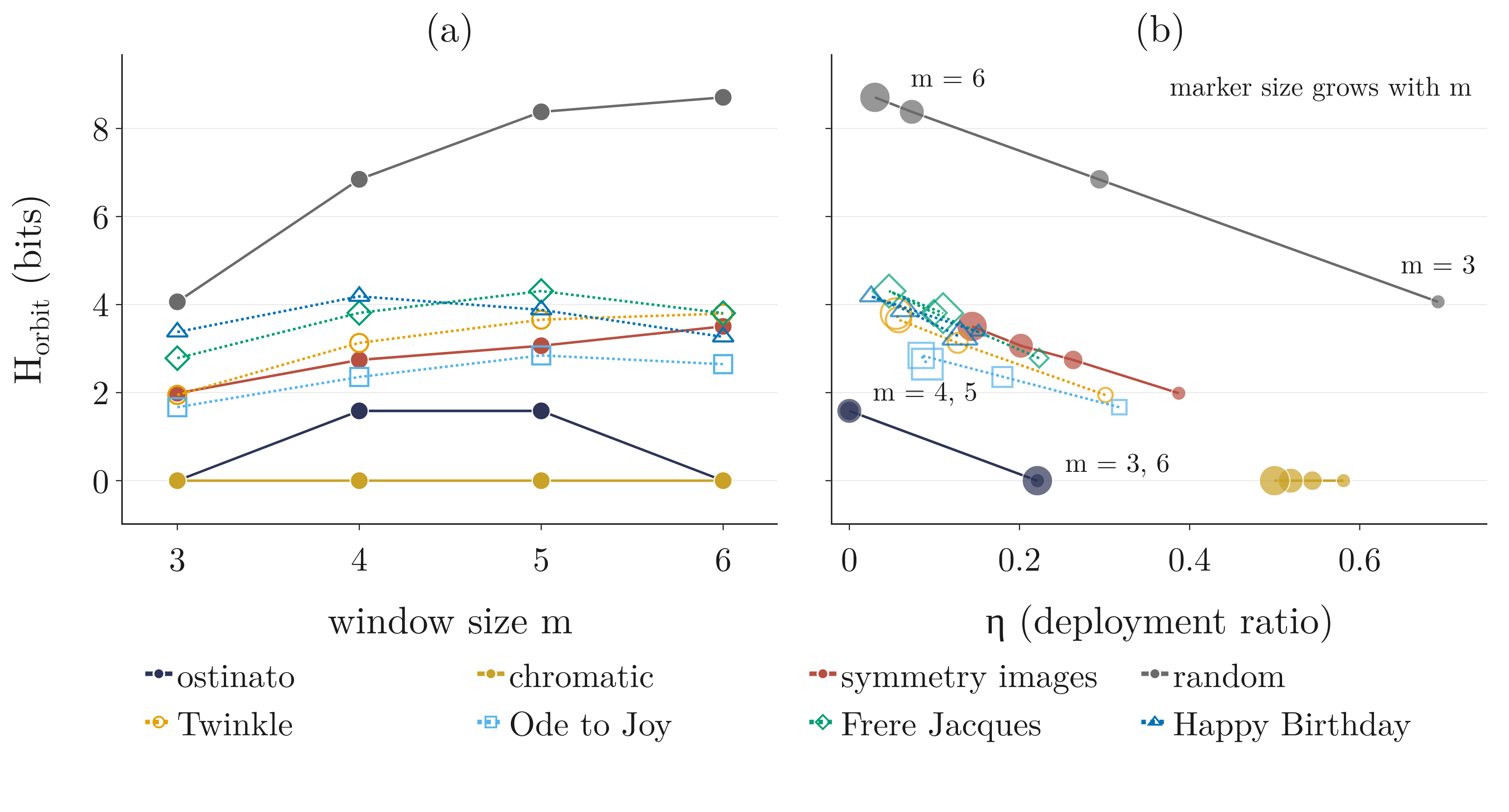}
\caption{\textbf{(a)} $\Hor$ as a function of window size $m$ for the eight examples. \textbf{(b)} The
same values against the deployment ratio $\eta$, one path per example from $m = 3$ to
$m = 6$ (marker size grows with $m$). The period-3 ostinato returns $\Hor = 0$ only when
$m$ is commensurate with its period ($m = 3, 6$). Random pitch tracks the ceiling while
its $\eta$ falls from $0.69$ toward $0$. }
\label{fig:proof-sweep}
\end{figure}

Window profiles are exact even for short lines. At $m = 3$ one matching pair makes a window
retrograde (Table~\ref{tab:window-types}), so retrograde windows dominate tonal tunes
(Twinkle $88\%$); from $m = 5$ return figures (retrograde) separate from scalar motion (RI),
and the chromatic scale is RI at every $m$ (Figure~\ref{fig:si-window}).

\subsectionnobarrier{Data and melody extraction}
\label{sec:data}
\label{sec:pipeline}
\label{sec:extraction}

To test each of the four questions stated in the Introduction, we use four data sets, each matched to one scale of symmetry (Table~\ref{tab:datasets}).
Statistical symmetry (windowed, given by $\Hor$, $\eta$, and symmetry types) is measured on a
balanced corpus of solo keyboard music and on a matched set of preludes and fugues. Exact
symmetry of a whole sequence, under $D_{12}\times D_L$, is measured on three corpora of
segments that people have already delimited as one idea (folk song phrases and
instrumental themes), because a whole piece has a single orbit and offers no distribution. The control sequences
of Section~\ref{sec:validation} complete the set.

\begin{table}[h]
\centering%
\caption{Data sets. $L$: median length of the analyzed sequence $S$ in notes.}
\label{tab:datasets}
\footnotesize\setlength{\tabcolsep}{3pt}
\begin{tabular}{l l r r l}
\toprule
\textbf{Data set} & \textbf{Sequence} & \textbf{Number} & \textbf{$L$} & \textbf{Measures} \\
\midrule
Solo keyboard corpus (KDF) & movement                & 852      & 322 & $\Hor$, $\Delta H$, $\eta$, symmetry types, relations \\
Prelude \& Fugue set       & prelude, fugue          & 108      & 289 & $\eta$, $\Delta\eta$, $\Delta H$, $H_{\text{stab}}$, relations \\
MTC-ANN 2.0.1              & phrase ($4\le L\le16$)  & 1{,}890  & 9   & exact symmetry, $O/E$ \\
Essen                      & phrase ($4\le L\le16$)  & 44{,}765 & 8   & exact symmetry, $O/E$ \\
MTD                        & theme (all $L$)         & 2{,}067  & 18  & exact symmetry, $O/E$ \\
\bottomrule
\end{tabular}
\end{table}

\paragraph{Solo keyboard corpus.} Pieces are movements of solo piano or harpsichord music
from the Kunst der Fuge MIDI collection \cite{kunstderfugue}, a source used in several earlier corpus studies \cite{Gonzalez-Espinoza2017,Gonzalez-Espinoza2020,Park2020,Serra-Peralta2021}.
Composers are placed in one of four periods by the median date of their works, and each
period holds $213$ pieces (Table~\ref{tab:corpus}), from composers with at least $10$ eligible pieces ($4$ to $30$ in the matched sample) and at most two movements per work. Melodies of $100$ to $1198$ notes are matched in
length across periods, since $\eta$ grows with $L$ (Section~\ref{app:length}). The corpus holds $852$ pieces from $822$ works by $49$ composers (Section~\ref{app:corpus-build}).

\begin{table}[h]
\centering%
\caption{Solo keyboard corpus by period (Schubert and Hummel placed in the Classical
period); lengths and composer shares in Table~\ref{tab:si-corpus}.}
\label{tab:corpus}
\small
\begin{tabular}{l c r r}
\toprule
\textbf{Period} & \textbf{Dates} & \textbf{Composers} & \textbf{Pieces} \\
\midrule
Baroque      & before 1750 & 14 & 213 \\
Classical    & 1750--1819  & 10 & 213 \\
Romantic     & 1820--1889  & 14 & 213 \\
20th century & from 1890   & 11 & 213 \\
\bottomrule
\end{tabular}
\end{table}

\paragraph{Prelude \& Fugue set.} To compare composers with form held fixed, we use Bach's
\textit{Well-Tempered Clavier}, Book~I (OpenWTC edition from PDMX \cite{openwtc,long2024pdmx}),
Mendelssohn's Op.~35 and Shostakovich's Op.~87 (Kunst der Fuge), $108$ pieces in which each
prelude and its fugue share a number, key, and mode (Section~\ref{app:pf}).

\paragraph{Folk song phrases and musical themes.} Phrases come from the Essen Folksong
Collection \cite{essen1995} and MTC-ANN \cite{vankranenburg2016mtcann,vankranenburg2019mtcfeatures},
themes from the Musical Theme Dataset \cite{zalkow2020mtd}. The segments were delimited by
people (encoders following text lines, rhyme, or syntax \cite{vankranenburg2014mtc}, editors
selecting complete themes \cite{zalkow2020mtd}), so boundary placement may itself contribute to
their symmetry; the corpora also differ in content (Section~\ref{app:census}).

\paragraph{Melody extraction.} Each keyboard piece is reduced to one melodic line with the
\texttt{:hybrid} method, in which $k$-means on the mean pitch of note groups keeps the highest cluster
\cite{ozcan2005melody}, and the skyline method then takes its top note at each onset
\cite{chai2001melody}. The $k$-means step is seeded from each file's name to make the selection reproducible. A register-based selector can mix voices in contrapuntal writing,
where a subject/theme migrates across registers, so a complexity-based selector
\cite{madsen2007complexity} is used as a double check. Where voices imitate one another the two
selectors produce the same motifs, and $70$ of $74$ tests keep their conclusion with either
(Section~\ref{app:extraction-check}). Lines
are not transposed to a common key, because transposition is already included in the operations of $G$, so every measure is exactly
invariant under a global transposition (local key changes affect at most $m - 1$ windows
each; Section~\ref{app:extraction}).

\subsection{Null models}
\label{sec:nulls}
\label{sec:census-method}
We use null references to isolate structural components of our original sequences to properly evaluate the features that we quantify in them. Each null represents a different conserved structural feature that we want to include or isolate.

\paragraph{Windows.} The orbit entropy value can be interpreted when compared against a null for the same pitch content. Each piece of the keyboard corpus is compared with $199$ surrogates from
each of three nulls that keep progressively more of the line. A random \textbf{shuffle} keeps its pitch distribution, \textbf{IAAFT} \cite{schreiber1996iaaft} also keeps its power spectrum (every linear autocorrelation), and a first-order \textbf{Markov} chain fitted to its
pitch-class transitions reproduces its one-step grammar. If each null removes structure the
real line has, orbit entropies should follow
\begin{equation}
\label{eq:rungs}
\Hor^{\text{real}} \le \Hor^{\text{Markov-1}} \le \Hor^{\text{IAAFT}} \le
\Hor^{\text{shuffle}} \le \log_2 K(m).
\end{equation}
For each null we record $\Delta H = \Hor - \overline{H}_{\text{null}}$, the relative deficit
$\Delta H/\overline{H}_{\text{null}}$, a two-sided rank $p$-value, and the expected share of
each symmetry class. Surrogates are seeded from the file name. In the Prelude \& Fugue set only
IAAFT is used, which also gives $\Delta\eta = \eta - \overline{\eta}_{\text{null}}$.

\paragraph{Whole sequences.} Each phrase or theme, a sequence read whole ($m = L$), is symmetric when its stabilizer under
$D_{12} \times D_L$ is non-trivial. Because the chance of symmetry depends on length, parity and
the notes used, segment $i$ is compared with rearrangements (permutations) of its own notes, and $p_i$ is the probability that one is symmetric, obtained by counting fixed arrangements (rearranging
intervals instead serves as a reference). At each length $L$,
\begin{equation}
\label{eq:oe}
\frac{O_L}{E_L}, \qquad E_L = \sum_{i:\,L_i = L} p_i ,
\end{equation}
with $O_L$ the observed number of symmetric segments and $E_L$ the expected number under random rearrangement of each segment's own notes, so that $O_L/E_L > 1$ means more symmetry
than each segment's notes give by chance. Lengths are pooled into bins for display, with exact
Poisson intervals \cite{garwood1936fiducial}; intervals that resample whole songs or works,
and per-length values smoothed with an empirical-Bayes model \cite{clayton1987empirical}, are
in Sections~\ref{app:stats} and~\ref{app:census}.

\paragraph{Relations.} The construction of the surrogates changes which orbits a line uses, so they cannot tell
whether its relations are chosen or follow from its vocabulary. The within-orbit shuffle keeps the
vocabulary. Each window keeps its position and orbit, and the forms that each orbit takes are
reassigned among its windows at random ($199$ permutations per piece), which leaves $\Hor$,
$\eta$, and $H_{\text{stab}}$ unchanged and alters only which form appears where. Identical
forms stay identical under this null, so repetition is compared instead with forms drawn
uniformly within each orbit and with the three surrogates. This null changes the relations between
windows while holding every quantity on which a vocabulary claim rests, so an excess of
transform-pairs over it is a statement about placement alone. For each type the observed share is
compared with the expected one as $O/E$, per piece and over the corpus with pieces weighted
equally.

\subsection{Statistical analysis}
\label{sec:stats}

We include statistical analysis for the four questions we explore in this work. The first is whether a piece has less orbit entropy than
surrogate lines that keep part of its structure, its pitch distribution, its power spectrum, or its
one-step transitions (Section~\ref{sec:nulls}). Its value is ranked among those of the $199$
surrogates from each null, and because this is tested for every piece, the expected share of
false positives is controlled across pieces \cite{benjamini1995fdr}. The second is whether
periods, forms, or composers differ. Here each piece, not each window, is one observation, and
group labels are shuffled to see how often a difference as large arises by chance. Pieces by
one composer resemble each other, so period labels are shuffled between composers and form
labels within each composer \cite{hurlbert1984}. In the Prelude \& Fugue set the three
composers are compared in historical order \cite{jonckheere1954}, and each prelude with its own
fugue. The third is whether whole segments are symmetric more often than their notes allow,
judged by the ratio of observed to expected counts, with intervals that resample whole songs or
works. Each comparison is reported with an effect size, the probability that a piece from one
group exceeds one from the other \cite{cliff1993dominance}, and corrected across the four window
sizes \cite{holm1979simple}. For relations, per-piece ratios add $0.5$ to observed and expected counts, so that pieces
without a given relation are kept \cite{haldane1956estimation,anscombe1956estimating}, and
composer means are shrunk toward their period mean by an empirical-Bayes model
\cite{morris1983parametric}. Procedures and robustness checks are in Section~\ref{app:stats}.

\section{Results}
\label{sec:results}

The results are structured in three parts. We ask whether melodies carry more orbit structure than
randomized versions of themselves (Section~\ref{sec:res-surrogates}), where that structure varies
across periods, forms, and composers who wrote in the same form (Sections~\ref{sec:res-hor-eta}
and~\ref{sec:res-pf}), and which symmetries it consists of, in short windows, in the relations
between them, and in whole phrases and themes (Sections~\ref{sec:res-genome},
\ref{sec:res-relations}, and~\ref{sec:res-census}). In brief, melodies draw on fewer motif classes
than chance but on more kinds of symmetry, periods differ in vocabulary, forms and composers in
deployment, recurring motifs are transposed more often than chance, preferably by a fifth, and whole
phrases and themes are symmetric more often than their notes allow. Full
results with their statistical procedures are in Sections~\ref{app:stats} and~\ref{app:figs}.

\subsection{Orbit structure against surrogates}
\label{sec:res-surrogates}

Figure~\ref{fig:surrogate-z} compares one typical piece from each period with its surrogates. In every case the real line (black) lies below all three bands, and across the corpus the
orbit entropy falls below the shuffled and IAAFT surrogates in more than $99\%$ of the $852$ pieces
and below the Markov surrogates, which keep the one-step transitions between notes, in $86$ to
$98\%$. The gaps between successive nulls show where this structure lies (Table~\ref{tab:rungs}).
Keeping the power spectrum adds at most $0.08$ bits, so the structure that is captured by $\Hor$ is not necessarily reflected in the power spectrum of the same series, despite the $1/f$-like spectra commonly found in music \cite{voss1975music,levitin2012rhythm}, as the
time irreversibility of pieces from the same collection also indicates \cite{Gonzalez-Espinoza2020}.
One-step transitions explain part of the structure, and what remains grows from $0.12$ bits at $m = 3$ to
$1.91$ bits at $m = 6$, so the structure spans whole motifs rather than neighboring notes.

Deployment departs from chance in a direction that depends on window size. In about $80\%$ of the
pieces, real lines repeat cells of three notes literally, in fewer transformed forms than their
surrogates, but restate motifs of five and six notes in more forms, transposed, inverted, or
reversed. These comparisons hold piece by piece after correction for multiple testing
(Section~\ref{app:stats}; Figure~\ref{fig:si-nulls} and Table~\ref{tab:si-deta}).

\begin{table}[h]
\centering
\caption{Measured gaps between consecutive levels of Eq.~\eqref{eq:rungs}: median over
the 852 pieces, in bits. At large $m$ the first row also contains the sampling ceiling (a
line of $N$ windows cannot exceed $\log_2 N$ bits).}
\label{tab:rungs}
\small
\begin{tabular}{l l r r r r}
\toprule
\textbf{Gap} & \textbf{Structure it reflects} & $m=3$ & $m=4$ & $m=5$ & $m=6$ \\
\midrule
$\log_2 K(m) - \Hor^{\text{shuffle}}$ & pitch bias (key, scale, range) & 0.26 & 0.51 & 1.14 & 3.05 \\
$\Hor^{\text{shuffle}} - \Hor^{\text{IAAFT}}$ & linear correlation (spectrum) & 0.00 & 0.04 & 0.07 & 0.08 \\
$\Hor^{\text{IAAFT}} - \Hor^{\text{Markov-1}}$ & one-step transitions & 0.38 & 0.60 & 0.89 & 0.69 \\
$\Hor^{\text{Markov-1}} - \Hor^{\text{real}}$ & beyond one-step transitions & 0.12 & 0.57 & 1.21 & 1.91 \\
\bottomrule
\end{tabular}
\end{table}

\begin{figure}[h]
\centering
\includegraphics[width=0.75\linewidth]{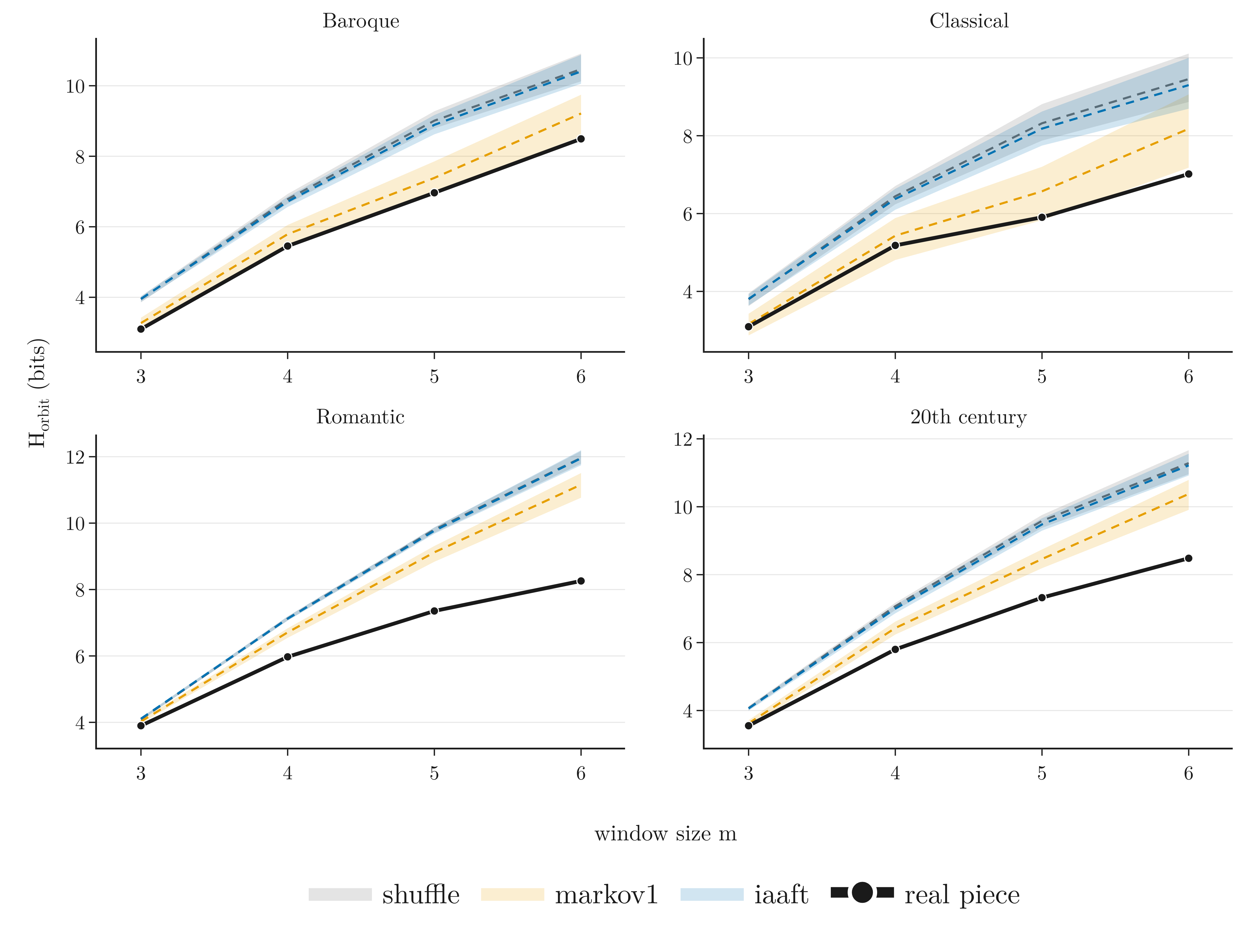}
\caption{Real $\Hor$ (black) against the shuffle (gray), IAAFT (blue), and Markov-1 (yellow) nulls (bands:
central $95\%$ of 199 surrogates; dashed: median) for one typical piece per period, the piece
whose IAAFT deficit at $m = 4$ equals its period's median: Froberger, Capriccio~3; Mozart,
Fugue K.~375g; Liszt, Grand \'Etude~4; Scriabin, \'Etude Op.~8 No.~6.}
\label{fig:surrogate-z}
\end{figure}

\subsectionnobarrier{Vocabulary across periods, deployment across forms}
\label{sec:res-hor-eta}
\begin{figure}[!tbp]
\centering
{\small\textbf{(a)}}\\[2pt]
\includegraphics[width=\linewidth]{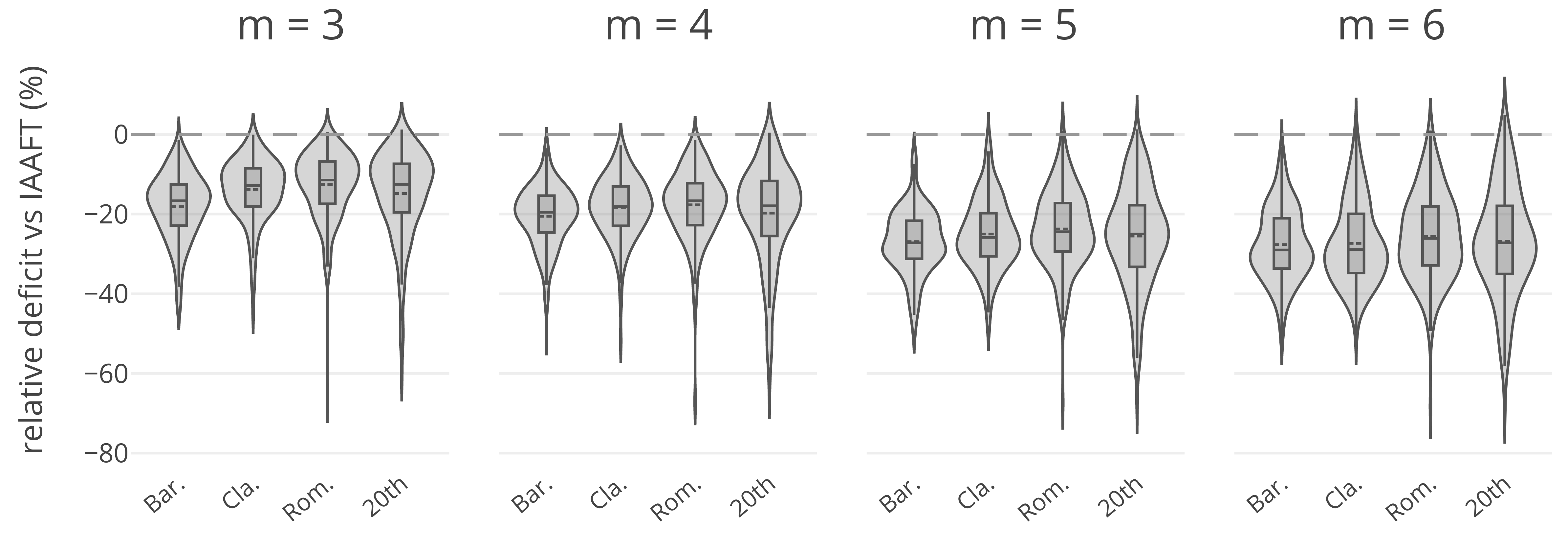}\\[6pt]
{\small\textbf{(b)}}\\[2pt]
\includegraphics[width=\linewidth]{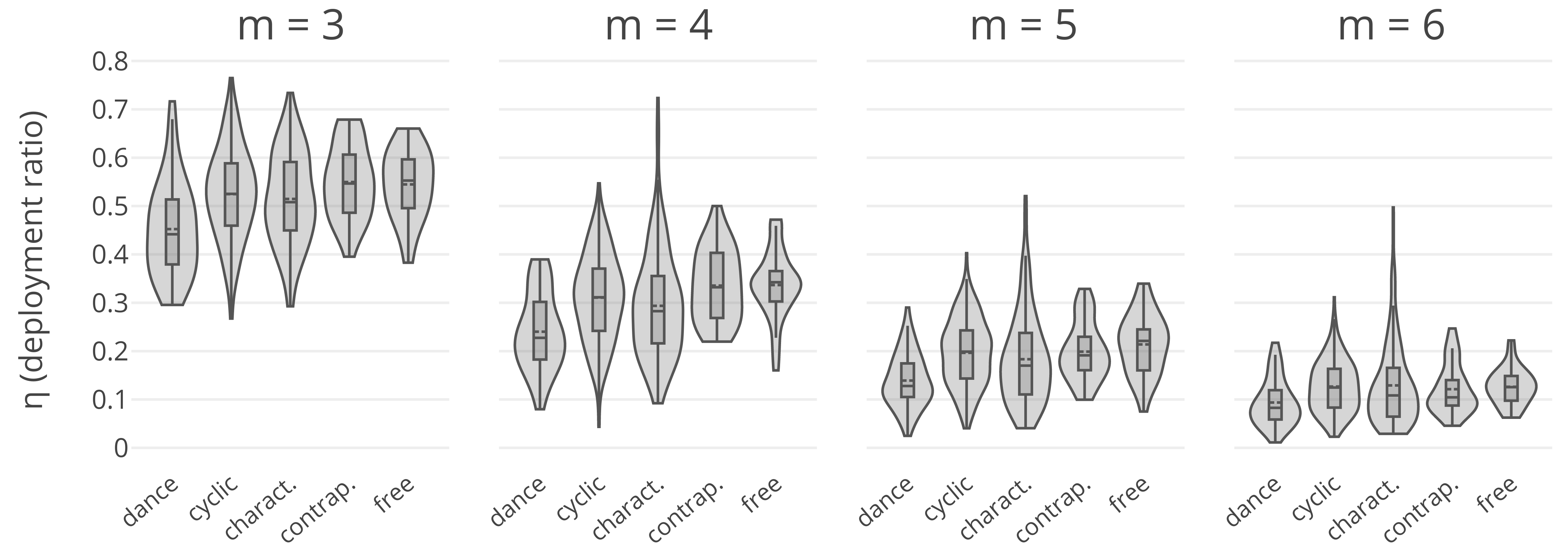}\\[6pt]
{\small\textbf{(c)}}\\[2pt]
\includegraphics[width=\linewidth]{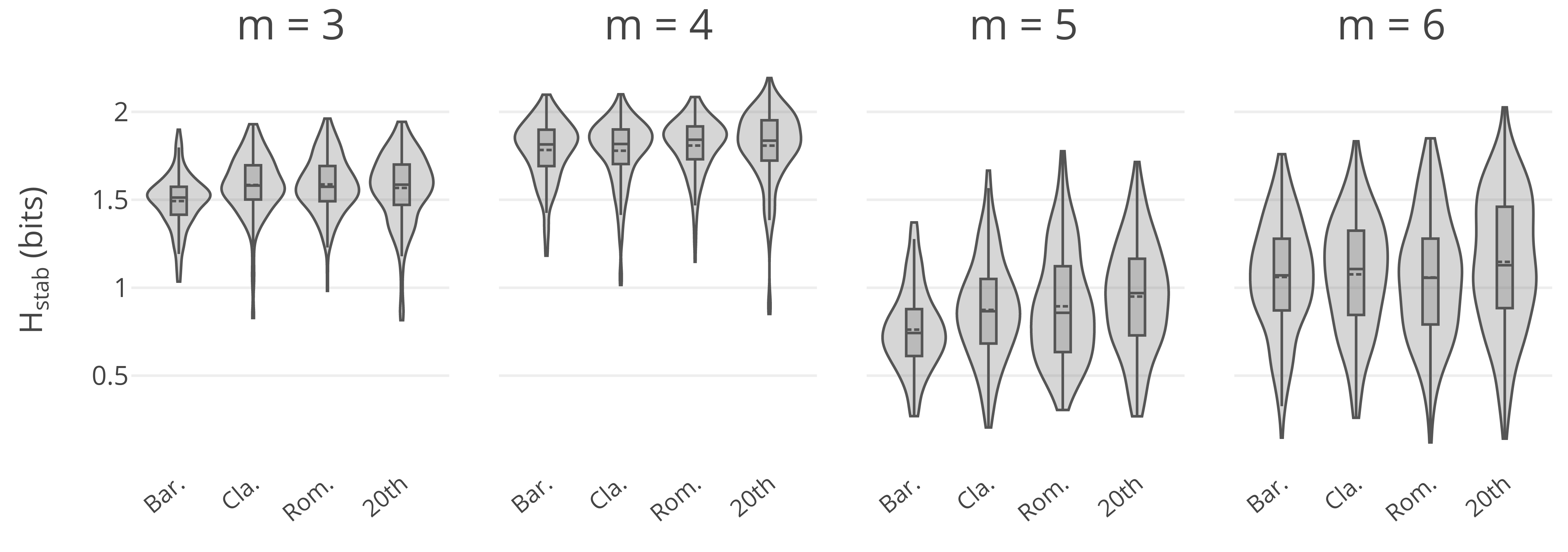}
\caption{The three measures across periods and forms, $m = 3$--$6$. \textbf{(a)} Relative
deficit of $\Hor$ against IAAFT by period (vocabulary). \textbf{(b)} Deployment ratio $\eta$ by
form family (deployment). \textbf{(c)} Symmetry-type entropy $H_{\text{stab}}$ by period (mix of
symmetry types). $213$ pieces per period (Bar.\ Baroque, Cla.\ Classical, Rom.\ Romantic, 20th:
20th century); $508$ labeled pieces by form (dance $65$, cyclic $211$, character $167$,
contrapuntal $35$, free $30$). Character: single-movement pieces (prelude, nocturne, \'etude,
impromptu, ballade, intermezzo); cyclic: sonata, suite, partita movements. Boxes: interquartile
range and median; dashed: mean. Separation tests for every measure and grouping are in
Table~\ref{tab:si-separation}.}
\label{fig:forms-periods}
\end{figure}

Figure~\ref{fig:forms-periods} shows how the three measures (the relative deficit of $\Hor$ against IAAFT, $\eta$, and $H_{\text{stab}}$) change across periods and forms.
Periods differ mostly over short windows. At $m = 3$ Baroque lines are the most structured, about
$17\%$ below their surrogates against $12$ to $13\%$ in later periods, a difference that fades as the
window grows (panel a), while the mix of symmetry types rises slowly from the Baroque to the 20th
century (panel c), carried by repeated notes and oscillating figures (Section~\ref{sec:res-genome}).
Forms differ instead in deployment, since pieces of different forms by the same composer draw on a
similar vocabulary but deploy it differently (panel b). Dances spread each motif over the fewest
transformed forms (median $\eta = 0.23$ at $m = 4$, against $0.28$ to $0.34$ for the other forms), as
befits music built on literal repeats, whereas contrapuntal and free pieces, which restate their
material in other keys, deploy it most widely. Across all three measures the composer (more than the
period or the form) accounts for most of the variation, and these differences hold when the pieces
of each composer are kept together (Section~\ref{app:stats}; Tables~\ref{tab:si-separation}
and~\ref{tab:si-variance}). Because $40\%$ of the pieces carry no form label and forms are unevenly
spread across composers, the form result rests on composers who wrote in several forms
(Section~\ref{app:corpus-build}; Figure~\ref{fig:si-grid}).

\subsectionnobarrier{Deployment in preludes and fugues}
\label{sec:res-pf}
\begin{figure}[!tbp]
\centering
\includegraphics[width=\linewidth]{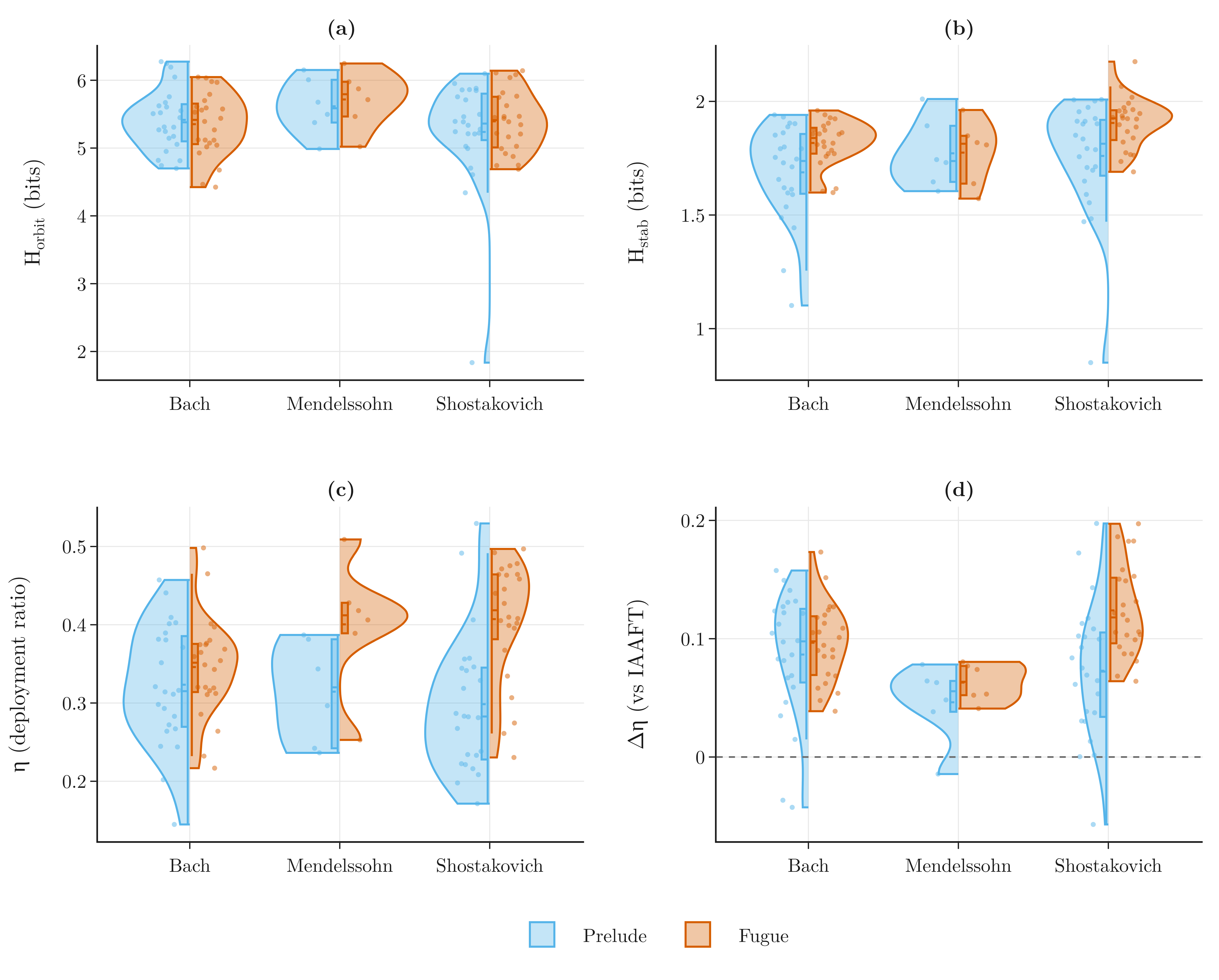}
\caption{Structure and deployment in the Prelude \& Fugue set, $m = 4$.
\textbf{(a)} $\Hor$, \textbf{(b)} $H_{\text{stab}}$ (six symmetry classes), \textbf{(c)} $\eta$ and
\textbf{(d)} $\Delta\eta$ (difference from the mean of 199 IAAFT surrogates), preludes
(left) and fugues (right) by composer ($n = 24$, $6$, $24$ pieces per form).
Jonckheere--Terpstra trend across composers for fugues, Holm-adjusted over $m = 4$--$6$:
$p = 1.0$ (a), $0.014$ (b), $0.002$ (c), $0.035$ (d).}
\label{fig:pf}
\end{figure}

Figure~\ref{fig:pf} compares the preludes and fugues of Bach, Mendelssohn, and Shostakovich, where
key and form are fixed by design. Their fugues draw on a similar vocabulary of motifs (panel a) but
differ in the mix of symmetry types (panel b) and, most clearly, in deployment (panels c and d).
Bach's fugues spread each motif over the fewest transformed forms (median $\eta = 0.35$),
Mendelssohn's and Shostakovich's over more ($0.41$ and $0.42$). The contrast between Bach and
Shostakovich remains when each fugue is compared with its own surrogates, which at this window size
also removes the effect of length, whereas Mendelssohn's high values are largely reproduced by his
surrogates and rest on six fugues. Bach deploys preludes and fugues alike, while Shostakovich deploys
his fugues more widely than his preludes.

\subsection{Symmetry types of short windows}
\label{sec:res-genome}

\begin{figure}[!tbp]
\centering
{\small\textbf{(a)}}\\[1pt]
\includegraphics[width=\linewidth]{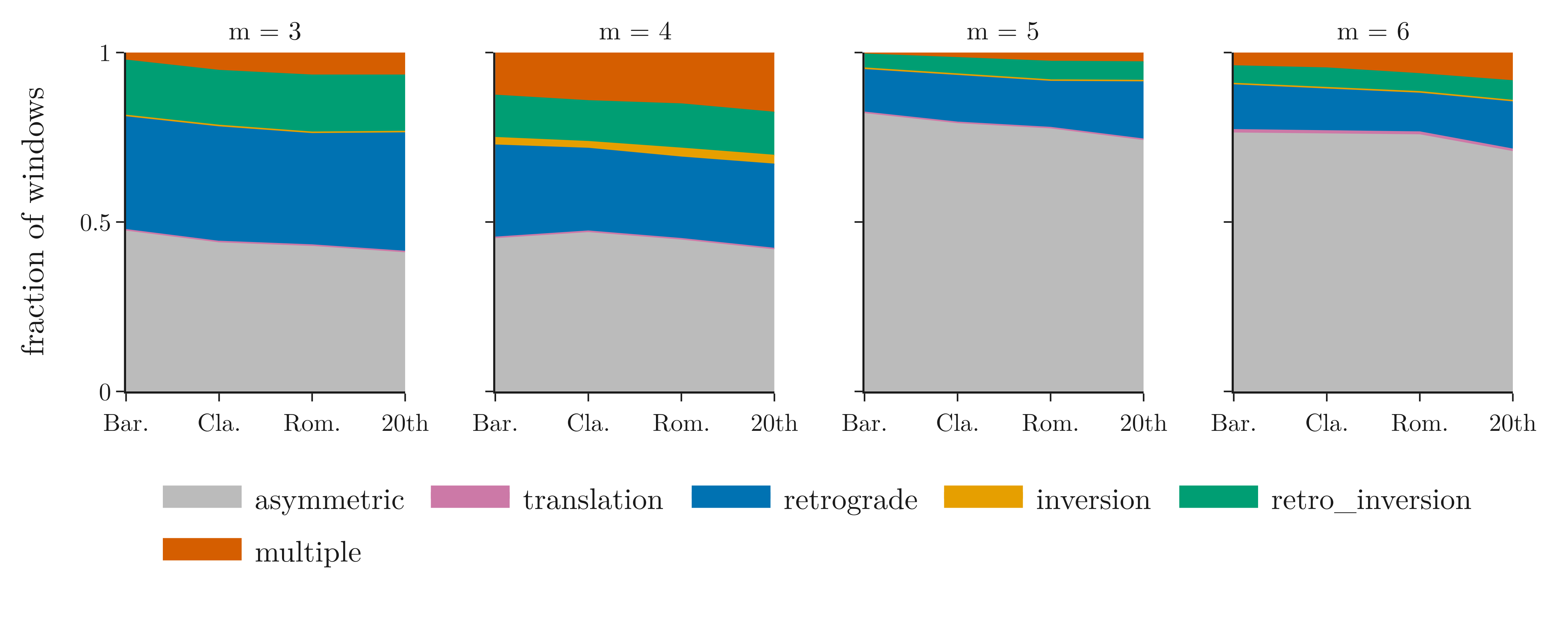}\\[4pt]
{\small\textbf{(b)}}\\[1pt]
\includegraphics[width=\linewidth]{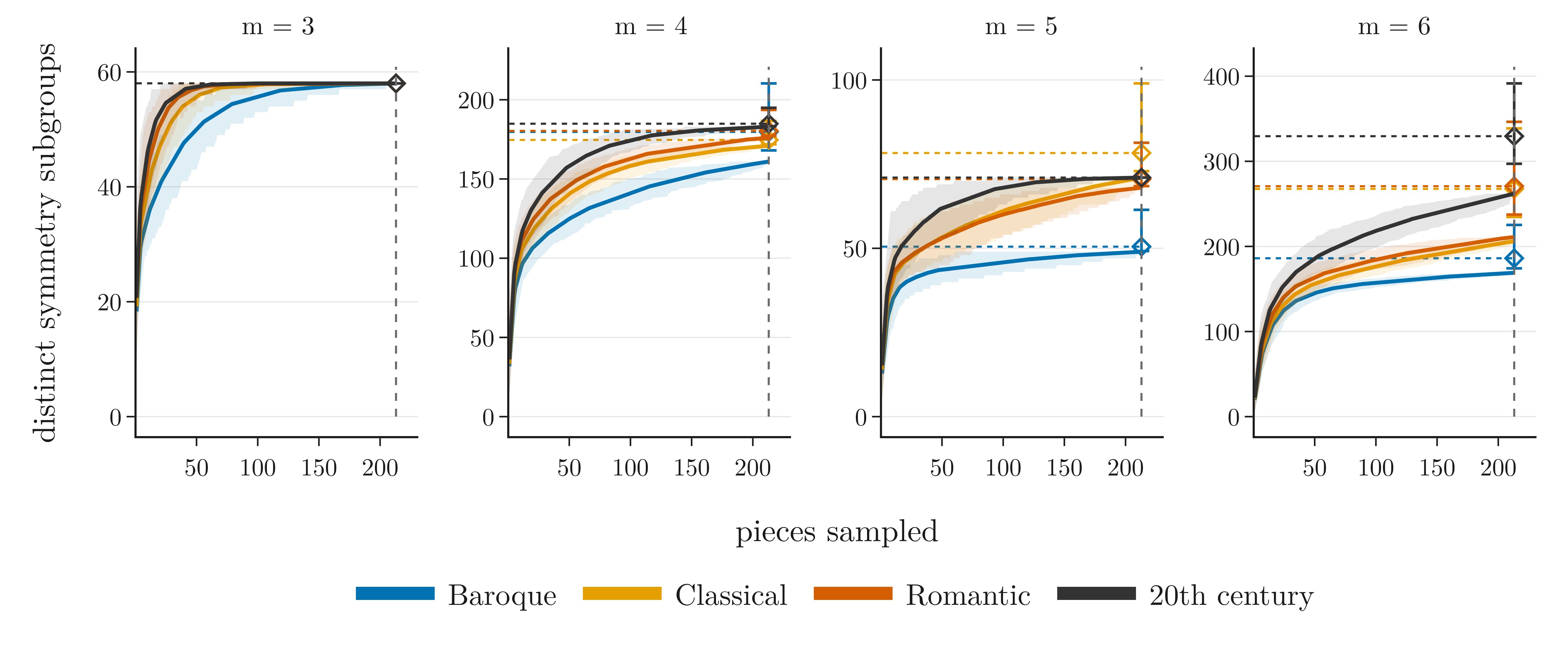}
\caption{Symmetry types by period, $m = 3$--$6$. \textbf{(a)} Share of windows in each symmetry class, computed per piece, averaged within each period, and stacked to one (Bar.\ Baroque, Cla.\ Classical, Rom.\ Romantic, 20th century). \textbf{(b)} Distinct symmetry subgroups found as pieces are added, by period. Lines and bands show the mean and $95\%$ range of the rarefaction curve, diamonds and horizontal dashed lines the estimated total (Chao2) with its interval, and the vertical dashed line the common sample of $213$ pieces.}
\label{fig:genome-periods}
\end{figure}

Most short windows are asymmetric, and more so as windows grow, as counting all $12^m$ windows
predicts. Asymmetric orbits, $7$ of $19$ at $m = 3$, already hold $58\%$ of all windows, a share that
reaches $95\%$ at $m = 6$ (Table~\ref{tab:si-baseline}). Real windows are nevertheless more often
symmetric than in lines that keep their notes in random order, increasingly with window size, and
retrograde windows are about twice as common as in shuffled lines at $m = 6$. Melodic lines that keep only
the one-step transitions reproduce most of this excess, so mirror figures in short windows come
largely from stepwise motion that returns, such as neighbor notes and turns. Windows with several
symmetries at once, mostly repeated notes and two-note oscillations, stay above every null and become
more common over time, from $13\%$ of windows in the Baroque to $18\%$ in the 20th century at $m = 4$
(Figure~\ref{fig:genome-periods}a). Later repertoires use more repetitive figuration or patterns. Windows
symmetric under inversion alone are rarer than chance.

The symmetry-type entropy $H_{\text{stab}}$ summarizes this mix for each piece
(Figure~\ref{fig:forms-periods}c). From $m = 5$ on, real lines carry a more varied mix of symmetry
types than any of their surrogates, the counterpart of their lower orbit entropy, so melodies use
fewer motif classes than chance but more kinds of symmetry among them. This variety grows from the
Baroque to the 20th century at $m = 3$ and $5$ (Tables~\ref{tab:si-classnull},
\ref{tab:si-hstabnull}, and~\ref{tab:si-separation}).

Figure~\ref{fig:genome-periods}b counts how many distinct symmetry subgroups, the exact sets of
operations that fix a window, each period draws on as pieces are added. At $m = 3$ every period
reaches all $58$ possible subgroups. From $m = 4$ the curves separate, with the Baroque lowest at
every window size and the 20th century nearly exhausting those possible at $m = 4$ and $5$ ($183$ of
$183$ and $71$ of $72$). At $m = 6$, with every period sampled at $213$ pieces, Baroque pieces use
$169$ subgroups against $262$ in the 20th century, and the estimated totals ($186$ against $329$)
have intervals that do not overlap. Later repertoires thus use not only more repetition-based figuration but a wider range of
symmetry types, in
agreement with earlier corpus studies, in which the harmonic vocabulary of classical music grows
richer \cite{Serra-Peralta2021}, its key diversity rises \cite{gonzalezespinoza2023harmony}, and
melodies of art music become more complex \cite{mcbride2024information} from the Baroque on. The
curves treat pieces as independent draws, so the comparison is descriptive
(Section~\ref{app:richness}).

\subsectionnobarrier{Transformations between recurring motifs}
\label{sec:res-relations}

Figure~\ref{fig:relations} shows results for the $O/E$ ratio in the operations that link recurring windows. A typical piece has
more than $1{,}000$ of such pairs at $m = 4$, and three quarters of its windows repeat at least once. Literal
repetition is $14$ times more common than for forms drawn at random within each orbit and $3.4$
times more than in shuffled lines, mostly reproduced by lines that keep their one-step transitions,
and pieces that repeat more deploy their motifs over fewer forms. Relations that merely rotate a
window, $60\%$ of the total, are rarer than for random forms in $98\%$ of pieces and rarer than under
the within-orbit shuffle in $88\%$.

\begin{figure}[!htb]
\centering
\includegraphics[width=\linewidth]{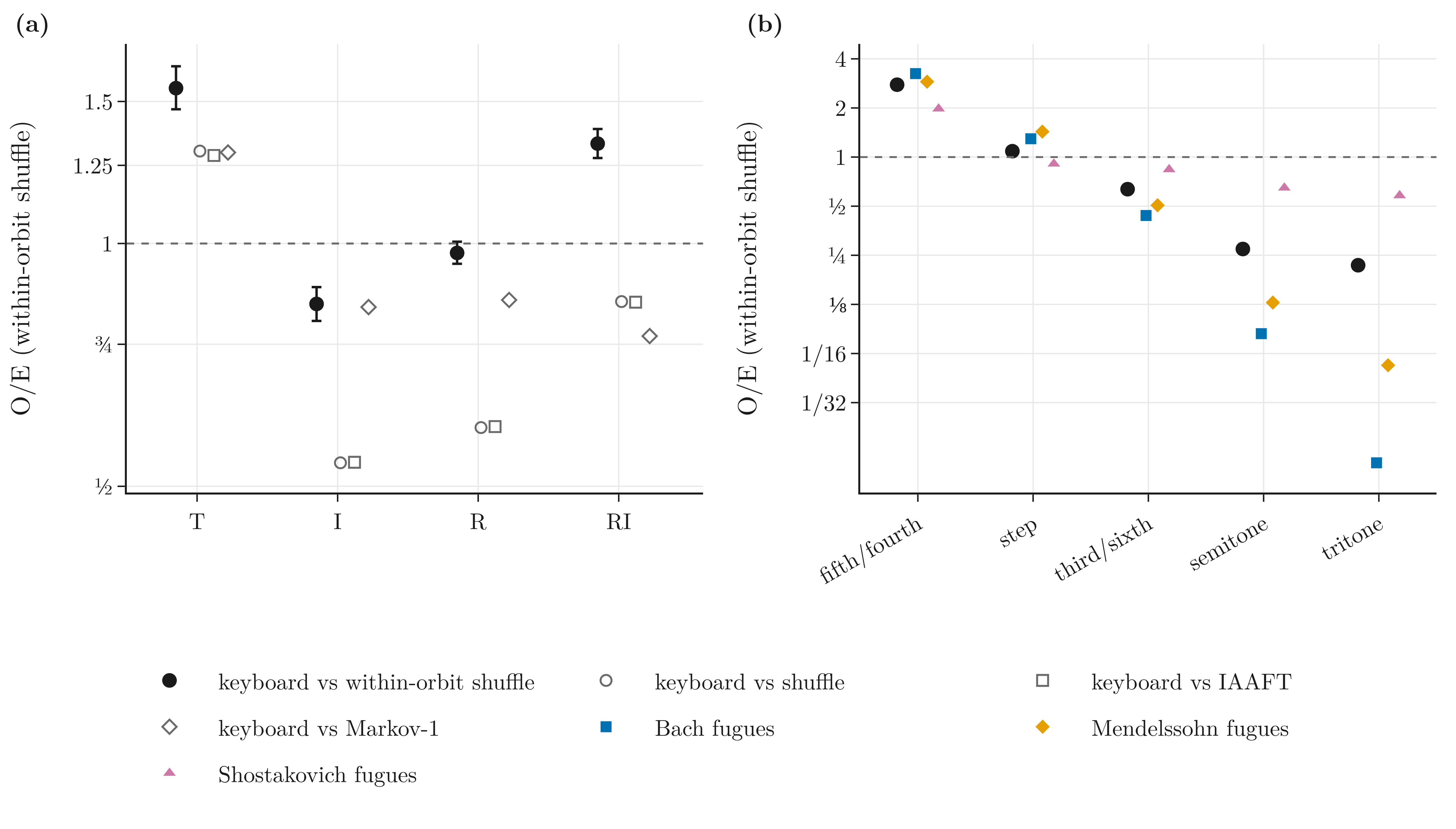}
\caption{Relations between recurring windows, $m = 4$, pairs of asymmetric windows.
\textbf{(a)} Share of each relation type relative to the within-orbit shuffle, which keeps the orbit of
every window and the forms each orbit takes (filled, with $90\%$ intervals that resample
composers), and relative to shuffled, IAAFT, and Markov-1 surrogates (open). T transposition,
I inversion, R retrograde, RI retrograde-inversion. \textbf{(b)} Transposition intervals relative to
the within-orbit shuffle in the keyboard corpus and in the fugues of the Prelude \& Fugue set, fifth
or fourth ($k = 5, 7$), step ($2, 10$), third or sixth ($3, 4, 8, 9$), semitone ($1, 11$), and
tritone ($6$). Dashed lines mark $O/E = 1$.}
\label{fig:relations}
\end{figure}

\begin{figure}[!tbp]
\centering
\includegraphics[width=\linewidth]{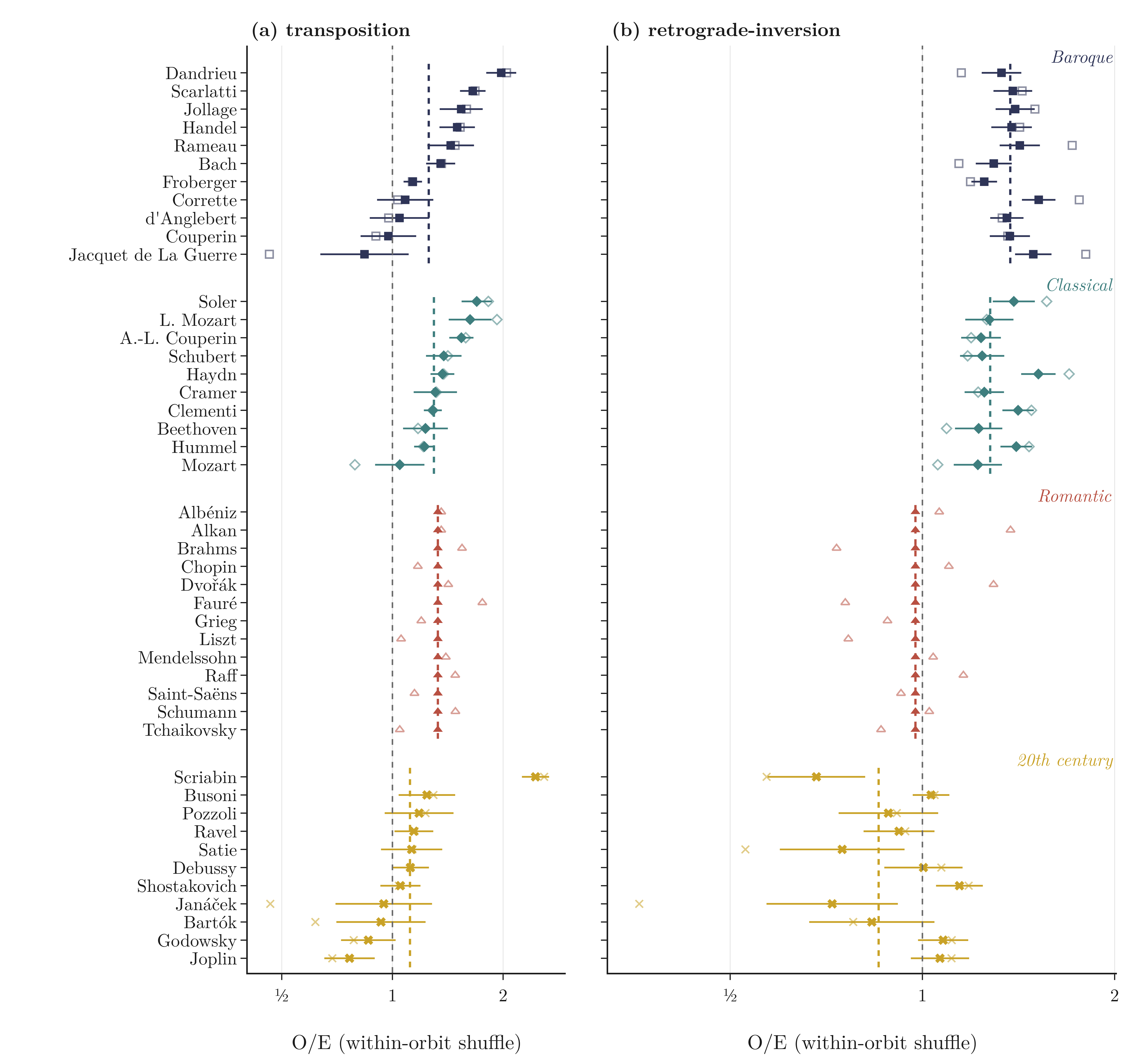}
\caption{Composers relative to the within-orbit shuffle, $m = 4$, pairs of asymmetric windows.
\textbf{(a)} Transposition and \textbf{(b)} retrograde-inversion for composers with at least $10$
pieces, in the same rows (ordered by transposition), shrunk toward the period mean (filled, $\pm 1$
s.d.), with the unshrunk mean (open) and the period mean (dotted). Transposition varies more between
composers than between periods (composer $13\%$ and period $0.1\%$ of the variance), whereas
retrograde-inversion varies between periods (trend $\rho = -0.24$). Romantic composers coincide after
shrinkage because their between-composer variance is estimated as zero.}
\label{fig:relations-composers}
\end{figure}

With vocabulary and deployment held fixed, transposition is $1.56$ times more common than in random
permutations of the forms, in three quarters of the pieces, and twice as common at $m = 5$
(Figure~\ref{fig:relations}a). Retrograde-inversion is also favored ($1.33$), inversion is avoided
($0.84$), and retrograde occurs about as often as chance. The five types are used about as evenly as
under the null ($2.2$ of a possible $2.3$ bits in every period), so the preference lies in which
operations a line uses, not in how many, but the intervals within transposition are not used
evenly. Their entropy lies $0.30$ bits below the within-orbit shuffle at four notes, about $0.45$
bits in the Baroque and Classical periods against $0.21$ in the Romantic and $0.13$ in the 20th
century, and in the fugues it falls from Bach ($0.50$) to Shostakovich ($0.05$)
(Figure~\ref{fig:si-rel-entropy}).

Transposition operations are biased toward the fifth (Figure~\ref{fig:relations}b). Transpositions by a fifth or fourth
make up half of all transpositions in a typical piece, nearly three times their share under the
null, while thirds and sixths, semitones, and tritones are rarer, and steps, the interval of the
sequence, occur as often as chance. The share of fifths falls from about $60\%$ of transpositions in
the Baroque and Classical periods to $42\%$ in the Romantic and $35\%$ in the 20th century, and the
fugues repeat this history in miniature, from $57\%$ in Bach to $52\%$ in Mendelssohn and $35\%$ in
Shostakovich, whose fugues transpose by semitone and tritone at more than half the chance rate,
where Bach's almost never do.

Composers differ in transposition more than periods (Figure~\ref{fig:relations-composers}a). Within
each period they range from $0.76$ times the null for Joplin to $2.5$ for Scriabin, and composer
accounts for $13\%$ of the variation between pieces, period for less than $1\%$, with the same
pattern for inversion at every window size and Romantic composers as alike as sampling can detect at
four notes. Retrograde-inversion shows the opposite pattern (Figure~\ref{fig:relations-composers}b),
the only type whose use changes across periods. Baroque and Classical composers restate motifs as
their retrograde inversion about $1.3$ times more often than chance, Romantic composers at chance,
and 20th-century composers below it ($0.85$), Shostakovich being the exception ($1.14$). A prelude and
its fugue do not differ. These comparisons use the within-orbit shuffle, keep pieces by one composer
together, and are corrected across window sizes (Section~\ref{app:relations},
Figures~\ref{fig:si-rel-types}--\ref{fig:si-rel-repetition}).

\subsection{Symmetric phrases and themes}
\label{sec:res-census}

Figure~\ref{fig:census} shows how often whole phrases and themes are exactly symmetric, by length.
Most are not, about $6\%$ of MTC-ANN phrases, $10\%$ of Essen phrases, and $2\%$ of MTD themes, and
whole melodies almost never are, as counting predicts, since the share of symmetric sequences over a
uniform alphabet falls from $42\%$ at three notes to below $1\%$ at eight. Compared with
rearrangements of their own notes, however, symmetric segments are more common than chance, about
$1.6$ times in the folk corpora and $5$ times among the MTD themes, whether notes or intervals are
rearranged, and increasingly so with length, from no excess in phrases of four to six notes to about
$13$ times in Essen phrases of $13$ to $16$ notes. The kind of symmetry changes with scale. Folk
phrases are symmetric mostly by mirroring ($59$ to $69\%$ of their symmetric phrases), while longer
segments turn to repetition, with translation, a repeated or sequenced cell, accounting for half of
the symmetric MTD themes and for $55$ of the $57$ symmetric Essen melodies. These excesses hold when
segments from the same song or work are resampled together (Section~\ref{app:census}).

\begin{figure}[h]
\centering
\includegraphics[width=\linewidth]{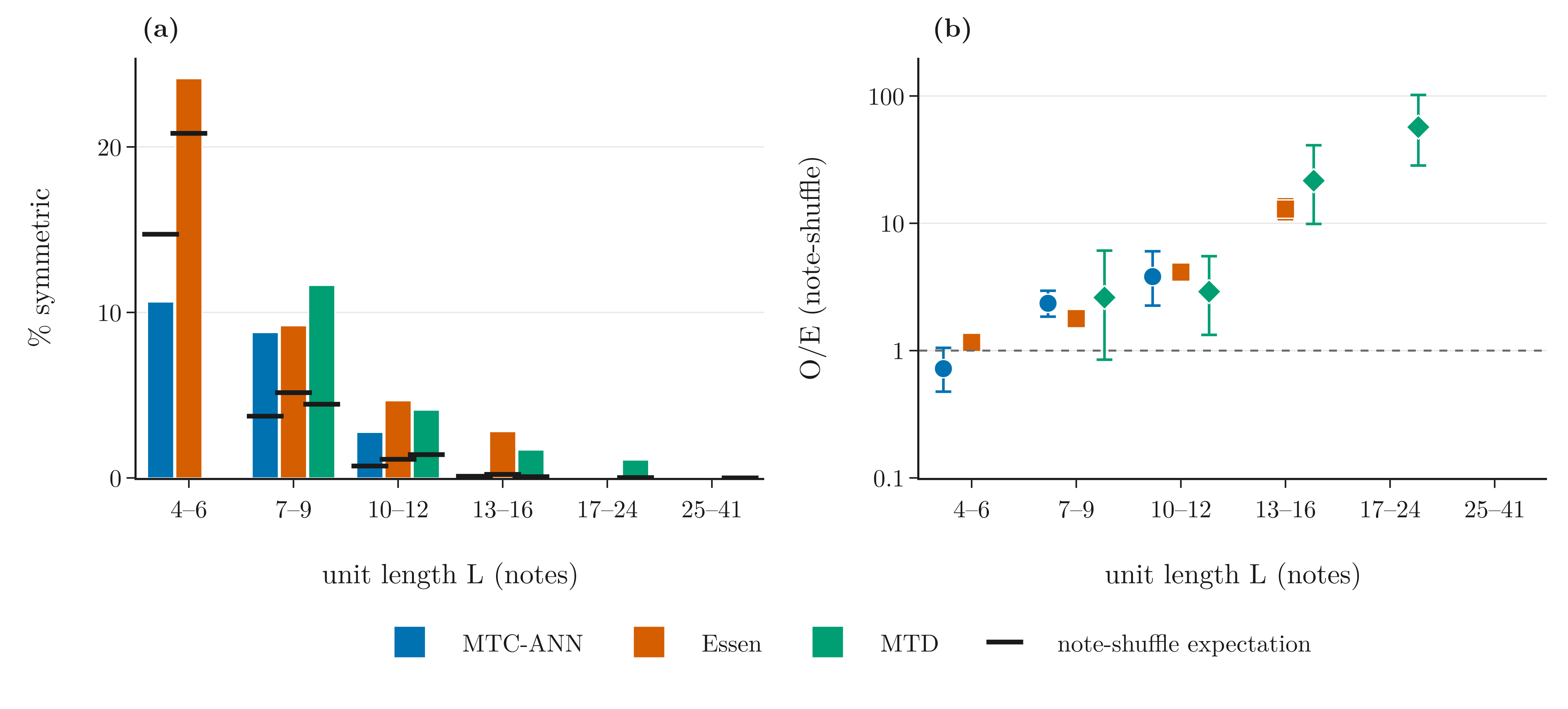}
\caption{Exact symmetry of whole sequences by length bin. \textbf{(a)} Share of symmetric
segments (bars) and the share expected from each segment's own notes in random order (black
tick). \textbf{(b)} Enrichment $O/E = \sum O_L / \sum E_L$ within each bin, with exact
Poisson $95\%$ interval \cite{garwood1936fiducial}. The dotted line marks $O/E = 1$. Bins with
fewer than ten segments, fewer than $0.1$ expected or no observed symmetric segments are omitted
from (b). Per-length estimates are in
Figure~\ref{fig:si-census-null}.}
\label{fig:census}
\end{figure}

\section{Discussion}
\label{sec:discussion}

We developed a framework that reads a melodic line through the orbits and stabilizers of its short
windows under transposition, inversion, retrograde, and cyclic shift, and splits the entropy of the
windows, as statistical mechanics separates macrostates from microstates, into how many motifs a line
draws on and how widely it transforms each. Control sequences fall where this reading interprets them, a
strict ostinato at the crystal limit and a random line near the equipartition ceiling. Real melodies
sit between the two limits. Nearly every piece draws on fewer motif classes than surrogates that keep its
pitches, spectrum, or one-step transitions, a gap that widens with window size
(Section~\ref{sec:res-surrogates}), and from five notes on ($m \geq 5$) it uses a more varied mix of symmetry types
than chance, as if composers reused few ideas and varied them by transformation. The same deficit appears under other readings: block entropy, which uses a smaller group, and permutation entropy, which uses a different one, also fall below the surrogates. Neither splits the count into vocabulary and deployment, so neither can trace a difference to its source; only the orbit decomposition assigns differences between periods to vocabulary and differences between forms to deployment (Section~\ref{app:entropies}).

Most windows and whole sequences are asymmetric, as counting predicts (sorting all the permutations by their symmetry group and selecting randomly a symmetry type with a dice weighted by its proportion), since each added note is
another condition a symmetry must satisfy (e.g., a window of three notes reads the same backwards
when its outer notes agree, one of five notes only when two pairs agree)
\cite{reiner1985enumeration,fripertinger1992enumeration}, so the finding lies in the departure from
counting. In short windows the excess of mirror figures over shuffled lines comes largely from
one-step transitions, stepwise motion that returns, such as neighbor notes and turns. Repetition-based
symmetry, in contrast, exceeds every null and grows from the Baroque to the 20th century, although it
covers only $13$ to $18\%$ of windows at four notes. Whole phrases and themes are rarely symmetric
($2$ to $10\%$) but $1.6$ to $5$ times more often than their own notes allow, by mirroring in folk
songs and by repetition in instrumental themes. Their null reorders the notes, and so breaks the steps.
Short folk mirrors are therefore likely stepwise returns, since folk phrases move by step and trace
arches \cite{huron1996melodic}, whereas palindromes of $13$ to $16$ notes need an exact order and
should survive a null that keeps each phrase's steps
\cite{altschul1985significance,kandel1996shuffling}, a test left for future work.

Periods differ in vocabulary and symmetry types, forms and composers in how they deploy their motifs
(Sections~\ref{sec:res-hor-eta} and~\ref{sec:res-pf}). Baroque lines are the most structured over
three notes and draw on the fewest symmetry subgroups, and later repertoires add repetitive figuration
and a wider range of symmetry types, the same direction as in the harmonic vocabulary and key
diversity of the repertoire \cite{Serra-Peralta2021,gonzalezespinoza2023harmony}. With form held
fixed, deployment separates composers, and Bach's fugues spread each motif over fewer forms than
Shostakovich's. The relations between motifs add how ideas are restated. Composers differ in how often
they transpose, and the preference for the fifth interval weakens from the
Baroque to the 20th century and from Bach's fugues to Shostakovich's, where transposition intervals
are spread almost as by chance. Retrograde-inversion is the one relation that follows the periods, and
since it reads a motif's intervals backwards, turning a major broken chord into a minor one, its
excess fits the alternating major and minor figuration of Baroque and Classical sequences and its
decline parallels that of the fifth. The distributions of the seventeen wallpaper groups separate the
cultures that produced an ornament \cite{erbudak2019sanmarco,erbudak2020ornaments}, while the mix of
window symmetry types changes only mildly over four centuries (Figure~\ref{fig:genome-periods}). In
melody, style seems to lie less in which symmetries occur than in how figuration repeats and how
motifs are related. Counts of symmetry classes also depend on how they are implemented, as successive
studies of the Alhambra show \cite{grunbaum1986moorish,grunbaum2006alhambra}.

Counting also makes precise a long-standing view in music theory, that symmetry operations act mainly as
relations between distinct objects \cite{lewin1987gmit,forte1973structure,Tymoczko2011}, from motivic
development \cite{schoenberg1950style,reti1951thematic} to tune families \cite{bayard1950}, with
self-symmetry rare and marked \cite{messiaen1944technique}. By the orbit--stabilizer theorem an
asymmetric motif has the most forms, so its inversion is a new object, as the major and minor triads,
the chiral trichords, are of each other \cite{rahn1980basic}, and as a chiral molecule is of its
mirror image. Transposition, the relation listeners recognize most readily, is the one favored beyond chance, while inversion and retrograde, recognized poorly~\cite{dowling1972recognition,krumhansl1990cognitive}, sit at the surrogate level. Accidental transform-pairs between asymmetric motifs are already present in the surrogates at their expected rate, and such accidents would enrich every relation alike; the enrichment is instead specific to transposition, which chance among asymmetric figures does not produce. The null that isolates these choices keeps a line's vocabulary and deployment, holds each window's position fixed, and resamples only which form of its orbit appears there, so the test reports how motifs are related with which motifs are used held constant. In Lewin's terms these relations are
intervals between motifs, and the census gives their distribution over a corpus and not within a
single work. The excess symmetry of protein complexes has been explained by generative processes that
favor short descriptions \cite{goodsell2000structural,johnston2022symmetry}, and composition and oral
transmission, which work with short rules such as repeat, answer, and mirror, would likewise predict
more excess symmetry in more formulaic repertoires.

The framework presented in this work reads one extracted line at a time, in equal temperament and without rhythm. The
choice of line has little effect on the conclusions. Two selectors, one reading register and one
information content, return the same line in only $40$ to $47\%$ of pieces with several tracks, yet
$70$ of $74$ tests keep their conclusion with either, and in fugues both lines carry the subject
(Section~\ref{app:extraction-check}). The keyboard corpus is balanced by period, composer, and length
but not by form, with periods and forms assigned from metadata, and all corpora come from digitized
Western repertoire. The census segments were cut by people, who may favor the symmetry easiest to
recognize, and the census corpora differ in content as well as in how they were cut. Relations
between voices and pieces, as in canon and fugue \cite{mann1958study}, serial rows
\cite{babbitt1960twelve}, and tune families \cite{bayard1950}, are the natural next count. Reading
durations as well as pitches would separate rhythmic variants, though by counting alone nearly every
segment would then be asymmetric, so the count of relations between segments would matter more than
their self-symmetry. At six notes, $1{,}122$ of the $10{,}962$ orbits never occur in the corpus,
inversion-type orbits in larger proportion, and whether they are avoided or only undersampled remains
open, as does the rise of windows with several symmetries at once, which stay above every null.

Nothing in the orbit decomposition is specific to pitch classes. It needs an
alphabet, a group of operations that identify a figure with its transforms, and a
window short enough that its patterns recur. Changing any of the three poses the same question of a
new system. For example, with the seven scale degrees and diatonic transposition, a tonal sequence
that moves a figure up by step, turning a major second into a minor one, falls into
one orbit instead of several, so the deficit measured here could be split into
chromatic and diatonic parts. With onset and rest as the alphabet and the rotations of the cycle as the group, a
timeline together with its rotations forms a necklace, the orbit named by the geometry
of rhythm~\cite{toussaint2013geometry}, and a nontrivial stabilizer marks the rhythms
whose onset pattern repeats under a rotation shorter than the cycle. Outside music, time and amplitude reversal act on the ordinal patterns of a
time series~\cite{bandt2002permutation}, and the reverse complement acts on a DNA
strand, where the near-equal counts of complementary bases within one
strand~\cite{rudner1968separation} are exactly within-orbit equipartition, $\eta = 1$
at a window of one. In each case orbits, stabilizers, and relations between windows are
defined exactly as above, and only the nulls must be rebuilt for the question asked.
What carries over is the separation reported above, between how many distinct ideas a
sequence uses and how widely it transforms each, a distinction a single entropy cannot
draw.

\ifrsos
\section*{Ethics}
This article does not present research with ethical considerations.

\section*{Data accessibility}
Sources: Kunst der Fuge MIDI files (\url{https://www.kunstderfuge.com}); OpenWTC from PDMX
\cite{long2024pdmx} (dataset CC BY 4.0; the $24$ scores used are CC0, none flagged with a
license conflict); MTC-ANN 2.0.1 (CC BY-NC-SA 3.0); the Essen Folksong Collection
(CC BY-SA 4.0); and the Musical Theme Dataset, which states no license (the CC BY 4.0 of
\cite{zalkow2020mtd} covers the article) and is used as distributed at
\url{https://www.audiolabs-erlangen.de/resources/MIR/MTD}, not redistributed. Selection lists,
extracted melodic lines, and all derived tables are available to reviewers at submission and will be deposited in a
public archive with a persistent identifier upon acceptance.

The analysis uses two Julia packages written for this study, \texttt{MelodicSymmetry.jl}
(group action, orbits, stabilizers, and entropies;
\url{https://github.com/spiralizing/MelodicSymmetry.jl}) and \texttt{MelodyExtraction.jl}
(melody extraction; \url{https://github.com/spiralizing/MelodyExtraction.jl}), and the scripts
and notebooks that produce every table and figure
(\url{https://github.com/spiralizing/SymmetryEntropiesInMelodies}). All code is released under
the MIT License, is available to reviewers at submission, and will be made public upon acceptance.

\section*{Declaration of AI use}
We used VSCode with GitHub Copilot (Claude Sonnet 4.6) as auto-complete assistant to write most of the code in the functions of the packages \texttt{MelodicSymmetry.jl} and \texttt{MelodyExtraction.jl}, and the Notebooks for the computations and plots in this work. Claude Code with Opus 4.8 was used to build tests and documentation for the packages. Claude Code with Claude Opus 5.5 was used to write most of the code for the statistical
analyses reported in the Supplementary Information, to create and modify the TikZ diagrams, and to organize text and tables in the Supplementary Information. All code and documentation generated with LLMs was reviewed, tested, and modified prior to publication.

\section*{Authors' contributions}
A.G.-E.: [to be completed before submission, e.g. conceptualization, methodology, software, formal
analysis, writing].

\section*{Conflict of interest declaration}
I declare I have no competing interests.

\section*{Funding}
This work received no specific funding.

\section*{Acknowledgements}
I thank Gustavo Mart\'inez-Mekler and Eduardo Garc\'ia Castrej\'on for valuable feedback during the conceptual development of this work. I am grateful to the open-source and open-science communities whose freely shared software, data, and documentation made this work possible.
\else
\section*{Acknowledgments}
I thank Gustavo Mart\'inez-Mekler and Eduardo Garc\'ia Castrej\'on for valuable feedback during the conceptual development of this work. I am grateful to the open-source and open-science communities whose freely shared software, data, and documentation made this work possible. 

\section*{Data Availability}
Sources: Kunst der Fuge MIDI files (\url{https://www.kunstderfuge.com}); OpenWTC from PDMX
\cite{long2024pdmx} (dataset CC BY 4.0; the $24$ scores used are CC0, none flagged with a
license conflict); MTC-ANN 2.0.1 (CC BY-NC-SA 3.0); the Essen Folksong Collection
(CC BY-SA 4.0); and the Musical Theme Dataset, which states no license (the CC BY 4.0 of
\cite{zalkow2020mtd} covers the article) and is used as distributed at
\url{https://www.audiolabs-erlangen.de/resources/MIR/MTD}, not redistributed. Selection lists,
extracted melodic lines, and all derived tables will be deposited in a public archive with a
persistent identifier upon journal publication and are available from the author on request.

\section*{Code Availability}
The analysis uses two Julia packages written for this study, \texttt{MelodicSymmetry.jl}
(group action, orbits, stabilizers, and entropies;
\url{https://github.com/spiralizing/MelodicSymmetry.jl}) and \texttt{MelodyExtraction.jl}
(melody extraction; \url{https://github.com/spiralizing/MelodyExtraction.jl}), and the scripts
and notebooks that produce every table and figure
(\url{https://github.com/spiralizing/SymmetryEntropiesInMelodies}). All code is released under
the MIT License, will be made public upon journal publication, and is available from the author on
request.

\section*{Use of LLM disclosure}

We used VSCode with GitHub Copilot (Claude Sonnet 4.6) as auto-complete assistant to write most of the code in the functions of the packages \texttt{MelodicSymmetry.jl} and \texttt{MelodyExtraction.jl}, and the Notebooks for the computations and plots in this work. Claude Code with Opus 4.8 was used to build tests and documentation for the packages. Claude Code with Claude Opus 5.5 was used to write most of the code for the statistical
analyses reported in the Supplementary Information, to create and modify the TikZ diagrams, and to organize text and tables in the Supplementary Information. All code and documentation generated with LLMs was reviewed, tested, and modified prior to publication.
\fi

\bibliographystyle{unsrt}
\bibliography{references}

%% file: SI.tex
\clearpage
\setcounter{section}{0}\setcounter{subsection}{0}\setcounter{figure}{0}
\setcounter{table}{0}\setcounter{equation}{0}
\renewcommand{\thesection}{S\arabic{section}}
\renewcommand{\thefigure}{S\arabic{figure}}
\renewcommand{\thetable}{S\arabic{table}}
\renewcommand{\theequation}{S\arabic{equation}}
\renewcommand{\theHsection}{S\arabic{section}}
\renewcommand{\theHfigure}{S\arabic{figure}}
\renewcommand{\theHtable}{S\arabic{table}}
\renewcommand{\theHequation}{S\arabic{equation}}
\begin{center}{\Large\bfseries Supplementary Information}\end{center}
\addcontentsline{toc}{section}{Supplementary Information}
\startcontents[si]
\begingroup
\setcounter{tocdepth}{2}
\printcontents[si]{}{1}{\noindent{\large\bfseries Contents}\par\medskip}
\endgroup
\clearpage
\section{Mathematical details}
\label{app:decomposition}

\subsection{Commutativity of tonal and positional components}
\label{app:commute}

\begin{theorem}[Commutativity of tonal and positional components]
\label{thm:commute}
For any $g = (g_1, g_2) \in G$ and any $x \in \Z_{12}^m$,
$g_2 \cdot (g_1 \cdot x) = g_1 \cdot (g_2 \cdot x)$.
\end{theorem}

\begin{proof}
$g_1 \in D_{12}$ acts on values and $g_2 \in D_m$ on indices, so
$\bigl(g_1 \cdot (g_2 \cdot x)\bigr)_i = g_1 \cdot x_{g_2^{-1}(i)} = \bigl(g_2 \cdot (g_1 \cdot x)\bigr)_i$
for every $i$.
\end{proof}

\subsection{Entropy chain rule and the intra-orbit bound}

For the projection $\pi: \Z_{12}^m \to X/G$, $w \mapsto \hat o$, and a window $W$ drawn from the
empirical distribution,
\begin{equation}
\label{eq:app-chain}
H(W) \;=\; H(\pi(W)) \;+\; H(W \mid \pi(W)),
\end{equation}
with $H(W) = \Hmic$ (surface variety), $H(\pi(W)) = \Hmac = \Hor$ (orbit variety), and
$H(W \mid \pi(W)) = \langle H_{\text{intra-orbit}} \rangle$ (within-orbit variety),
\begin{equation}
\label{eq:app-intraorbit}
\langle H_{\text{intra-orbit}} \rangle
\;=\;
\sum_{o \in X/G} P(o) \, H\!\left(P_{\text{within}}(\,\cdot \mid o)\right),
\qquad
P_{\text{within}}(w \mid o) = \frac{P(w)}{P(o)}
\text{ for } w \in o.
\end{equation}
$P_{\text{within}}(\cdot \mid o)$ has support $|\Orb(o)|$, so
\begin{equation}
\label{eq:app-perorbit}
H\!\left(P_{\text{within}}(\,\cdot \mid o)\right) \;\le\; \log_2 |\Orb(o)| \;=\; \log_2 \frac{|G|}{|\Stab(o)|} \;=\; S(o),
\end{equation}
and, averaging over $P(o)$,
\begin{equation}
\label{eq:app-globalceiling}
\langle H_{\text{intra-orbit}} \rangle \;\le\; \langle S \rangle,
\end{equation}
with equality iff every visited orbit is uniformly deployed.

\subsection{Closed form for $|\Fix(g)|$ in the orbit-counting lemma}

A window is fixed by $g = (g_1, g_2)$ iff it is constant on each cycle of $g_2$ with a value
fixed by $g_1$, so
\begin{equation}
\label{eq:app-cycle}
|\Fix(g)| \;=\; \prod_{C \in \text{cycles}(g_2)} |\Fix_{\Z_{12}}(g_1)|
              \;=\; |\Fix_{\Z_{12}}(g_1)|^{c(g_2)},
\end{equation}
with $c(g_2)$ the number of cycles of $g_2$, which in Eq.~\eqref{eq:orbit-counting} gives $K(m)$ in
closed form \citeSI{SI:ibragimova2025algebraic}.

\subsection{Group-action conventions and canonicalization}
\label{app:conventions}

The generators of Eq.~\eqref{eq:generators} satisfy $r^{12} = f^2 = e$ and $f r f = r^{-1}$,
so $|D_{12}| = 24$; likewise $R^m = F^2 = e$ and $F R F = R^{-1}$ in $D_m$. Composite
elements are read by functional composition, $r^k$ first and then $f$,
\begin{equation}
\label{eq:fr-convention}
f r^k \cdot p \;=\; -p - k \bmod 12,
\end{equation}
that is, inversion through the pitch axis at $-k/2 \bmod 6$.
Pitch classes are indexed from $C = 0$. Ibragimova and Nehaniv
\citeSI{SI:ibragimova2025algebraic} index them $1$--$12$ with $\mathrm{Inv}(m) = 13 - m$, the same
$D_{12}$ action with another origin.

\emph{Canonicalization.} The $48m$ images $w_i \ast g$ (tonal part $f^a r^k$, then $F^b R^j$, in
either order by Theorem~\ref{thm:commute}) are compared lexicographically, and the minimum is the
canonical representative $\hat o_i$, which assigns and labels the orbit in one pass. $K(m)$ up to
$m = 8$ is in Table~\ref{tab:Km}.

\begin{table}[h]
\centering
\caption{Total orbit count $K(m)$ for $G = D_{12} \times D_m$.}
\label{tab:Km}
\small
\begin{tabular}{rrrrrrr}
\toprule
$m$    &  3 &   4 &      5 &       6 &        7 &           8 \\
\midrule
$K(m)$ & 19 & 163 & 1{,}110 & 10{,}962 & 107{,}509 & 1{,}126{,}566 \\
\bottomrule
\end{tabular}
\end{table}

\subsection{Stabilizer levels and retrograde-inversion}
\label{app:stab-levels}

A stabilizer is read at four levels (Table~\ref{tab:stab-levels}), the exact subgroup, its
conjugacy class, shared by all windows of an orbit ($6$ classes at $m = 3$, $21$ at $m = 4$), its
abstract family ($\Z_2$, $\Z_2 \times \Z_2$, $D_3$, \ldots), and its class, used in the main text.
Conjugation in $G$ acts on each dihedral factor separately and keeps rotations and reflections
apart, so the class is a function of the conjugacy class and hence of the orbit. Class and family
are not nested ($58$ subgroups, $19$ orbits at $m = 3$).

\begin{table}[h]
\centering

\caption{Stabilizers of four motifs at $m = 3$ ($|G| = 144$), with the number of
surface forms $|G|/|\Stab|$. Elements are written $f^a r^k F^b R^j$ (inversion,
transposition, retrograde, cyclic shift).}
\label{tab:stab-levels}
\small
\begin{tabular}{l l c l l r}
\toprule
\textbf{Motif} & \textbf{Non-identity elements of $\Stab$} & $|\Stab|$ &
\textbf{Family} & \textbf{Class} & \textbf{Forms} \\
\midrule
$(0,4,7)$ & none & 1 & trivial & asymmetric & 144 \\
$(0,1,2)$ & $f r^{10} F$ & 2 & $\Z_2$ & retrograde-inversion & 72 \\
$(0,6,0)$ & $F$,\; $f$,\; $fF$ & 4 & $\Z_2 \times \Z_2$ & multiple & 36 \\
$(0,4,8)$ & $r^4 R$,\; $r^8 R^2$,\; $f F R^2$,\; $f r^4 F$,\; $f r^8 F R$ & 6 & $D_3$ & multiple & 24 \\
\bottomrule
\end{tabular}
\end{table}

Retrograde windows detect return figures (arches, neighbors, turns), and RI windows inversional
balance, $x_i + x_{\rho(i)} = \text{const}$. Retrograde-inversion reverses the order of the
intervals, so every constant-step window, such as $(0,1,2)$, is RI-type at every $m$
(Figure~\ref{fig:ri-contour}).

\begin{figure}[h]
\centering
\begin{tikzpicture}[every node/.style={font=\scriptsize}]
  \foreach \ox in {0, 1.9, 3.8, 5.7} {
    \draw[->, kink!35] (\ox, 0) -- ({\ox+1.35}, 0);
    \draw[->, kink!35] (\ox, 0) -- (\ox, 1.45);
  }
  \node[kink!60, below]      at (1.25, 0)   {$t$};
  \node[kink!60, left]       at (0, 1.35)   {$x$};
  \draw[thick, kink] (0.1,0.6) -- (0.45,0.6) -- (0.85,1.2) -- (1.2,1.2);
  \begin{scope}[shift={(1.9,0)}]
    \draw[kindigo, dashed] (0.65, 0.1) -- (0.65, 1.35);
    \draw[thick, kink] (0.1,1.2) -- (0.45,1.2) -- (0.85,0.6) -- (1.2,0.6);
  \end{scope}
  \begin{scope}[shift={(3.8,0)}]
    \draw[kterra, dashed] (0.05, 0.6) -- (1.3, 0.6);
    \draw[thick, kink] (0.1,0.0) -- (0.45,0.0) -- (0.85,0.6) -- (1.2,0.6);
  \end{scope}
  \begin{scope}[shift={(5.7,0)}]
    \draw[kink!35, dashed] (0.1,0.0) -- (0.45,0.0) -- (0.85,0.6) -- (1.2,0.6);
    \draw[->, kteal, thick] (0.28, 0.08) -- (0.28, 0.52);
    \draw[thick, kink] (0.1,0.6) -- (0.45,0.6) -- (0.85,1.2) -- (1.2,1.2);
  \end{scope}
  \draw[->, thick, kindigo] (1.4, 0.7) -- (1.8, 0.7);
  \draw[->, thick, kterra]  (3.3, 0.7) -- (3.7, 0.7);
  \draw[->, thick, kteal]   (5.2, 0.7) -- (5.6, 0.7);
  \node[kink]    at (0.65, -0.35) {start};
  \node[kindigo] at (2.55, -0.35) {retrograde};
  \node[kterra]  at (4.45, -0.35) {inversion};
  \node[kteal]   at (6.35, -0.35) {transposition};
  \node[kink]    at (6.35, -0.65) {$=$ start};
\end{tikzpicture}
\caption{Retrograde-inversion returns a flat--rise--flat shape or {\em contour} to
itself. Retrograde flips the contour in time (dashed vertical axis),
inversion flips it in pitch (dashed horizontal axis), and transposition
moves it back into place. The three operations together amount to a
half-turn rotation of the contour, so any contour that looks the same
rotated by $180^\circ$ is fixed by retrograde-inversion. Every
constant-step sequence (chromatic when step is one) is such a contour.}
\label{fig:ri-contour}
\end{figure}
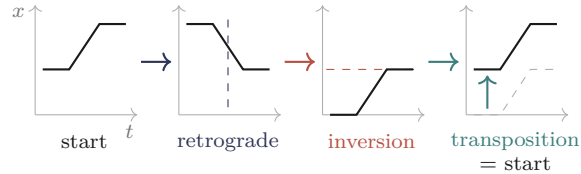

\section{Entropy estimation: the NSB estimator}
\label{app:nsb}

The plug-in estimate $\hat H_{\mathrm{MLE}} = -\sum_o \hat P(o) \log_2 \hat P(o)$, with
$\hat P(o) = n(o)/N$, is biased downward when $N \lesssim K(m)$ (Table~\ref{tab:Km}) and returns $0$
for a piece that visits one orbit. We use the Nemenman--Shafee--Bialek (NSB) estimator
\citeSI{SI:nsb2002,SI:nemenman2011}. A symmetric Dirichlet prior of concentration $\beta$ over the
$K = K(m)$ orbits,
\begin{equation}
\label{eq:nsb-prior}
p(P \mid \beta) = \frac{\Gamma(K\beta)}{\Gamma(\beta)^{K}}
   \prod_{o=1}^{K} P(o)^{\beta - 1},
\end{equation}
gives the Dirichlet--multinomial evidence
\begin{equation}
\label{eq:nsb-evidence}
\mathcal{L}(\beta) = \frac{\Gamma(K\beta)}{\Gamma(N + K\beta)}
   \prod_{o=1}^{K} \frac{\Gamma\bigl(n(o) + \beta\bigr)}{\Gamma(\beta)}
\end{equation}
and, at fixed $\beta$, the posterior mean entropy in nats \citeSI{SI:wolpert1995}
\begin{equation}
\label{eq:nsb-wolpert}
\langle H \rangle_{\beta}
  = \psi_0(N + K\beta + 1)
  - \sum_{o=1}^{K} \frac{n(o) + \beta}{N + K\beta}\,
      \psi_0\bigl(n(o) + \beta + 1\bigr),
\end{equation}
with $\psi_0$ the digamma function. The hyperprior
\begin{equation}
\label{eq:nsb-hyperprior}
\rho(\beta) \propto \frac{d\xi}{d\beta}, \qquad
\xi(\beta) = \psi_0(K\beta + 1) - \psi_0(\beta + 1),
\end{equation}
makes the induced prior on the entropy approximately uniform on $[0, \log_2 K]$, and the estimate
is the posterior mean
\begin{equation}
\label{eq:nsb}
\hat H_{\mathrm{NSB}}
  = \frac{\displaystyle\int d\beta\; \rho(\beta)\,
          \mathcal{L}(\beta)\, \langle H \rangle_{\beta}}
         {\displaystyle\int d\beta\; \rho(\beta)\, \mathcal{L}(\beta)},
\end{equation}
converted to bits by $1/\ln 2$. It does not collapse to $0$ under sparse sampling and gives a
posterior variance. Since the correction depends on the alphabet size, the support is fixed at the
exact orbit count $K(m)$ (Eq.~\eqref{eq:orbit-counting}), not at the number of observed orbits.

\section{Orbit entropy among other entropy measures}
\label{app:entropies}

\textbf{Readings of one count.} Each reading is the entropy of a quotient of the window space
(Table~\ref{tab:si-entropies-def}), the block ($m$-gram) entropy of pitch classes ($\Hmic$,
Eq.~\eqref{eq:decomp}), the interval $m$-gram entropy (up to transposition, as in $n$-gram models of
melody \citeSI{SI:conklin1995multiple}), and $\Hor$ (up to $G$). Permutation entropy
\citeSI{SI:bandt2002permutation} counts rank orders of pitch height, with repeated notes as ties
(weak orderings of $m$ items). The first three are nested, $H_{\text{block}} \ge
H_{\text{interval}} \ge \Hor$. The entropy rate $h_m = H_{\text{block}}(m) - H_{\text{block}}(m-1)$
is the conditional entropy of an order-$(m-1)$ model \citeSI{SI:crutchfield2003regularities}. Each
entropy is estimated with NSB over its own $K$.

\begin{table}[htbp]
\centering\footnotesize
\caption{Readings of the windows of a line and their numbers of states $K$.}
\label{tab:si-entropies-def}
\setlength{\tabcolsep}{4pt}
\begin{tabular}{l l l r}
\toprule
Reading & Windows equal when & Group & $K$ at $m = 3, 4, 5, 6$ \\
\midrule
block ($m$-gram) & identical & $\{e\}$ & 1{,}728; 20{,}736; 248{,}832; 2{,}985{,}984 \\
interval $m$-gram & equal up to transposition & $\Z_{12}$ & 144; 1{,}728; 20{,}736; 248{,}832 \\
orbit & equal up to T, I, R, cyclic shift & $D_{12} \times D_m$ & 19; 163; 1{,}110; 10{,}962 \\
permutation & same rank order (ties kept) & order-preserving maps & 13; 75; 541; 4{,}683 \\
\bottomrule
\end{tabular}
\end{table}

\textbf{Hierarchy and sampling.} Medians at $m = 4$ are $8.38$ bits (block), $7.25$ (interval),
$5.47$ (orbit), and $5.08$ (permutation) (Figure~\ref{fig:si-entropies}a). Transposition removes
$1.1$--$1.7$ bits from the block entropy and the rest of $G$ another $1.2$--$2.3$ ($m = 3$--$6$). At
$m = 4$ a typical piece visits $0.8\%$ of the pitch-class 4-grams, $7\%$ of the interval 3-grams,
$37\%$ of the orbits, and $65\%$ of the rank patterns, with NSB corrections of $1.02$, $0.62$,
$0.22$, and $0.15$ bits (Figure~\ref{fig:si-entropies}b). The group keeps the entropy estimable as
windows grow.

\begin{figure}[htbp]
\centering
\includegraphics[width=\linewidth]{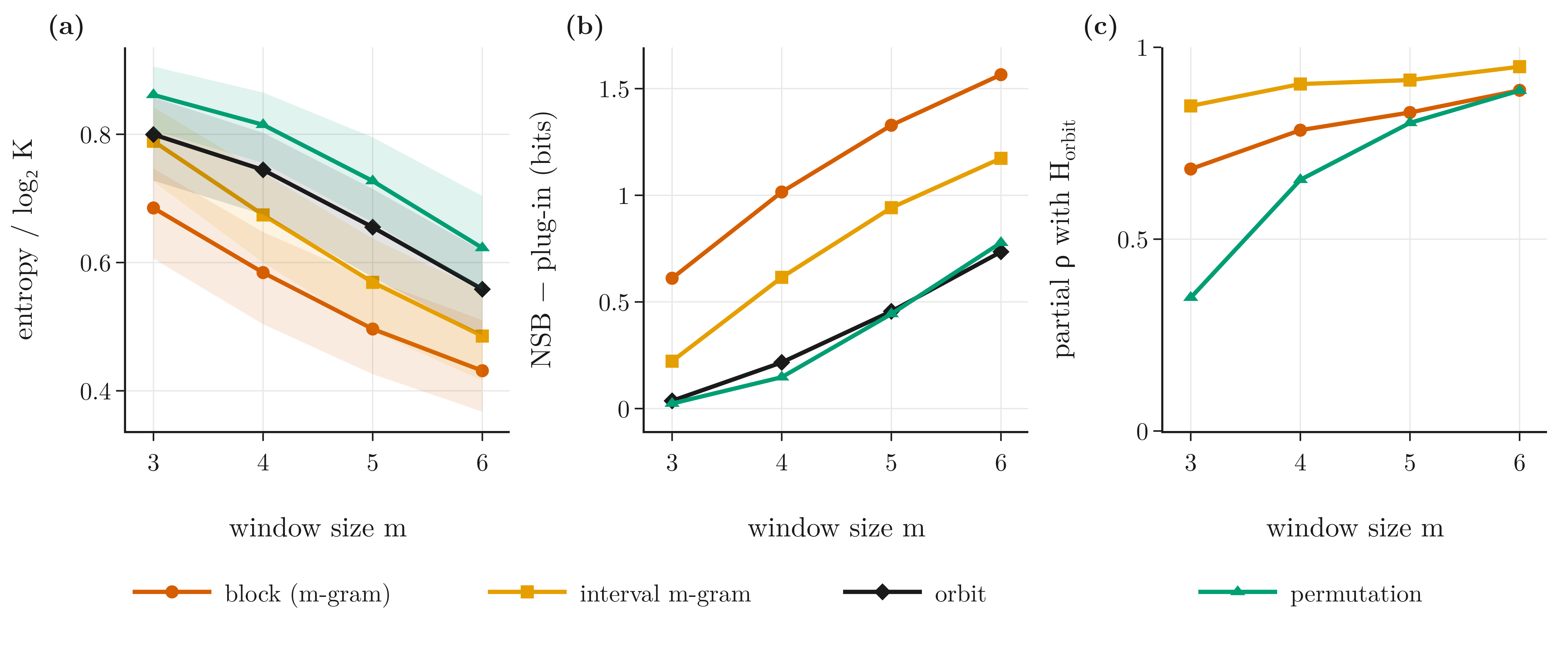}
\caption{Orbit entropy among other entropy measures, keyboard corpus. \textbf{(a)} Median entropy
normalized by $\log_2 K$ (bands: interquartile range), \textbf{(b)} NSB correction to the plug-in
estimate, and \textbf{(c)} partial rank correlation with $\Hor$ (length removed), for windows of
$m = 3$--$6$ notes.}
\label{fig:si-entropies}
\end{figure}

\textbf{Agreement and separation.} Partial rank correlations with $\Hor$ (length removed) are
$0.85$--$0.95$ (interval), $0.68$--$0.89$ (block), and $0.35$--$0.89$ (permutation, $m = 3 \to 6$)
(Figure~\ref{fig:si-entropies}c). Tests are in Table~\ref{tab:si-entropies}. Block entropy is $\Hor$
plus the deployment term (Eq.~\eqref{eq:decomp}), so its form and fugue differences come from
deployment, which $\eta$ isolates. Neither $\eta$ nor the relations between motifs
(Section~\ref{app:relations}) has a counterpart in readings that discard the forms within an orbit.

\begin{table}[htbp]
\centering\footnotesize
\caption{Entropy readings compared at $m = 4$, keyboard corpus (Prelude \& Fugue set for the fugue
trend). Periods: Kruskal--Wallis with composers permuted between periods (Holm over $m$; values at
$m = 3$ where they differ); ordered trend as Spearman $\rho$. Forms: Kruskal--Wallis with labels
permuted within composers. Variance: adjusted $R^2$ of ranks for period, composer beyond period,
and form beyond composer. Fugue trend: Jonckheere--Terpstra, Bach to Shostakovich. Last column:
share of pieces significantly below / above their IAAFT surrogates after Benjamini--Hochberg control.
Bold: $p < 0.05$. All values for $m = 3$--$6$ in \texttt{Data/entropy\_comparison\_stats.csv}; permutations
are independent of Table~\ref{tab:si-separation}, so $p$ may differ in the third decimal.}
\label{tab:si-entropies}
\setlength{\tabcolsep}{3pt}
\resizebox{\linewidth}{!}{%
\begin{tabular}{l l c c c c}
\toprule
Measure & Periods $p$ ($\rho$) & Forms $p$ & Variance & Fugue $p$ & Below / above (\%) \\
\midrule
block & 1.00 ($+0.13$) & \textbf{0.026} & 0.02 / 0.31 / 0.02 & \textbf{0.021} & 98 / 0 \\
interval & 1.00 ($+0.11$) & 0.437 & 0.02 / 0.28 / 0.01 & 0.526 & 99 / 0 \\
orbit & 0.187 ($+0.16$); \textbf{0.001} at $m = 3$ & 1.00 & 0.04 / 0.22 / 0.01 & 0.864 & 97 / 0 \\
permutation & 1.00 ($+0.03$) & 0.144; \textbf{0.040} at $m = 3$ & 0.00 / 0.18 / 0.01 & 0.376 & 91 / 1 \\
rate & 1.00 ($+0.02$) & \textbf{0.001} & 0.00 / 0.34 / 0.02 & 0.080 & 96 / 0 \\
$\eta$ & 1.00 ($-0.04$) & \textbf{0.001} & 0.00 / 0.19 / 0.02 & \textbf{0.002} & 27 / 42 \\
$H_{\text{stab}}$ & 0.869; \textbf{0.005} at $m = 3$ & 0.435; \textbf{0.012} at $m = 3$ & 0.00 / 0.11 / 0.01 & \textbf{0.011} & 9 / 23 \\
\bottomrule
\end{tabular}}
\end{table}

\section{Data}
\label{app:data}

\subsection{Construction of the solo keyboard corpus}
\label{app:corpus-build}

The source is the Kunst der Fuge collection of about $18{,}000$ MIDI files
\citeSI{SI:kunstderfugue}, the same collection used in earlier corpus studies
\citeSI{SI:Gonzalez-Espinoza2017,SI:Gonzalez-Espinoza2020,SI:Park2020,SI:Serra-Peralta2021}. It was
curated to balance periods, composers, and melody lengths in three scripted stages
(\texttt{corpus\_build/}).
\textbf{Candidates:} only piano or harpsichord programs (GM 0--7); no vocal, ensemble,
organ, four-hand, or arrangement words in the title; at most two note-bearing channels or a
keyboard-genre title (keyboard fugues are often encoded one voice per channel); one file
per duplicated title. \textbf{Periods:} median year of composition per composer, binned at 1750, 1820, and 1890; Schubert and Hummel placed in
Classical. \textbf{Pieces:} multi-movement files split at movement-index markers
(\textit{(2)}, \textit{II}, \textit{No.~2}), else at a silence of $\ge 8$ beats coinciding
with a meter or key change, with $\ge 300$ notes per segment (Beethoven Op.~106 $\to$ 4,
Chopin Op.~25 $\to$ 12, Liszt's Sonata $\to$ 1). Movements share a work only in
multi-movement works; each piece of a collection is its own work. \textbf{Sample:}
melodies extracted with seeded \texttt{:hybrid}; $100\le L\le 1198$; $\le 2$ pieces per
work; composers with $\ge 10$ eligible pieces; equal quotas per period in six log-$L$ bins,
round-robin over composers, $\le 30$ per composer, fixed seed
(Table~\ref{tab:si-composers}). \textbf{Limitations:} heuristic classification; dating per composer; form labels unknown for 40\% of pieces;
Telemann, Mattheson, Gade, and Krebs end below ten pieces. \textbf{Curation check:} on an earlier,
unbalanced set of 1{,}131 solo piano pieces $\Delta H < 0$ in $99$--$99.6\%$ of pieces at
$m = 3$--$6$, as on the present corpus ($98.9$--$99.9\%$).

\begin{table}[h]
\centering\small
\caption{Composers (pieces).}
\label{tab:si-composers}
\begin{tabular}{l p{12.5cm}}
\toprule
Baroque & Bach 24, Handel 22, D.~Scarlatti 22, Couperin 21, d'Anglebert 20, Dandrieu 20,
Froberger 17, Corrette 12, Jollage 12, Jacquet de La Guerre 11, Rameau 11, Krebs 9,
Mattheson 8, Telemann 4 \\
Classical & Beethoven 30, Mozart 30, Haydn 28, Schubert 23, A.-L.~Couperin 22, Cramer 20,
Clementi 18, L.~Mozart 17, Hummel 14, Soler 11 \\
Romantic & Chopin 20, Alkan 19, Raff 19, Schumann 19, Brahms 18, Faur\'e 18, Liszt 18,
Tchaikovsky 17, Saint-Sa\"ens 13, Alb\'eniz 12, Dvo\v{r}\'ak 11, Grieg 11, Mendelssohn 10,
Gade 8 \\
20th century & Busoni 30, Debussy 30, Scriabin 30, Joplin 20, Pozzoli 17, Ravel 17,
Shostakovich 17, Satie 16, Bart\'ok 13, Godowsky 13, Jan\'a\v{c}ek 10 \\
\bottomrule
\end{tabular}
\end{table}

\begin{table}[htbp]
\centering
\caption{Solo keyboard corpus by period. ``Top'' is the share of the period's pieces by its
most represented composer; Schubert and Hummel are placed in the Classical period.}
\label{tab:si-corpus}
\small
\begin{tabular}{l c r r r c r}
\toprule
\textbf{Period} & \textbf{Dates} & \textbf{Composers} & \textbf{Pieces} & \textbf{Works} &
\textbf{$L$, median (IQR)} & \textbf{Top} \\
\midrule
Baroque      & before 1750 & 14 & 213 & 208 & 301 (186--577) & 11\% \\
Classical    & 1750--1819  & 10 & 213 & 193 & 322 (184--584) & 14\% \\
Romantic     & 1820--1889  & 14 & 213 & 209 & 317 (173--561) &  9\% \\
20th century & from 1890   & 11 & 213 & 212 & 327 (179--563) & 14\% \\
\midrule
Total        &             & 49 & 852 & 822 & 322 (179--571) &      \\
\bottomrule
\end{tabular}
\end{table}

\subsection{Prelude \& Fugue set}
\label{app:pf}

108 pieces: Bach WTC~I from OpenWTC in PDMX \citeSI{SI:openwtc,SI:long2024pdmx}, split at the
fugue's first measure (no clean split edition in Kunst der Fuge); Mendelssohn Op.~35 and
Shostakovich Op.~87 from Kunst der Fuge. Extraction is fixed to seeded \texttt{:hybrid}
(the method changes $\Hor$ by $\approx 0.8$ bits). In Op.~35 No.~6 (prelude) a 28-note
secondary track was removed (melody 15 $\to$ 292 notes); fugue results are unchanged.

\subsection{Census corpora}
\label{app:census}

Essen \citeSI{SI:essen1995} and MTC-ANN~2.0.1 \citeSI{SI:vankranenburg2016mtcann} are read through
MTCFeatures \citeSI{SI:vankranenburg2019mtcfeatures}; a phrase is a run of notes with one
phrase index, kept if $4\le L\le 16$ (1{,}890 and 44{,}765 phrases). MTD
\citeSI{SI:zalkow2020mtd}: 2{,}067 themes from MusicXML, top note under simultaneities, all
$L$ (no license stated; used as distributed, not redistributed).
Boundaries were set by people, Essen encoders mostly following text lines
\citeSI{SI:hothker2001segmentation}, MTC transcribers rhyme and syntax
\citeSI{SI:vankranenburg2014mtc}, and MTD editors complete themes.
 \textbf{Null:} for each segment, the exact probability $p_i$ that a random
rearrangement of its notes (or of its intervals, from the same first note) has a
non-trivial stabilizer under $D_{12}\times D_L$, by counting fixed arrangements per
prime-order element (first-order bound above $2\times10^6$ arrangements: 113 MTD segments,
overcount $<1\%$). Enrichment $O_L/E_L$ with $E_L=\sum_{L_i=L}p_i$, per exact $L$,
smoothed by a Gamma--Poisson model with parity-specific prior means (the null itself
favors even lengths, so expected counts are computed per exact length before pooling)
\citeSI{SI:clayton1987empirical}, shown where $E_L\ge 0.1$. The main-text figure pools exact lengths into bins
($4$--$6$, $7$--$9$, $10$--$12$, $13$--$16$, $17$--$24$, $25$--$41$), $O/E = \sum O_L/\sum
E_L$ with an exact Poisson interval \citeSI{SI:garwood1936fiducial}; pooling replaces the
smoothing, whose estimates fall below one at lengths with no symmetric segment and small $E_L$
(MTD, $L = 9$ and $11$).

\section{Melody extraction}

\subsection{Melody extraction methods}
\label{app:extraction}

Skyline \citeSI{SI:uitdenbogerd1998manipulation} keeps the highest note at each onset and truncates
it at the next onset. The revised variant \citeSI{SI:chai2001melody} first drops notes mostly
covered by other voices. \texttt{:hybrid} \citeSI{SI:ozcan2005melody} clusters the note groups (MIDI
tracks, MusicXML parts) by mean pitch with $k$-means, keeps the highest cluster, and applies
skyline. \texttt{:complexity} \citeSI{SI:madsen2007complexity} keeps the group with the largest
summed entropy of pitch, interval, and duration, regardless of register, and applies skyline. All
analyses use \texttt{:hybrid}, and \texttt{:complexity} is the check of
Section~\ref{app:extraction-check}. Lines are not transposed to a common key (e.g.\ with the Center
of Effect \citeSI{SI:chew2000spiral}). $\Hor$, $\eta$, and the symmetry types are invariant under a
global transposition, and a local key change alters at most $m - 1$ windows.

\subsection{Implementation}

The pipeline is implemented in Julia with \texttt{MelodicSymmetry.jl} and
\texttt{MelodyExtraction.jl}, developed for this study and tested by a sweep over every group
element (Theorem~\ref{thm:commute}), hand-computed stabilizers, and published worked examples
\citeSI{SI:ibragimova2025algebraic}. Stabilizer and canonical representative cost $O(48m^2)$ per
window ($48m$ group elements, $O(m)$ each), so a line of $N$ windows costs $O(48Nm^2)$; $K(m)$ is
exact in \texttt{BigInt} through Eq.~\eqref{eq:app-cycle}.

\subsection{Extraction consistency: \texttt{:hybrid} versus \texttt{:complexity}}
\label{app:extraction-check}
\textbf{Agreement.} Each piece was also reduced with \texttt{:complexity} and with plain skyline,
which agree by construction for the $393$ single-track keyboard pieces. In the $459$ multi-track
pieces \texttt{:hybrid} and \texttt{:complexity} return the identical line in $40\%$ (two tracks
$61\%$, four or more $21\%$, dances $63\%$, contrapuntal and free $27$--$30\%$), and in $47\%$ of
the Prelude \& Fugue set (fugues $43\%$, preludes $52\%$). Disagreeing selectors keep a different
track and share no notes in $73\%$ of cases.

\textbf{Distance between the two lines.} The Jensen--Shannon divergence of the orbit
distributions is scaled from $0$, between the two halves of the extracted line, to $1$, to another
piece by the same composer, at equal length. Where the selectors disagree, the other line lies at
$0.54$ at $m = 4$ (dances $0.86$, cyclic $0.64$, contrapuntal $-0.33$), at $-0.19$ to $0.02$ in
fugues ($m = 3$--$6$), whose voices share the subject, and at $0.69$--$0.75$ in preludes, where it
is the accompaniment. The scale resolves best at $m = 3$--$4$.

\textbf{Sensitivity.} Over multi-track pieces at $m = 4$ the median change is $0.05$ bits in
$\Hor$, $0.008$ in $\eta$, and $0.02$ bits in $H_{\text{stab}}$ (rank correlations between
selectors $0.77$, $0.81$, $0.63$), and $0.26$, $0.038$, and $0.10$ where the lines differ. These
are below the period differences in $\Hor$ ($0.40$ bits between medians) and the form differences
in $\eta$ ($0.12$), but comparable to the period differences in $\eta$ and $H_{\text{stab}}$ at
$m = 4$.

\textbf{Robustness.} Every main-text test was rerun on the $578$ keyboard and $51$ Prelude \&
Fugue pieces with identical lines, and on the \texttt{:complexity} line of every piece with IAAFT
nulls on the same seeds. In each reading $70$ of $74$ tests keep their conclusion and $92$--$97\%$
of pieces stay significantly below IAAFT. With \texttt{:complexity} lines the period trend in
$H_{\text{stab}}$ becomes significant at $m = 4$, the form difference in the IAAFT deficit at
$m = 6$, and the fugue trend in $H_{\text{stab}}$ at $m = 6$, while the form difference in
$H_{\text{stab}}$ at $m = 3$ is lost (Table~\ref{tab:si-separation}). On identical lines only
Prelude \& Fugue tests change, by loss of power (one Mendelssohn and two Shostakovich fugues keep
their line). The $\Delta\eta$ fugue trend and the Shostakovich fugue$-$prelude gap hold with
\texttt{:complexity} lines.

\section{Statistical procedures}
\label{app:stats}

Inference is rank- or permutation-based, with an effect size for every test. Pieces are the
observations, and dependence is respected by permuting or resampling whole clusters (composers,
songs, works) \citeSI{SI:hurlbert1984}. Table~\ref{tab:si-stats} lists the procedure behind each
result.

\begin{table}[htbp]
\centering\footnotesize
\caption{Procedure behind each result. KW: Kruskal--Wallis; JT: Jonckheere--Terpstra;
BH: Benjamini--Hochberg; $\delta$: Cliff's delta.}
\label{tab:si-stats}
\begin{tabular}{>{\raggedright\arraybackslash}p{2.3cm} >{\raggedright\arraybackslash}p{1.9cm}
>{\raggedright\arraybackslash}p{3.6cm} >{\raggedright\arraybackslash}p{2.0cm}
>{\raggedright\arraybackslash}p{2.3cm}}
\toprule
\textbf{Question} & \textbf{Data} & \textbf{Test} & \textbf{Multiplicity} & \textbf{Effect size} \\
\midrule
Lower $\Hor$ than nulls (\S\ref{sec:res-surrogates}) & piece & rank of $\Hor$ among 199 surrogates (two-sided); sign test across pieces & BH over pieces & $\Delta H$, $\Delta H/\overline{H}_{\text{null}}$ \\
Periods (\S\ref{sec:res-hor-eta}) & piece & KW; period labels permuted between composers (9{,}999); ordered trend as Spearman $\rho$, same permutation; measures $\Hor$, $\eta$, $H_{\text{stab}}$ & Holm over $m$ & $\varepsilon^2$; $\delta$ with composer-bootstrap 90\% CI \\
Forms (\S\ref{sec:res-hor-eta}) & labeled piece & KW; form labels permuted within composers (9{,}999); $\eta$ also as residual on $\log L$; $\eta$ and $\Delta\eta$ against IAAFT & Holm over $m$ & $\varepsilon^2$; $\delta$ \\
Where measures vary (\S\ref{sec:res-hor-eta}) & piece & adjusted $R^2$ of ranks on period, composer, form (descriptive) & --- & variance share \\
Composers, form fixed (\S\ref{sec:res-pf}) & piece & one-sided JT (50{,}000 permutations); exact Wilcoxon signed-rank for fugue$-$prelude pairs & Holm over $m = 4$--$6$ per measure & $\delta$ with 90\% stratified-bootstrap CI \\
Symmetry types (\S\ref{sec:res-genome}) & piece & class share ranked among 199 surrogates per null; sign test; period trend as Spearman $\rho$ with periods permuted between composers & BH over pieces & $O/E$ of class shares \\
Relations (\S\ref{sec:res-relations}) & piece & $\log_2 O/E$ against the within-orbit shuffle (Haldane--Anscombe); KW and Spearman trend with composers permuted between periods; composers permuted within periods; JT and Wilcoxon in the P\&F set; operation entropy per
within-orbit shuffle, period trend and JT on the interval deficit & Holm over $m$; BH over pieces & corpus $O/E$, composer-bootstrap 90\% CI; adjusted $R^2$ \\
Other entropy measures (\S\ref{app:entropies}) & piece & NSB per reading; periods and forms as above; JT on fugues; rank among 199 IAAFT surrogates & Holm over $m$; BH over pieces & partial $\rho$, adjusted $R^2$, share below / above surrogates \\
Extraction (\S\ref{app:extraction-check}) & piece & identity, Jaccard overlap, and orbit JSD against a noise floor and a different-music ceiling; every test rerun on pieces with identical lines and with \texttt{:complexity} lines & as in each test & share of tests with unchanged conclusion \\
Whole sequences (\S\ref{sec:res-census}) & segment & exact $p_i$ under note shuffle (interval shuffle as check) & --- & $O/E$, Garwood and song/work-bootstrap 95\% CI \\
\bottomrule
\end{tabular}
\end{table}

\textbf{Equivalence.} A non-significant difference is not evidence of equality. Similarity
claims are read from $\delta$ intervals against the negligible band $|\delta| < 0.147$
\citeSI{SI:romano2006ordinal}; intervals wider than the band (e.g.\ $\Delta H$, Bach versus
Shostakovich, $[-0.38, 0.19]$) are reported as no detectable difference.

\textbf{Reproducibility.} The script behind each table and figure is listed in the repository README
(Data accessibility). Surrogates are seeded from each file name; permutation seeds are fixed in each
script.

\section{Supplementary results}
\label{app:figs}

\subsection{Control sequences}

\begin{figure}[htbp]
\centering
\includegraphics[width=\linewidth]{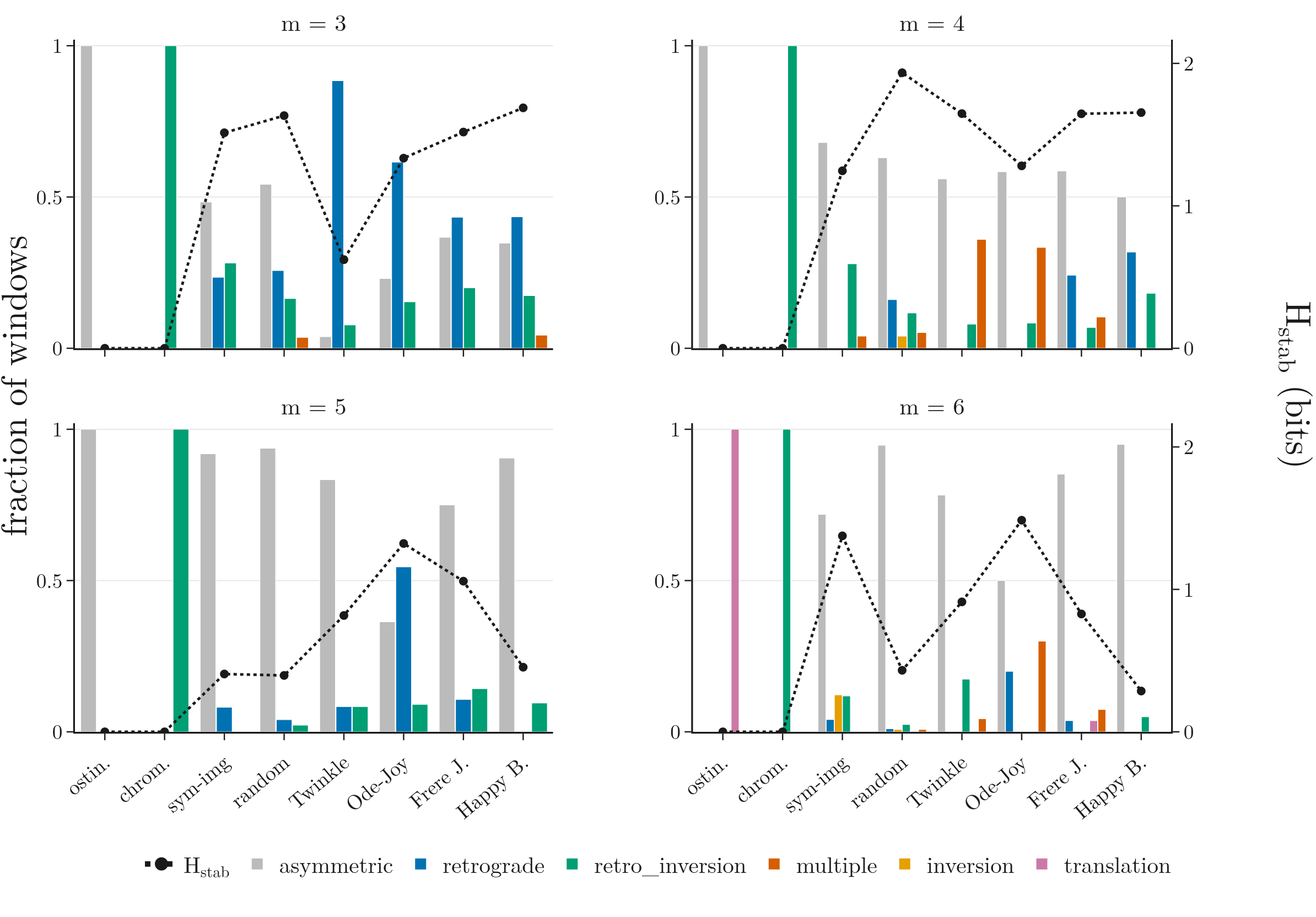}
\caption{Window stabilizer-type profiles of the annotated melodies and controls at
$m = 3$ and $m = 5$. At $m = 3$ short palindromes dominate; at $m = 5$ the types separate
return figures (retrograde), scalar or mirror motion (RI) and directed subjects
(asymmetric).}
\label{fig:si-window}
\end{figure}

\subsection{Orbit structure against surrogates}

\textbf{Real versus null lines.} With BH control over the $852$ pieces, $93$--$98\%$ are
significantly below the shuffle and IAAFT nulls at $m = 3$--$6$, and $26$, $72$, $83$, and
$87\%$ below Markov-1 \citeSI{SI:benjamini1995fdr}. The full ordering of
Eq.~\eqref{eq:rungs} holds piece by piece in $87$--$90\%$ of pieces at $m \ge 4$ and in $36\%$ at
$m = 3$, where the Markov null nearly reproduces the real line (Figure~\ref{fig:si-nulls}).

\begin{figure}[htbp]
\centering
\includegraphics[width=\linewidth]{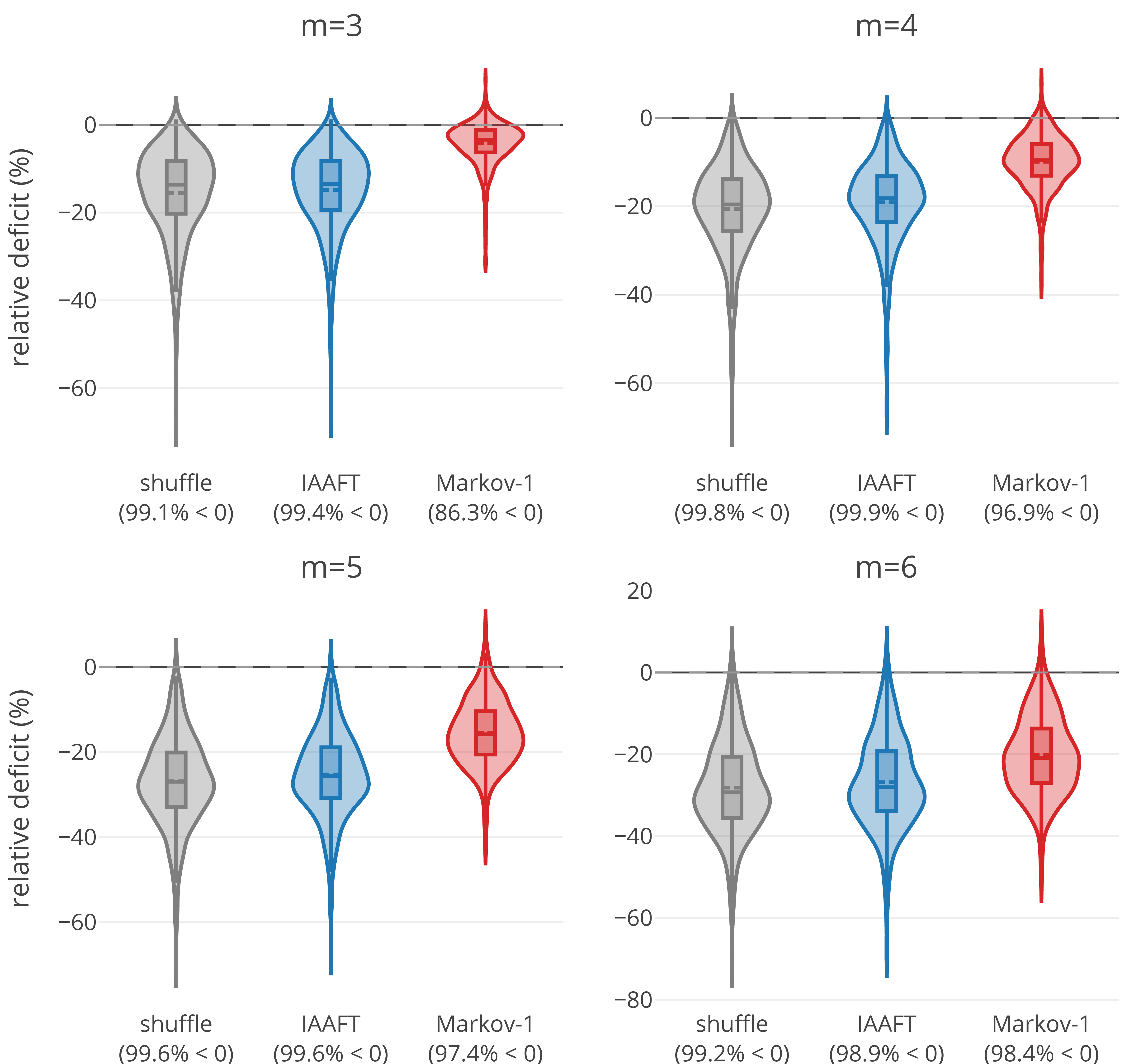}
\caption{Relative deficit $(\Hor - \langle H_{\text{null}}\rangle)/\langle H_{\text{null}}\rangle$
of the 852 corpus pieces under shuffle, IAAFT, and Markov-1 surrogates (199 each), $m = 3$--$6$;
labels give the share of pieces below their null.}
\label{fig:si-nulls}
\end{figure}

\subsection{Vocabulary across periods, deployment across forms}

\textbf{Separation tests} (Table~\ref{tab:si-separation}; panels not in
Figure~\ref{fig:forms-periods} in Figure~\ref{fig:si-grid}). Forms differ in the deficit and in
$\Hor$ only when pieces are treated as independent, and lose the effect once labels are permuted
within composers. The form effect on $\eta$ survives at every $m$, with or without $\log L$ removed.

\begin{table}[htbp]
\centering\footnotesize
\caption{Separation of periods and forms for each measure. Periods: Kruskal--Wallis with
period labels permuted between composers (Holm over $m$), $\varepsilon^2$, and the ordered trend
Baroque $\to$ 20th century as Spearman $\rho$ with composers permuted as blocks (Holm $p$ in
parentheses). Forms: Kruskal--Wallis with form labels permuted within composers (Holm over $m$),
$\varepsilon^2$. Bold: $p < 0.05$. Pairwise $\delta$ in \texttt{Data/cluster\_tests.csv}. $^\dagger$Not retained with \texttt{:complexity} lines (Section~\ref{app:extraction-check}).}
\label{tab:si-separation}
\begin{tabular}{l c ccc cc}
\toprule
 & & \multicolumn{3}{c}{Periods} & \multicolumn{2}{c}{Forms} \\
\cmidrule(lr){3-5}\cmidrule(lr){6-7}
Measure & $m$ & $p$ & $\varepsilon^2$ & trend $\rho$ ($p$) & $p$ & $\varepsilon^2$ \\
\midrule
\multirow{4}{*}{$\Hor$ deficit vs IAAFT} & 3 & \textbf{0.003} & 0.06 & \textbf{+0.19 (0.010)} & 1.000 & 0.00 \\
 & 4 & 1.000 & 0.02 & +0.10 (0.599) & 1.000 & 0.01 \\
 & 5 & 1.000 & 0.01 & +0.09 (0.599) & 1.000 & 0.05 \\
 & 6 & 1.000 & 0.00 & +0.05 (0.599) & 0.175 & 0.10 \\
\addlinespace
\multirow{4}{*}{$\Hor$} & 3 & \textbf{0.002} & 0.09 & \textbf{+0.22 (0.008)} & 1.000 & 0.00 \\
 & 4 & 0.191 & 0.04 & +0.16 (0.161) & 1.000 & 0.01 \\
 & 5 & 0.265 & 0.04 & +0.14 (0.182) & 1.000 & 0.04 \\
 & 6 & 0.498 & 0.02 & +0.10 (0.239) & 0.573 & 0.10 \\
\addlinespace
\multirow{4}{*}{$\eta$} & 3 & 1.000 & -0.00 & +0.03 (1.000) & \textbf{0.003} & 0.06 \\
 & 4 & 1.000 & 0.00 & -0.04 (1.000) & \textbf{$<0.001$} & 0.08 \\
 & 5 & 0.698 & 0.02 & -0.08 (1.000) & \textbf{$<0.001$} & 0.07 \\
 & 6 & 1.000 & 0.01 & -0.02 (1.000) & \textbf{0.001} & 0.03 \\
\addlinespace
\multirow{4}{*}{$H_{\text{stab}}$} & 3 & \textbf{0.005} & 0.06 & \textbf{+0.18 (0.015)} & \textbf{0.010}$^\dagger$ & 0.04 \\
 & 4 & 0.868 & 0.00 & +0.08 (0.288) & 0.442 & 0.02 \\
 & 5 & \textbf{0.008} & 0.05 & \textbf{+0.22 (0.002)} & 0.442 & 0.03 \\
 & 6 & 0.868 & 0.01 & +0.07 (0.288) & 0.442 & 0.02 \\
\addlinespace
\multirow{4}{*}{$\eta$, length removed} & 3 & --- & --- & --- & \textbf{$<0.001$} & 0.14 \\
 & 4 & --- & --- & --- & \textbf{$<0.001$} & 0.11 \\
 & 5 & --- & --- & --- & \textbf{$<0.001$} & 0.08 \\
 & 6 & --- & --- & --- & \textbf{$<0.001$} & 0.03 \\
\addlinespace
\multirow{4}{*}{$\Delta\eta$ vs IAAFT} & 3 & 1.000 & 0.00 & --- & \textbf{$<0.001$} & 0.11 \\
 & 4 & 1.000 & 0.01 & --- & \textbf{$<0.001$} & 0.13 \\
 & 5 & 1.000 & 0.00 & --- & \textbf{$<0.001$} & 0.08 \\
 & 6 & 1.000 & 0.00 & --- & \textbf{$<0.001$} & 0.03 \\
\bottomrule
\end{tabular}
\end{table}

\textbf{Variance shares} (Table~\ref{tab:si-variance}). Composer has the largest share for every
measure. $\Delta\eta$ (Table~\ref{tab:si-deta}) is nearly length-free only at $m = 4$
(Section~\ref{app:length}). Within each period dances have the lowest $\eta$
(Table~\ref{tab:si-formperiod}).

\begin{table}[htbp]
\centering\footnotesize
\caption{Where the measures vary: share of rank variance explained (adjusted $R^2$ of a
least-squares fit on ranks) by period, composer, composer beyond period, form (labeled pieces),
and form beyond composer. Descriptive; composers are nested in periods.}
\label{tab:si-variance}
\begin{tabular}{l c ccccc}
\toprule
Measure & $m$ & period & composer & composer $|$ period & form & form $|$ composer \\
\midrule
\multirow{4}{*}{$\Hor$ deficit} & 3 & 0.06 & 0.17 & 0.10 & 0.00 & 0.02 \\
 & 4 & 0.02 & 0.22 & 0.21 & 0.01 & 0.00 \\
 & 5 & 0.01 & 0.26 & 0.25 & 0.05 & -0.00 \\
 & 6 & 0.00 & 0.30 & 0.30 & 0.10 & 0.00 \\
\addlinespace
\multirow{4}{*}{$\Hor$} & 3 & 0.09 & 0.21 & 0.13 & 0.00 & 0.02 \\
 & 4 & 0.04 & 0.26 & 0.22 & 0.01 & 0.01 \\
 & 5 & 0.04 & 0.28 & 0.25 & 0.04 & 0.01 \\
 & 6 & 0.02 & 0.32 & 0.30 & 0.10 & 0.01 \\
\addlinespace
\multirow{4}{*}{$\eta$} & 3 & -0.00 & 0.21 & 0.21 & 0.06 & 0.02 \\
 & 4 & 0.00 & 0.20 & 0.19 & 0.08 & 0.02 \\
 & 5 & 0.02 & 0.20 & 0.18 & 0.07 & 0.03 \\
 & 6 & 0.01 & 0.18 & 0.17 & 0.03 & 0.02 \\
\addlinespace
\multirow{4}{*}{$\eta$, length removed} & 3 & 0.00 & 0.23 & 0.23 & 0.14 & 0.03 \\
 & 4 & 0.02 & 0.20 & 0.18 & 0.11 & 0.04 \\
 & 5 & 0.05 & 0.19 & 0.14 & 0.08 & 0.05 \\
 & 6 & 0.02 & 0.15 & 0.13 & 0.03 & 0.04 \\
\addlinespace
\multirow{4}{*}{$H_{\text{stab}}$} & 3 & 0.06 & 0.14 & 0.08 & 0.04 & 0.03 \\
 & 4 & 0.00 & 0.11 & 0.11 & 0.02 & 0.01 \\
 & 5 & 0.05 & 0.18 & 0.13 & 0.03 & -0.00 \\
 & 6 & 0.01 & 0.12 & 0.12 & 0.02 & 0.00 \\
\addlinespace
\multirow{4}{*}{$\Delta\eta$} & 3 & 0.00 & 0.19 & 0.19 & 0.11 & 0.04 \\
 & 4 & 0.01 & 0.18 & 0.17 & 0.13 & 0.04 \\
 & 5 & 0.00 & 0.16 & 0.16 & 0.08 & 0.04 \\
 & 6 & 0.00 & 0.18 & 0.17 & 0.03 & 0.03 \\
\bottomrule
\end{tabular}
\end{table}

\begin{table}[htbp]
\centering\small
\caption{Deployment ratio against $199$ IAAFT surrogates per piece: medians of the real
$\eta$, the null mean, and $\Delta\eta$; percentage of pieces significantly above / below the
null after Benjamini--Hochberg control; Spearman $\rho$ of $\eta$ and $\Delta\eta$ with $\log L$.}
\label{tab:si-deta}
\begin{tabular}{c ccc c cc}
\toprule
$m$ & $\eta$ & $\overline{\eta}_{\text{null}}$ & $\Delta\eta$ & above / below (\%) & $\rho(\eta)$ & $\rho(\Delta\eta)$ \\
\midrule
3 & 0.501 & 0.576 & -0.071 & 4 / 80 & +0.76 & -0.38 \\
4 & 0.284 & 0.265 & +0.016 & 42 / 27 & +0.76 & -0.15 \\
5 & 0.170 & 0.096 & +0.063 & 76 / 3 & +0.66 & +0.35 \\
6 & 0.107 & 0.044 & +0.058 & 81 / 2 & +0.58 & +0.54 \\
\bottomrule
\end{tabular}
\end{table}

\begin{table}[htbp]
\centering\footnotesize
\caption{Deployment ratio $\eta$ by form within each period: Cliff's $\delta$ of the other
labeled forms against dances (positive: dances lower), with the $p$-value of a Kruskal--Wallis
test over all forms with form labels permuted within composers (9{,}999 permutations). Bold:
$p < 0.05$.}
\label{tab:si-formperiod}
\setlength{\tabcolsep}{3pt}
\begin{tabular}{l c c cccc}
\toprule
Period & \shortstack{composers\\ with $\ge 2$ forms} & dances & $m = 3$ & $m = 4$ & $m = 5$ & $m = 6$ \\
\midrule
Baroque & 5 & 19 & +0.36 (0.240) & +0.36 (0.180) & +0.43 (0.212) & +0.32 (0.098) \\
Classical & 8 & 23 & \textbf{+0.72 ($<0.001$)} & \textbf{+0.47 (0.011)} & \textbf{+0.41 (0.035)} & +0.17 (0.119) \\
Romantic & 10 & 12 & +0.05 (0.657) & +0.16 (0.672) & +0.20 (0.541) & +0.23 (0.655) \\
20th century & 5 & 11 & +0.35 (0.314) & \textbf{+0.56 (0.043)} & \textbf{+0.62 (0.015)} & \textbf{+0.67 (0.009)} \\
\bottomrule
\end{tabular}
\end{table}

\begin{figure}[htbp]
\centering
{\small\textbf{(a)}}\\[2pt]
\includegraphics[width=\linewidth]{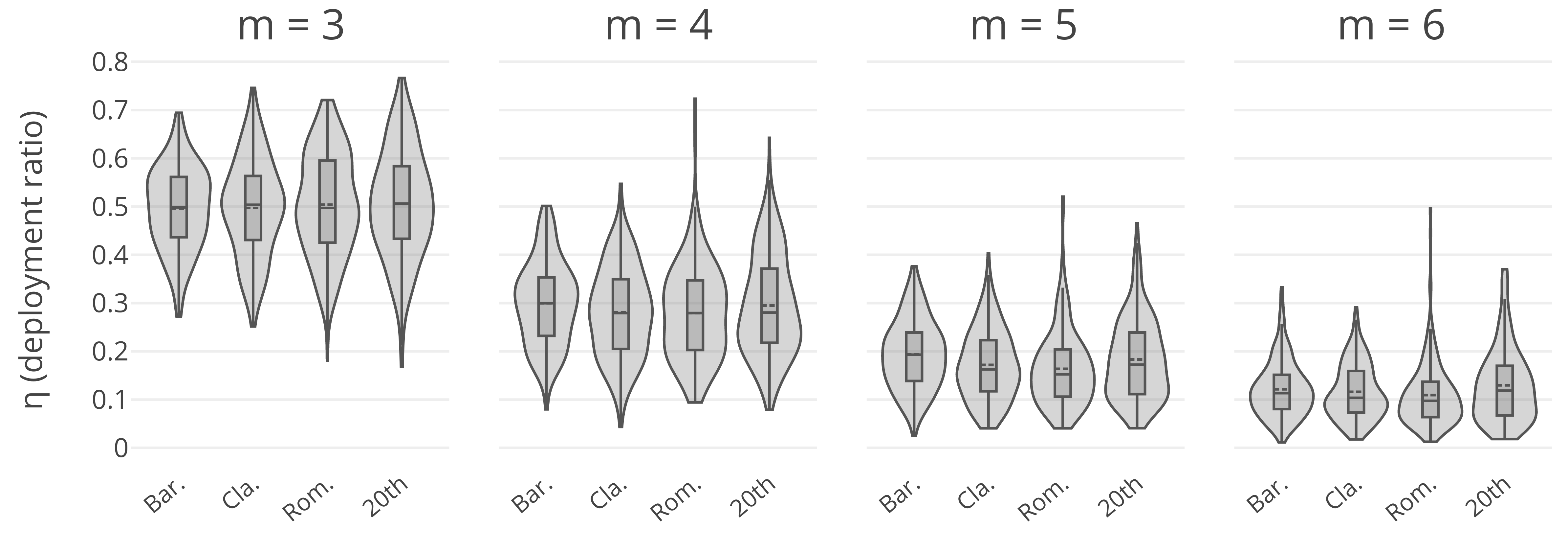}\\[4pt]
{\small\textbf{(b)}}\\[2pt]
\includegraphics[width=\linewidth]{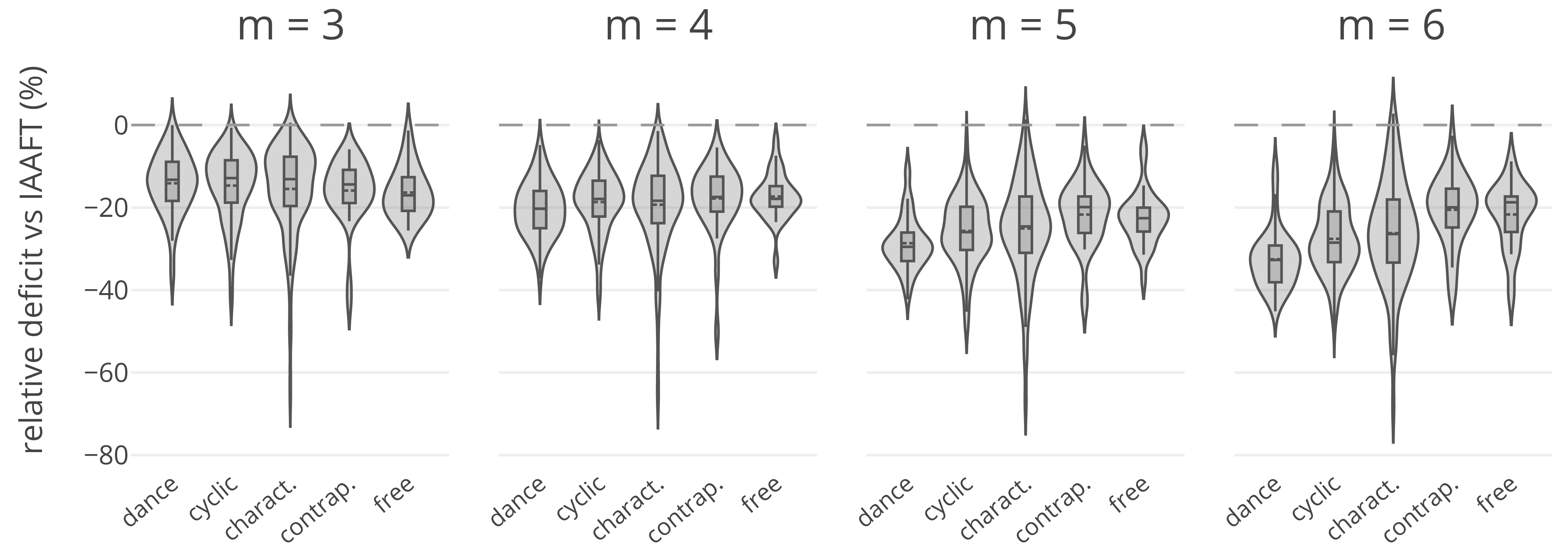}\\[4pt]
{\small\textbf{(c)}}\\[2pt]
\includegraphics[width=\linewidth]{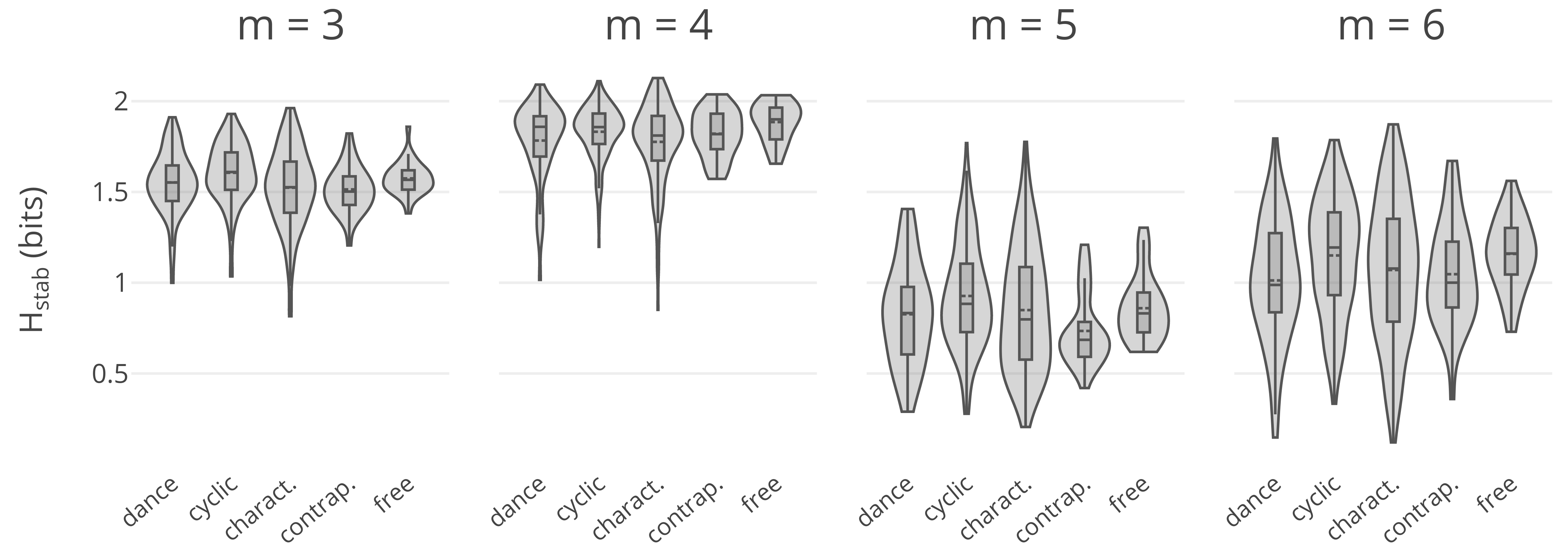}
\caption{Panels complementing Figure~\ref{fig:forms-periods}, $m = 3$--$6$. \textbf{(a)} $\eta$ by
period ($213$ pieces each). \textbf{(b)} Relative deficit of $\Hor$ against IAAFT and
\textbf{(c)} $H_{\text{stab}}$ by form family ($508$ labeled pieces). Tests in
Table~\ref{tab:si-separation}.}
\label{fig:si-grid}
\end{figure}

\subsection{Deployment in preludes and fugues}

Tests at $m = 4$ are in Table~\ref{tab:si-pf}, $\eta$ at every $m$ in Figure~\ref{fig:si-pf-eta},
and paired values in Figure~\ref{fig:si-pf-pairs}. At $m = 5$--$6$ $\Delta\eta$ still depends on
length (Section~\ref{app:length}), so these sizes point the same way but carry no claims.
$H_{\text{stab}}$ is taken over the six classes of Table~\ref{tab:window-types}; over stabilizer
conjugacy classes the fugue trend is the same (Holm $p = 0.018$ at $m = 4$).

\begin{table}[htbp]
\centering\small
\caption{Prelude \& Fugue tests at $m = 4$. Fugue trend: one-sided Jonckheere--Terpstra, Bach $<$
Mendelssohn $<$ Shostakovich, Holm over $m = 4$--$6$. $\delta$: Cliff's delta with $90\%$
stratified-bootstrap interval. Paired: exact Wilcoxon signed-rank on fugue$-$prelude $\Delta\eta$.}
\label{tab:si-pf}
\begin{tabular}{l l}
\toprule
Test & Result \\
\midrule
Fugue trend, $p$ & $\Hor$ 1.0;\; $\Delta H$ 1.0;\; $H_{\text{stab}}$ \textbf{0.014};\; $\eta$ \textbf{0.002};\; $\Delta\eta$ \textbf{0.035} \\
$\delta$ of fugue $\Delta\eta$ & Bach$\to$Shostakovich $0.38$ $[0.12, 0.62]$;\; Bach$\to$Mendelssohn $-0.64$ $[-0.86, -0.36]$ \\
Paired fugue$-$prelude, $p$ & Bach 1.0;\; Mendelssohn 0.94;\; Shostakovich \textbf{0.002} \\
Interaction trend, $p$ & \textbf{0.003} \\
$\rho$ with fugue/prelude length ratio & $0.14$ ($p = 0.30$) \\
\bottomrule
\end{tabular}
\end{table}

\subsubsection{Length dependence of the deployment ratio}
\label{app:length}

$\eta$ is a plug-in estimate and grows with melody length: in the Prelude \& Fugue set
$\rho(\eta,\log L) = 0.79$, $0.73$, $0.69$ at $m = 4$--$6$, also within composers
(Figure~\ref{fig:si-length}); $\Delta\eta$ is length-free only at $m = 4$ ($\rho = 0.06$).
Hence the length-matched corpus, on which period differences in $\eta$ largely vanish
($m = 4$ medians $0.28$--$0.30$), while $\Hor$ and $\Delta H$ stay near length-flat
($|r| = 0.12$, $0.10$).

\begin{figure}[h]
\centering
\includegraphics[width=\linewidth]{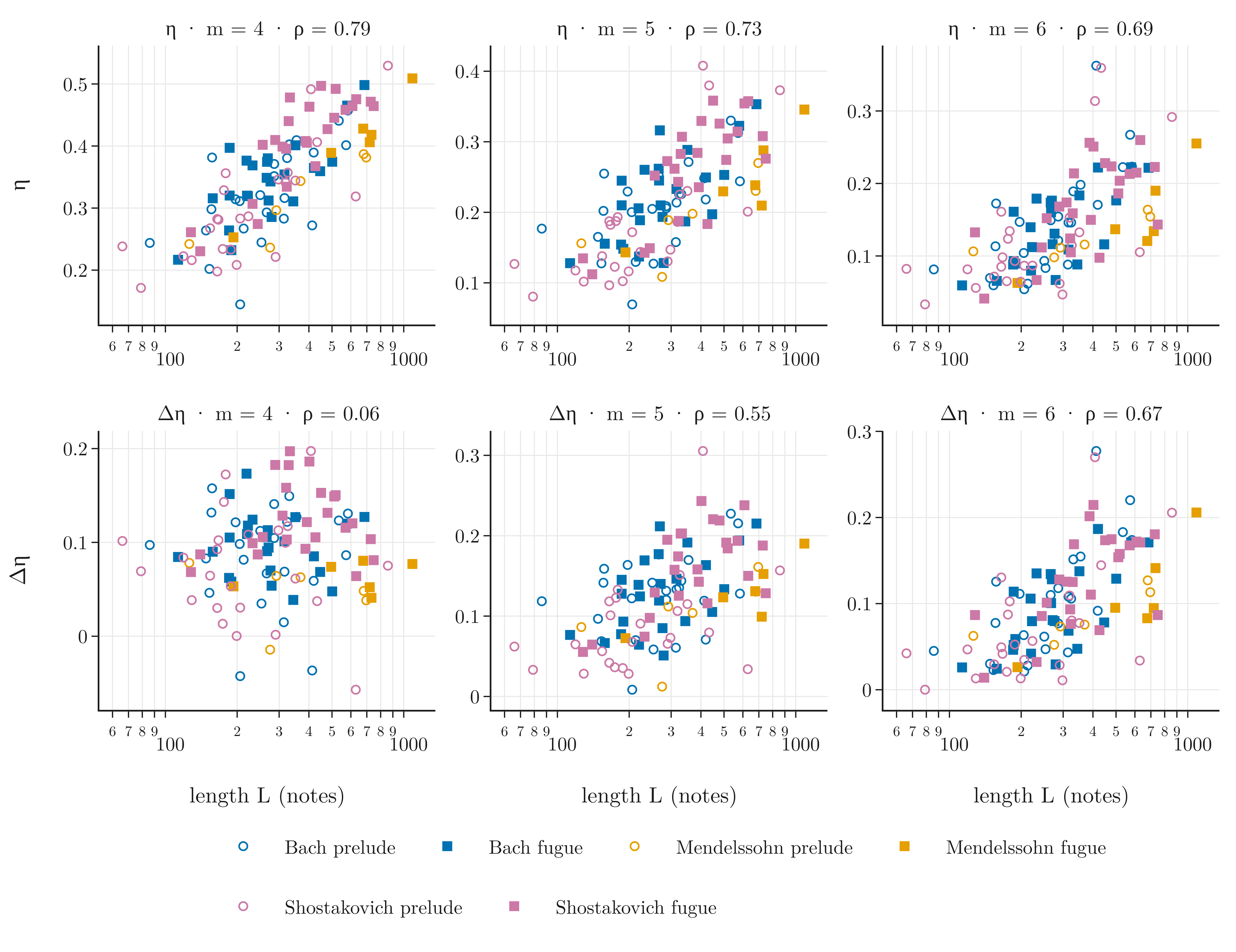}
\caption{$\eta$ (top) and $\Delta\eta$ (bottom) against melody length, Prelude \& Fugue
set, $m = 4$--$6$; $\rho$: Spearman correlation with $\log L$.}
\label{fig:si-length}
\end{figure}

\begin{figure}[htbp]
\centering
\includegraphics[width=\linewidth]{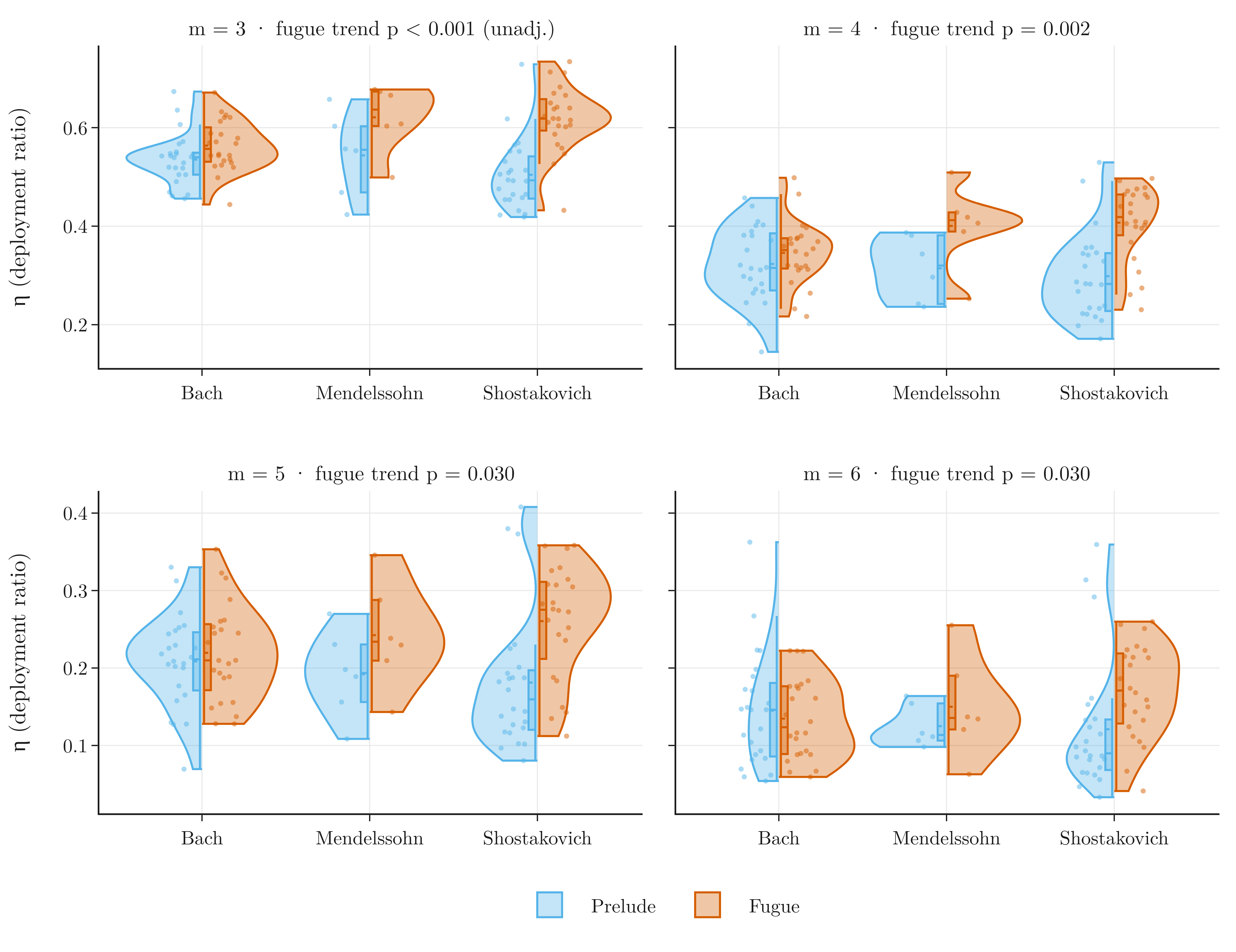}
\caption{$\eta$ of preludes and fugues by composer, $m = 3$--$6$; titles give the
Jonckheere--Terpstra fugue trend (Holm over $m = 4$--$6$).}
\label{fig:si-pf-eta}
\end{figure}

\begin{figure}[htbp]
\centering
\includegraphics[width=\linewidth]{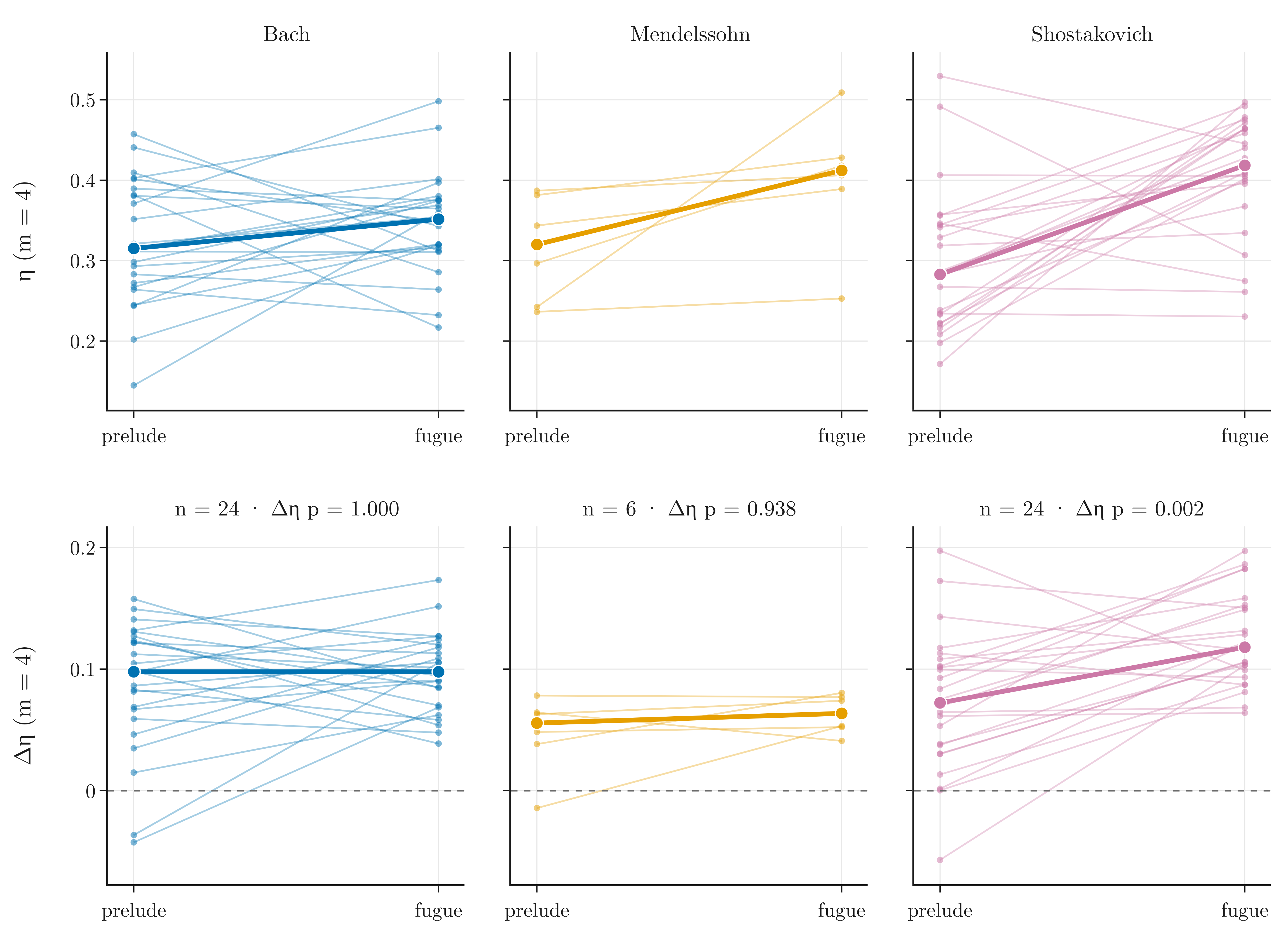}
\caption{Paired prelude and fugue of each number, $\eta$ (top) and $\Delta\eta$ (bottom),
$m = 4$; thick line: medians; $p$: exact Wilcoxon signed-rank test on $\Delta\eta$.}
\label{fig:si-pf-pairs}
\end{figure}

\subsection{Symmetry types of short windows}

Counting baseline over all $12^m$ windows for Figure~\ref{fig:genome-periods}a: Table~\ref{tab:si-baseline}.

\begin{table}[htbp]
\centering\small
\caption{Share of windows per stabilizer class (\%): uniform alphabet (u, all $12^m$
windows) and the range of the four period means in the corpus (real).}
\label{tab:si-baseline}
\begin{tabular}{l rr rr rr rr}
\toprule
 & \multicolumn{2}{c}{$m = 3$} & \multicolumn{2}{c}{$m = 4$} & \multicolumn{2}{c}{$m = 5$} & \multicolumn{2}{c}{$m = 6$} \\
\cmidrule(lr){2-3}\cmidrule(lr){4-5}\cmidrule(lr){6-7}\cmidrule(lr){8-9}
class & u & real & u & real & u & real & u & real \\
\midrule
asymmetric      & 58.3 & 41.3--47.7 & 63.9 & 42.2--47.3 & 93.2 & 74.4--82.4 & 94.9 & 70.7--76.2 \\
retrograde      & 20.8 & 33.3--35.4 & 13.9 & 24.0--27.2 & 3.4  & 13.0--17.3 & 2.2  & 11.6--14.0 \\
retro-inversion & 16.7 & 16.2--16.8 & 12.5 & 12.0--13.1 & 3.4  & 4.1--5.5   & 2.1  & 5.2--5.8 \\
inversion       & 0    & 0          & 4.6  & 2.0--2.6   & 0    & 0          & 0.5  & 0.3--0.4 \\
translation     & 0    & 0          & 0    & 0          & 0    & 0          & 0.1  & 0.7--1.1 \\
multiple        & 4.2  & 2.3--6.7   & 5.1  & 12.7--17.7 & 0.1  & 0.5--2.8   & 0.3  & 4.0--8.4 \\
\bottomrule
\end{tabular}
\end{table}

\textbf{Against nulls.} Class shares against each piece's surrogates are in
Table~\ref{tab:si-classnull}. The retrograde excess over shuffled lines grows with $m$ and nearly
vanishes against Markov-1. The multiple class (repeated notes, oscillations) exceeds every null,
and its rise across periods holds with composers permuted as blocks ($\log O/E$ against Markov-1,
$\rho = 0.19$, $p = 0.002$ at $m = 4$, and $\rho = 0.22$, $p < 0.001$ at $m = 6$). Inversion
windows are below every null.

\begin{table}[htbp]
\centering\small
\caption{Observed over expected share of windows per symmetry class (pieces weighted
equally), against 199 surrogates per piece.}
\label{tab:si-classnull}
\begin{tabular}{l rr rr rr}
\toprule
 & \multicolumn{2}{c}{shuffle} & \multicolumn{2}{c}{IAAFT} & \multicolumn{2}{c}{Markov-1} \\
\cmidrule(lr){2-3}\cmidrule(lr){4-5}\cmidrule(lr){6-7}
Class & $m = 4$ & $m = 6$ & $m = 4$ & $m = 6$ & $m = 4$ & $m = 6$ \\
\midrule
retrograde & 1.20 & 2.36 & 1.07 & 1.74 & 0.99 & 1.16 \\
RI         & 1.10 & 2.25 & 1.21 & 2.11 & 1.10 & 1.53 \\
inversion  & 0.59 & 0.77 & 0.67 & 0.77 & 0.58 & 0.86 \\
multiple   & 2.12 & 8.18 & 1.65 & 4.44 & 1.12 & 1.80 \\
\bottomrule
\end{tabular}
\end{table}

\textbf{Symmetry-type entropy.} $H_{\text{stab}}$ (Eq.~\eqref{eq:hstab}) is at or slightly below
the nulls at $m = 3$--$4$ and above every null from $m = 5$ (Table~\ref{tab:si-hstabnull}). It rises
across periods at $m = 3$ and $5$ only, possibly because at odd $m$ the multiple class holds almost
only repeated notes, whereas at even $m$ two-note oscillations, common in every period, also enter
it (not tested). Its correlation with $\log L$ ($\rho = 0.13$--$0.35$) does not affect the
length-matched periods but weakens the comparison of forms.

\begin{table}[htbp]
\centering\small
\caption{$H_{\text{stab}}$ of real lines against 199 surrogates per piece: median of
$H_{\text{stab}} - \overline{H}_{\text{stab,null}}$ in bits, with the percentage of pieces
significantly above / below the null after Benjamini--Hochberg control in parentheses.}
\label{tab:si-hstabnull}
\begin{tabular}{l cccc}
\toprule
Null & $m = 3$ & $m = 4$ & $m = 5$ & $m = 6$ \\
\midrule
shuffle & -0.06 (15/29) & +0.10 (34/9) & +0.22 (49/3) & +0.56 (80/1) \\
IAAFT & -0.09 (12/36) & +0.07 (23/9) & +0.10 (33/11) & +0.45 (70/2) \\
Markov-1 & -0.01 (0/0) & -0.06 (4/20) & +0.07 (20/2) & +0.19 (39/3) \\
\bottomrule
\end{tabular}
\end{table}

\subsection{Transformations between recurring motifs}
\label{app:relations}

\textbf{Definition.} The typing map $(a, k, b, j) \mapsto (a, b)$ is a homomorphism of $G$
onto $\Z_2 \times \Z_2$, so the four non-repetition types compose like the elements of the
Klein four-group, and the rotated flag $j \ne 0$ is read separately. The members $h_0 s$,
$s \in \Stab(w_p)$, of a relation coset (Eq.~\eqref{eq:relation}) share one type only when
the stabilizer holds no inversion or retrograde, which is why profiles are taken over pairs of
asymmetric windows. If the windows of an orbit take forms $g_p^{-1} \cdot \hat o$, their
relations are $g_q^{-1} g_p$, so the relation distribution of a line is the autocorrelation of
its form distribution over $G$ \citeSI{SI:diaconis1988group}. A line with $\eta = 0$ relates
each motif only to itself, and a line with $\eta = 1$ spreads its relations evenly, so $\eta$
bounds the spread of relations but does not fix their types.

\textbf{Nulls.} The within-orbit shuffle permutes, within each orbit, the forms of its windows, which
keeps every orbit, the number of windows in each form, $\Hor$, $\eta$, and $H_{\text{stab}}$.
Under uniform forms each window takes a form drawn uniformly from its orbit. Its expectation is
analytic, a rotated share of $(m - 1)/m$ (observed shares in Figure~\ref{fig:si-rel-rotated}) and,
among unrotated relations, $1/48$ repetition, $11/48$ transposition, and $12/48$ each of inversion,
retrograde, and retrograde-inversion. The surrogates are those of Section~\ref{sec:nulls}, with the
same seeds. Profiles over all windows dilute retrograde-inversion (Figure~\ref{fig:si-rel-leakage}).

\textbf{Model.} The corpus $O/E$ of a type is the mean observed over the mean expected share,
pieces weighted equally, with a $90\%$ composer-bootstrap interval (Table~\ref{tab:si-relations}).
Per piece, $\log_2 O/E$ uses pair counts plus $0.5$
\citeSI{SI:haldane1956estimation,SI:anscombe1956estimating}, since $4$--$10\%$ of pieces lack a
given type and dropping them would bias comparisons upward. Periods are compared by
Kruskal--Wallis and the ordered Spearman trend with composers permuted between periods, and
composers by permuting pieces between composers of one period, Holm-adjusted over $m = 3$--$6$.
Composer means are shrunk toward the period mean by a normal--normal empirical-Bayes model with
moment-estimated between-composer variance \citeSI{SI:morris1983parametric}, which is zero for
Romantic composers at $m = 4$ (Figures~\ref{fig:relations-composers} and~\ref{fig:si-rel-composers},
open markers unshrunk).

\textbf{Results} (Table~\ref{tab:si-relations}; Figures~\ref{fig:si-rel-types}--\ref{fig:si-rel-lags}).
Composers differ within periods for transposition and inversion at every $m$
(Holm $p < 0.001$), for retrograde-inversion at $m = 5$--$6$, and for retrograde at $m = 6$.
Composer explains $7$--$17\%$ of the rank variance of transposition and inversion, period at most
$2\%$. Across periods only retrograde-inversion differs, declining from the Baroque to the 20th
century ($\rho = -0.12$ to $-0.24$, Holm $p \le 0.03$), as does the share of fifths among
transpositions ($p < 0.001$ at every $m$; Figure~\ref{fig:si-rel-intervals}). In the fugues the fifth share falls from Bach to
Shostakovich (Jonckheere--Terpstra, Holm $p < 0.001$ at $m = 4$--$6$). Paired fugue$-$prelude
differences are not significant except for transposition in Shostakovich at $m = 6$ (Holm
$p = 0.04$). After Benjamini--Hochberg control few single pieces depart from the within-orbit
shuffle (none at $m = 4$, $31\%$ above for transposition at $m = 5$), so the results rest on
pooling pieces.

\begin{table}[htbp]
\centering\small
\caption{Relation types against the within-orbit shuffle, pairs of asymmetric windows: corpus $O/E$
with composer-bootstrap $90\%$ interval and, in parentheses, the percentage of pieces above
the null.}
\label{tab:si-relations}
\setlength{\tabcolsep}{4pt}
\begin{tabular}{l cccc}
\toprule
$m$ & transposition & inversion & retrograde & retrograde-inversion \\
\midrule
3 & 1.10 [1.06, 1.14] (54) & 0.90 [0.88, 0.92] (41) & 1.00 [0.98, 1.02] (54) & 1.21 [1.18, 1.23] (79) \\
4 & 1.56 [1.47, 1.66] (75) & 0.84 [0.80, 0.88] (36) & 0.97 [0.94, 1.01] (46) & 1.33 [1.28, 1.39] (71) \\
5 & 2.00 [1.86, 2.15] (76) & 0.75 [0.70, 0.80] (30) & 1.06 [1.01, 1.11] (46) & 1.38 [1.31, 1.46] (66) \\
6 & 3.17 [2.92, 3.44] (73) & 0.57 [0.51, 0.64] (22) & 0.84 [0.77, 0.92] (30) & 1.17 [1.06, 1.29] (48) \\
\bottomrule
\end{tabular}
\end{table}

\begin{figure}[htbp]
\centering
\includegraphics[width=\linewidth]{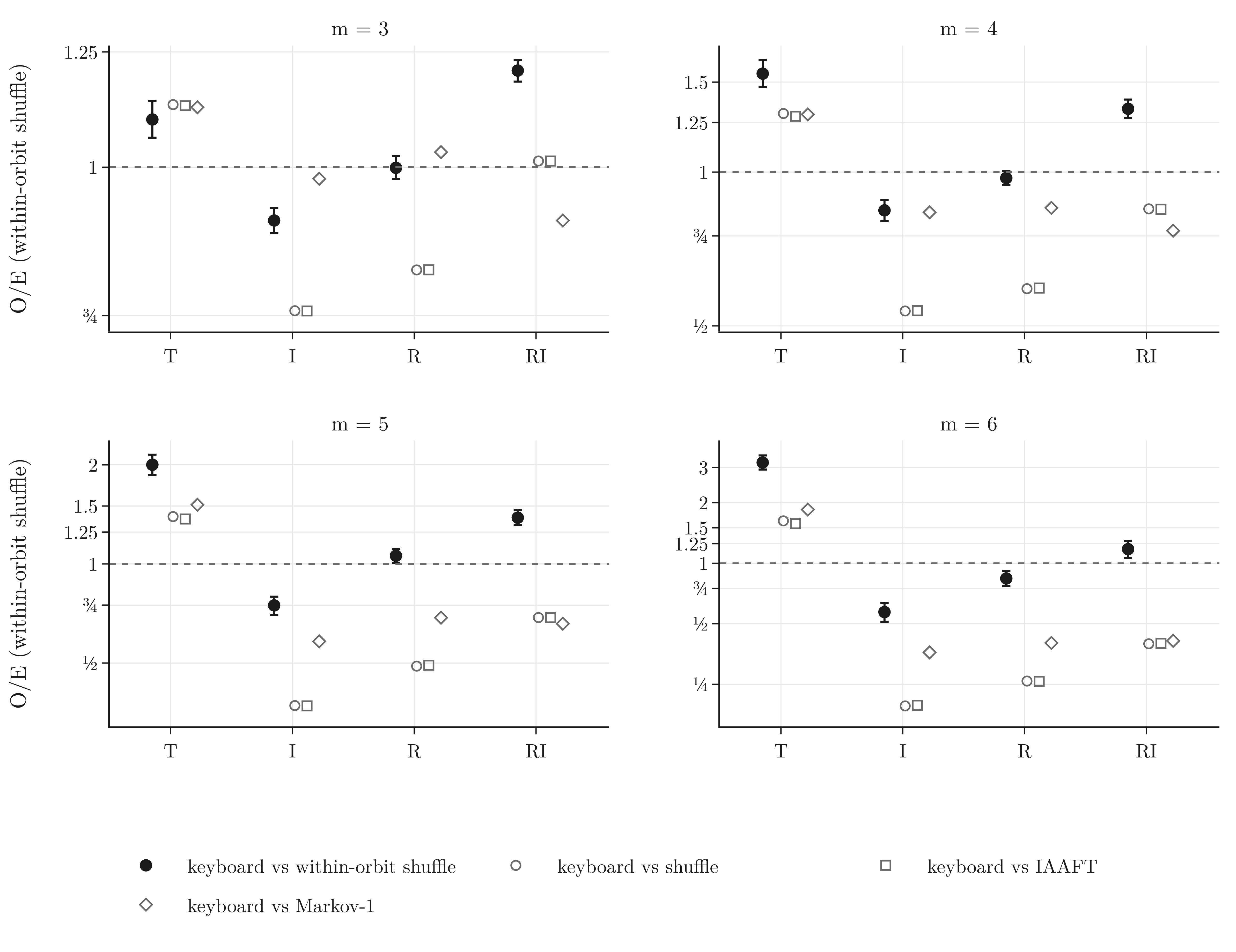}
\caption{Relation types relative to the within-orbit shuffle (filled, $90\%$ composer-bootstrap
interval) and to the shuffle, IAAFT, and Markov-1 surrogates (open), $m = 3$--$6$, pairs of
asymmetric windows.}
\label{fig:si-rel-types}
\end{figure}

\begin{figure}[htbp]
\centering
\includegraphics[width=\linewidth]{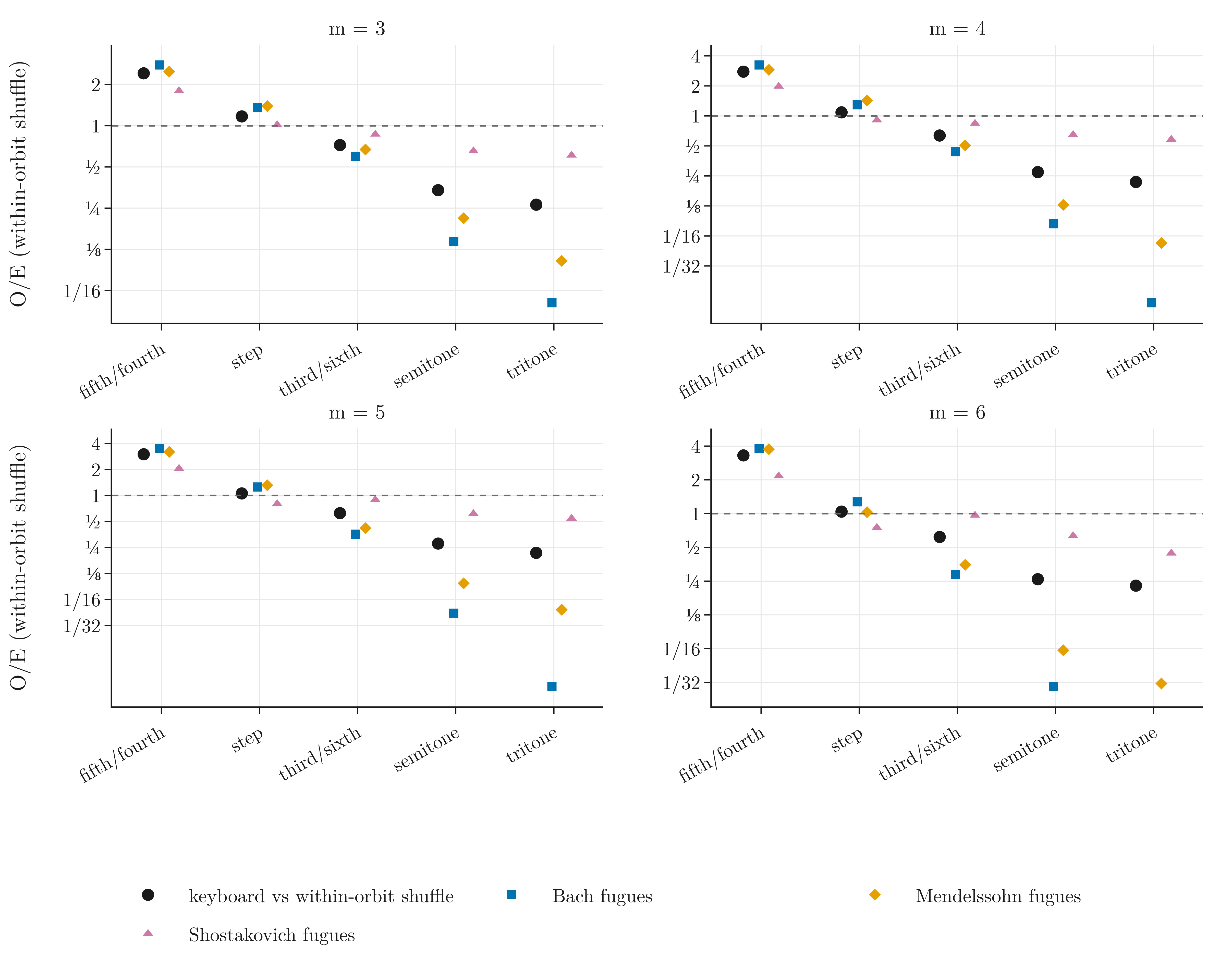}
\caption{Transposition intervals relative to the within-orbit shuffle, keyboard corpus and Prelude \&
Fugue fugues by composer, $m = 3$--$6$.}
\label{fig:si-rel-intervals}
\end{figure}

\begin{figure}[htbp]
\centering
\includegraphics[width=\linewidth]{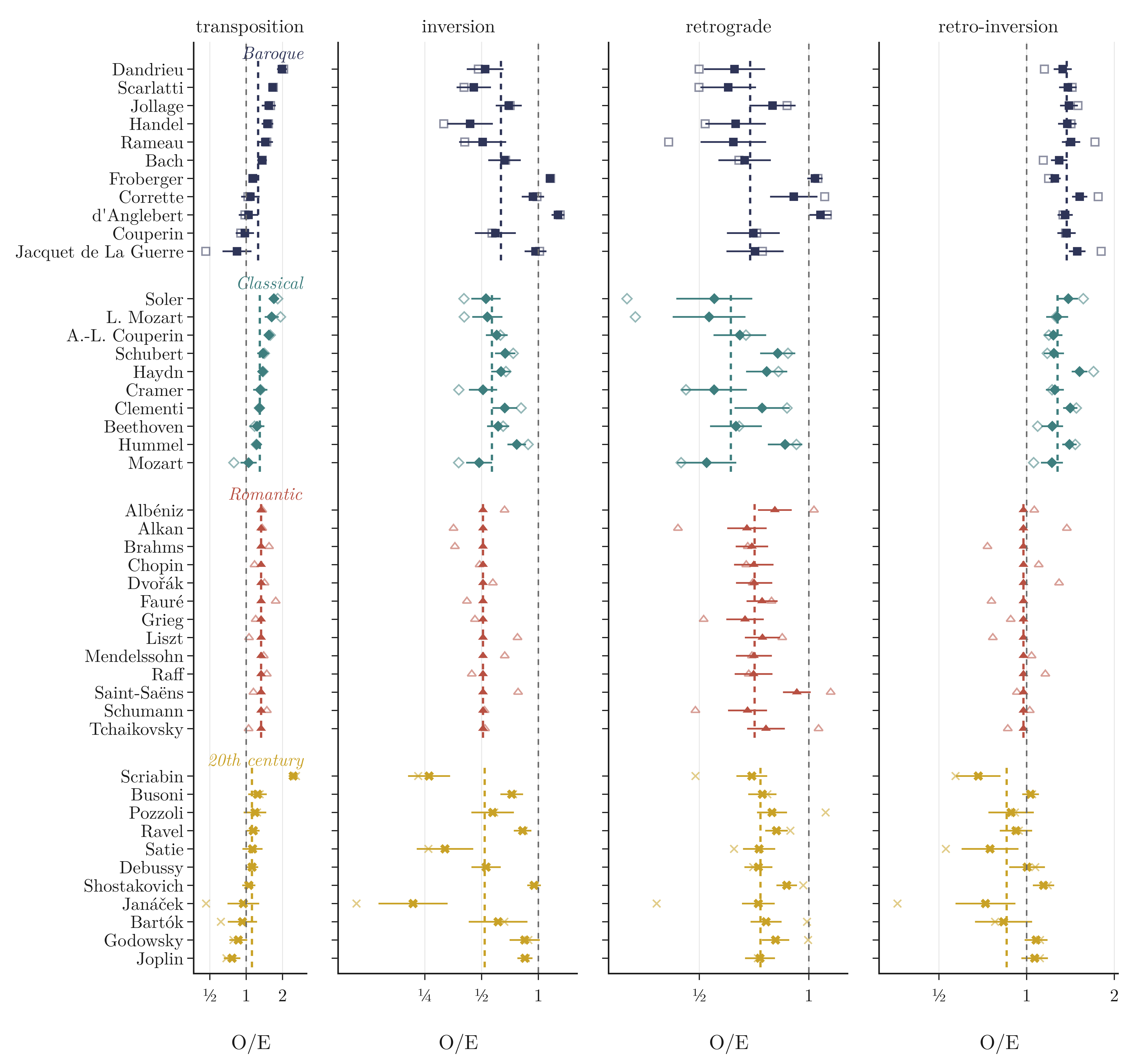}
\caption{Composer means of each relation type relative to the within-orbit shuffle, $m = 4$, shrunk
toward the period mean (filled, $\pm 1$ s.d.) with the unshrunk mean (open) and the period mean
(dotted). Composers are ordered by transposition in every panel.}
\label{fig:si-rel-composers}
\end{figure}

\begin{figure}[htbp]
\centering
\includegraphics[width=\linewidth]{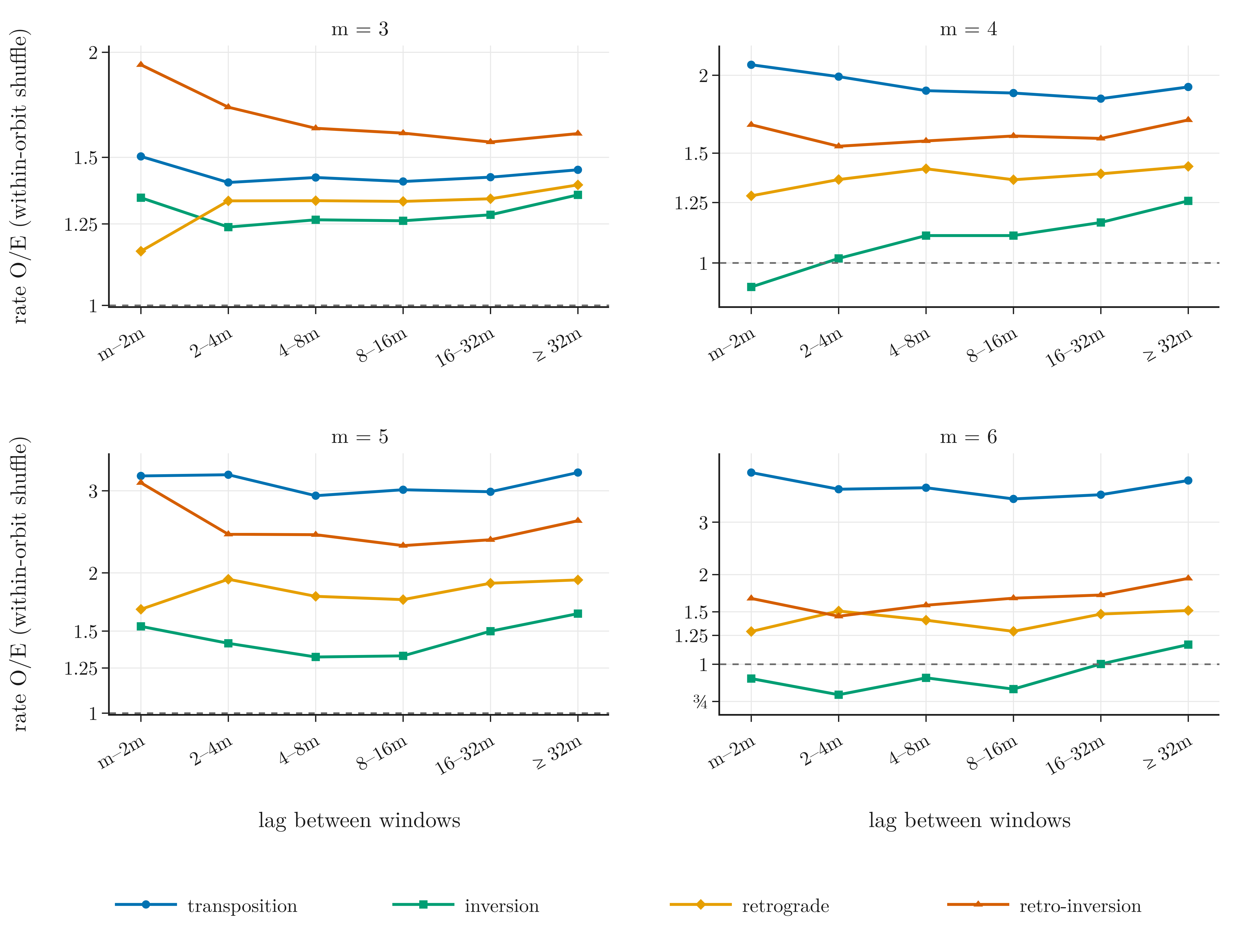}
\caption{Rate of each relation type per pair relative to the within-orbit shuffle, by the distance
$\Delta$ between the starts of the two windows, grouped in doubling ranges $[m, 2m), [2m, 4m), \dots,
\ge 32m$ (every integer $\Delta \ge m$ counted). Every rate exceeds one because rotated relations
are rarer than under the null, so types are compared within a distance range.}
\label{fig:si-rel-lags}
\end{figure}

\begin{figure}[htbp]
\centering
\includegraphics[width=\linewidth]{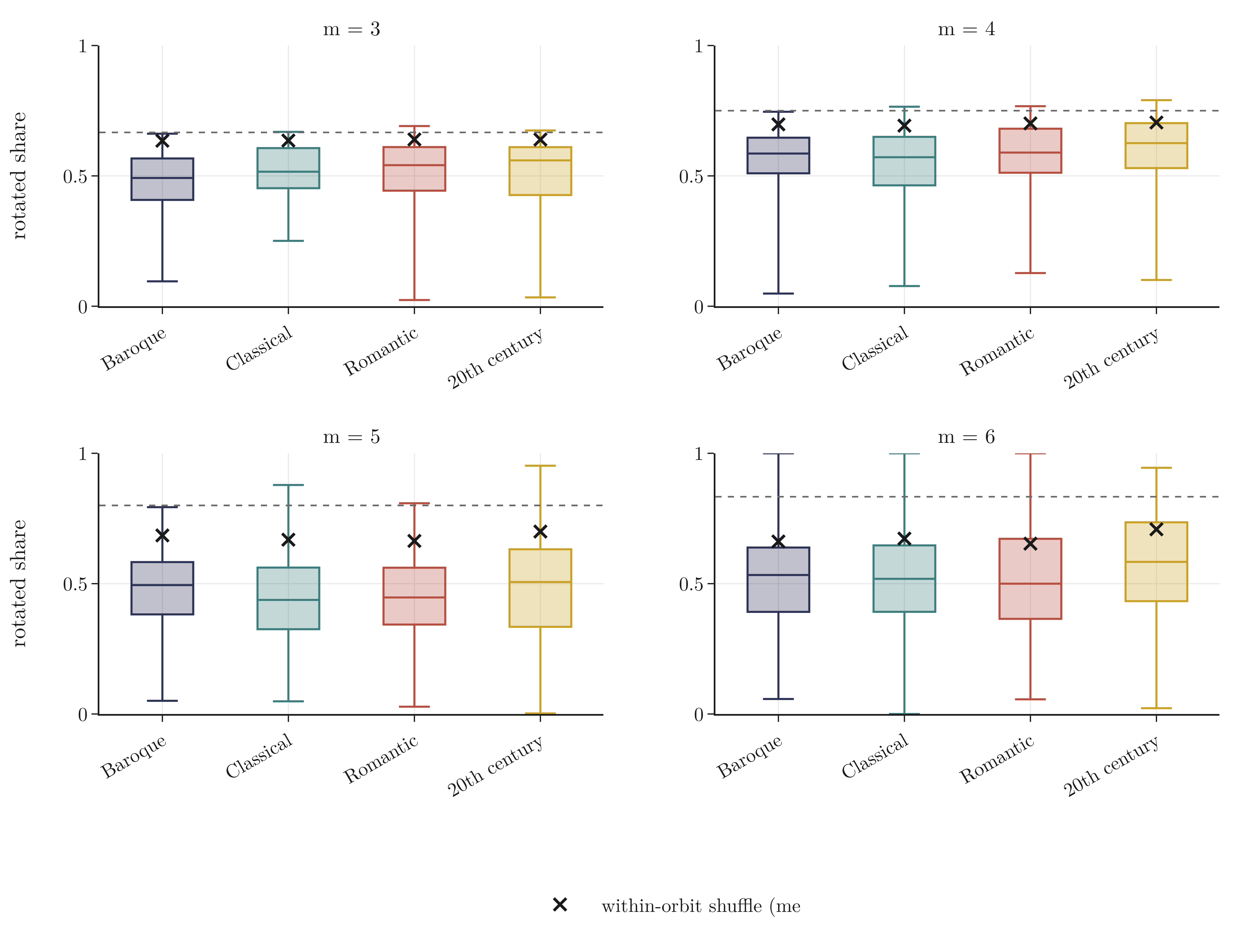}
\caption{Share of relations that include a cyclic shift, per piece by period (boxes), with
the median under the within-orbit shuffle (cross) and the uniform expectation $(m - 1)/m$ (dashed).}
\label{fig:si-rel-rotated}
\end{figure}

\begin{figure}[htbp]
\centering
\includegraphics[width=\linewidth]{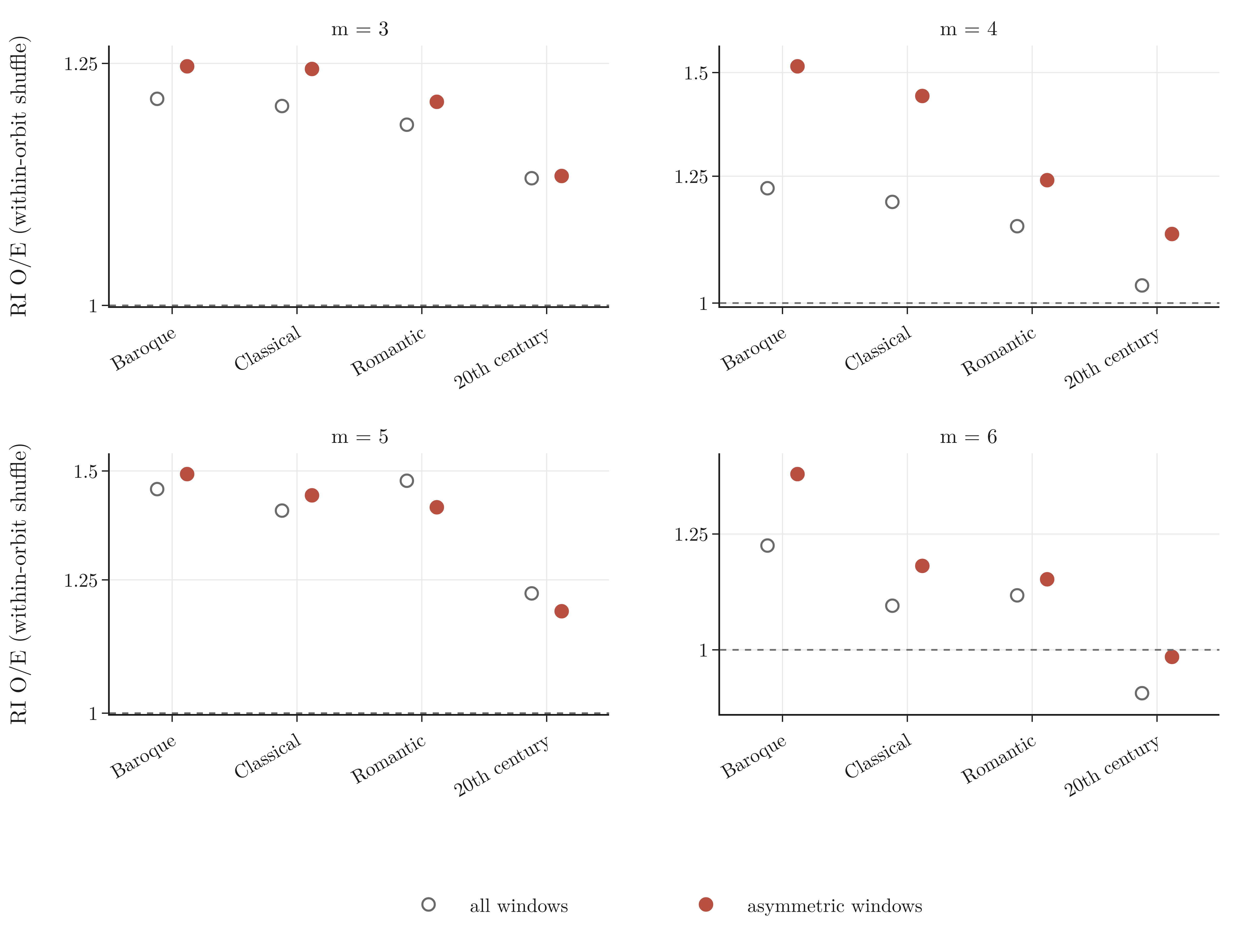}
\caption{Retrograde-inversion relative to the within-orbit shuffle over all windows (open) and over
asymmetric windows only (filled), by period. Over all windows, transpositions of
constant-step runs are also counted as retrograde-inversions, which dilutes the
comparison.}
\label{fig:si-rel-leakage}
\end{figure}

\begin{figure}[htbp]
\centering
\includegraphics[width=0.8\linewidth]{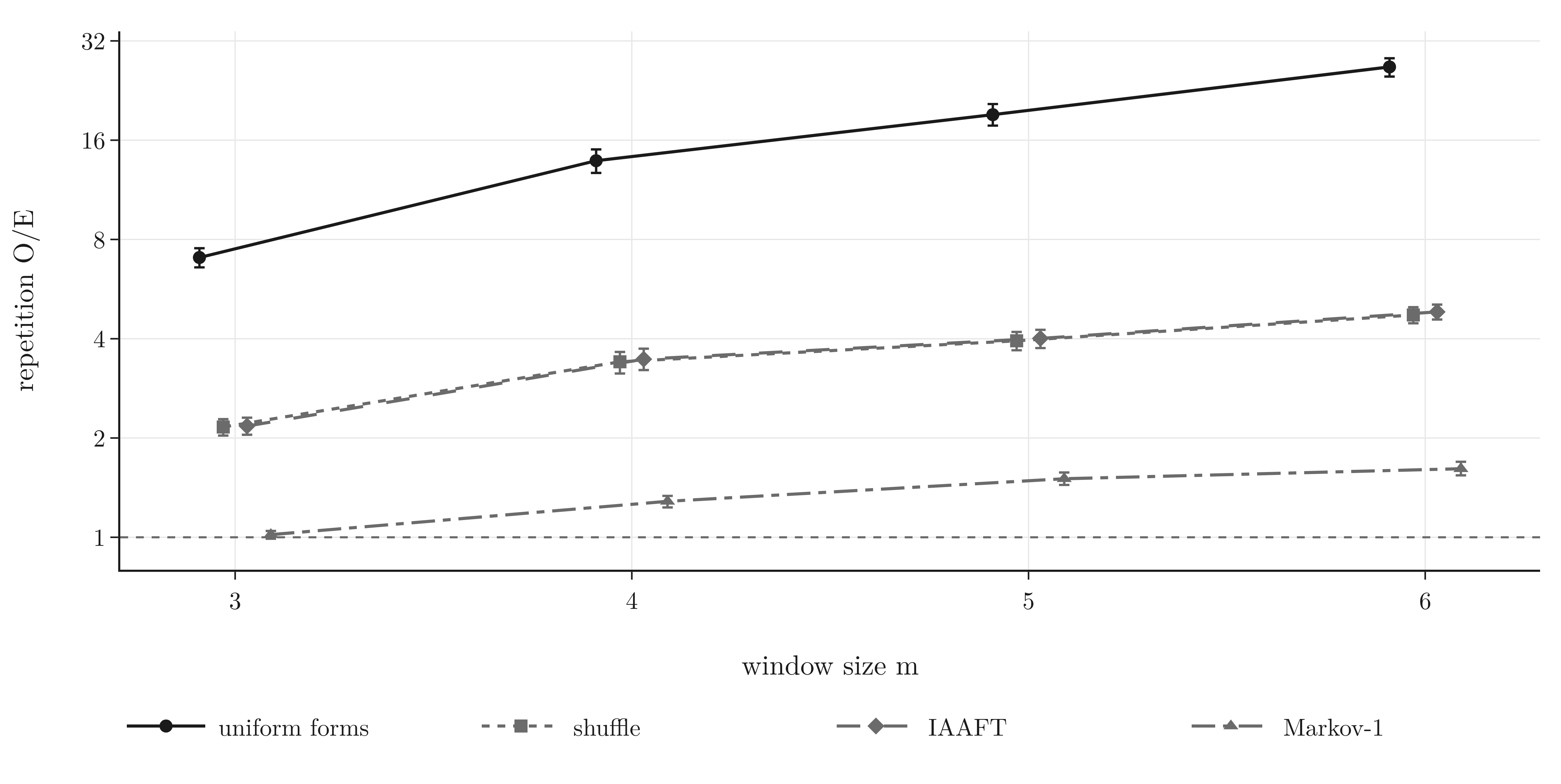}
\caption{Repetition relative to uniform forms and to the three surrogates, $m = 3$--$6$,
with $90\%$ composer-bootstrap intervals. The within-orbit shuffle keeps repetition by construction.}
\label{fig:si-rel-repetition}
\end{figure}

\textbf{Entropy of the operations.} Over the $48$ unrotated operations the chain rule gives
\begin{equation}
\label{eq:hop}
H_{\text{op}} = H_{\text{type}} + \sum_t P(t)\, H(k \mid t), \qquad H(k \mid \text{repetition}) = 0,
\end{equation}
with $H_{\text{type}} = H_{\text{rel}}$ over pairs of asymmetric windows and
$H_{\text{op}} \le \log_2 48$. Each $H(k \mid t)$ is invariant under transposing the whole line,
which shifts $k$ by $0$ or by twice the transposition. Every term is computed in each of the $199$
within-orbit shuffles and averaged, so observed and expected values rest on the same pairs
(Figure~\ref{fig:si-rel-entropy}). $H_{\text{type}}$ stays at the shuffle at every $m$, and
$H(k \mid t)$ lies below it in $70$--$88\%$ of pieces at $m = 3$--$5$. $H_{\text{op}}$ lies below
only at $m = 3$, since at larger $m$ transposition (eleven intervals) replaces repetition (one), so
it is not reported alone. The transposition-interval deficit, shuffle minus observed
$H(k \mid \text{T})$, falls across periods ($\rho = -0.38$, $-0.20$, and $-0.28$ at $m = 3$, $4$,
and $5$, composers permuted, $p \le 0.001$) and across the fugues from Bach to Shostakovich
(Jonckheere--Terpstra $p < 0.001$ at $m = 3$--$4$, $p = 0.008$ at $m = 5$, not at $m = 6$).

\begin{figure}[htbp]
\centering
\includegraphics[width=\linewidth]{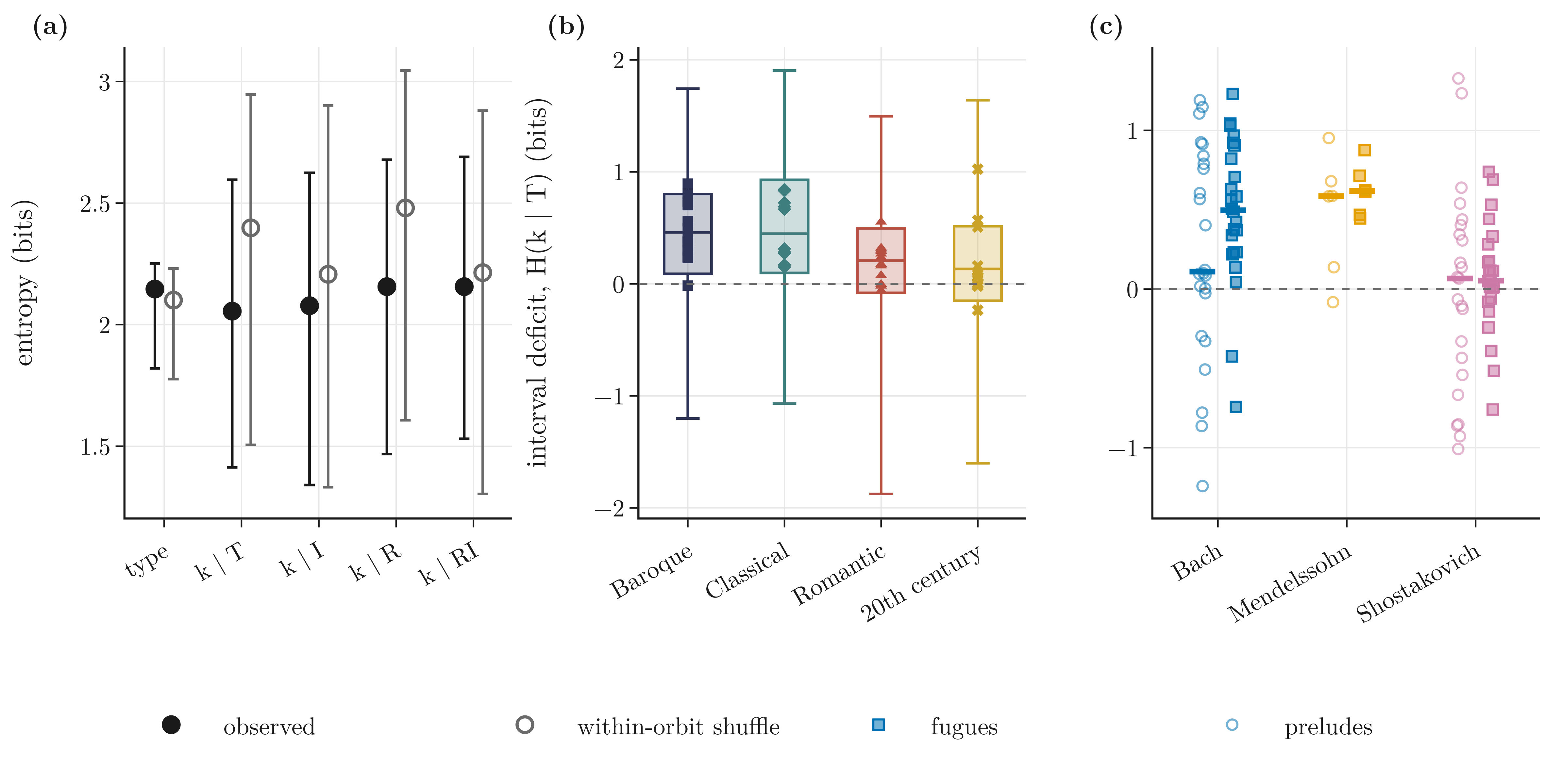}
\caption{Entropy of the operations relating windows, $m = 4$, pairs of asymmetric windows.
\textbf{(a)} Type entropy and interval entropy within each type, observed and under the within-orbit
shuffle (median and interquartile range over pieces). \textbf{(b)} Deficit of the transposition
intervals, shuffle minus observed $H(k \mid \text{T})$, by period (boxes) with composer medians (dots).
\textbf{(c)} The same deficit in the Prelude \& Fugue set, fugues (squares) and preludes (circles) by
composer, with medians as bars.}
\label{fig:si-rel-entropy}
\end{figure}

\subsection{Symmetric phrases and themes}

\begin{figure}[htbp]
\centering
\includegraphics[width=0.8\linewidth]{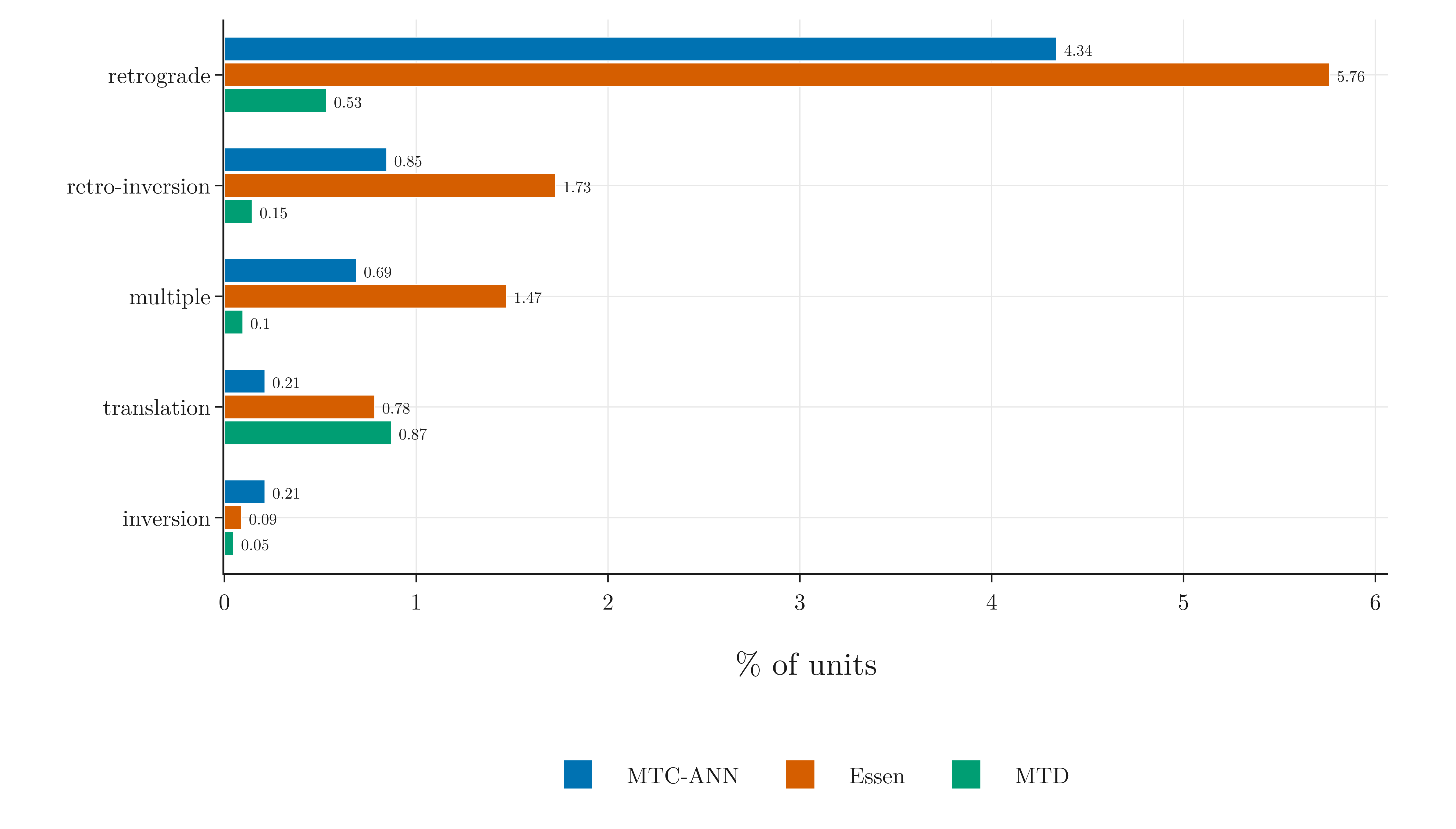}
\caption{Share of segments with each exact symmetry class, per census corpus.}
\label{fig:si-census-class}
\end{figure}

\begin{figure}[htbp]
\centering
\includegraphics[width=\linewidth]{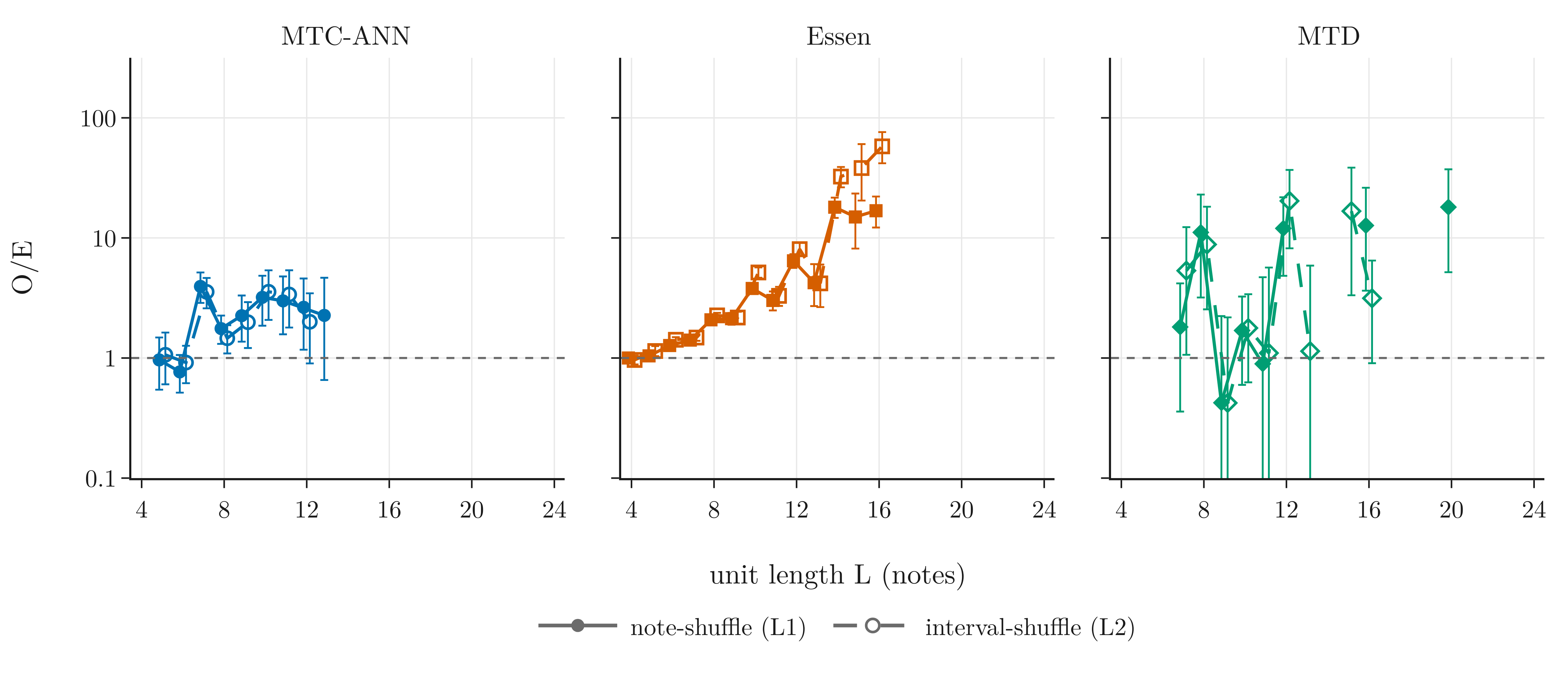}
\caption{Enrichment $O_L/E_L$ under the note-shuffle (filled) and interval-shuffle
(open) nulls, per exact length $L$ (log scale; dashed line: no enrichment).}
\label{fig:si-census-null}
\end{figure}

Class shares and the interval-shuffle check are in Figures~\ref{fig:si-census-class}
and~\ref{fig:si-census-null}. \textbf{Clustering.} The census null treats segments as independent trials with their own $p_i$,
but segments of one song or work are dependent and may repeat literally. Resampling whole songs
(Essen, MTC-ANN) or works (MTD) gives $O/E = 1.61$ $[1.34, 2.02]$ (MTC-ANN), $1.59$
$[1.54, 1.64]$ (Essen), and $5.08$ $[3.01, 9.74]$ (MTD). Keeping each literally repeated phrase
once per song changes $O/E$ by at most $12.9 \to 11.2$ (Essen, $L = 13$--$16$). Per-bin intervals
match the Garwood intervals for Essen and MTC-ANN but are much wider for long MTD themes
($L = 17$--$24$: $[16.6, 1136]$).

\subsection{Symmetry-subgroup richness}
\label{app:richness}

Exact stabilizer subgroups are counted as species and pieces as samples, as harmonic vocabulary
was counted on the same corpus \citeSI{SI:Serra-Peralta2021}. For $N$ pieces, $Y_i$ is the number
of pieces containing subgroup $i$, $S_{\mathrm{obs}}$ the number of subgroups observed, and $Q_k$
the number found in exactly $k$ pieces. Since observed richness grows with sampling effort,
periods are compared at a common effort $k$ by incidence rarefaction
\citeSI{SI:colwell2004,SI:gotelli2001},
\begin{equation}
\label{eq:rarefaction}
\tilde S(k) \;=\; S_{\mathrm{obs}} \;-\; \binom{N}{k}^{-1}
              \sum_{i} \binom{N - Y_i}{k},
\qquad k = 1, \dots, N,
\end{equation}
and total richness is estimated by Chao2 \citeSI{SI:chao1987,SI:chaochiu2016},
\begin{equation}
\label{eq:chao2}
\hat S_{\mathrm{Chao2}} \;=\;
\begin{cases}
S_{\mathrm{obs}} + \dfrac{Q_1^{2}}{2 Q_2}, & Q_2 > 0, \\[2.2ex]
S_{\mathrm{obs}} + \dfrac{Q_1 (Q_1 - 1)}{2}, & Q_2 = 0.
\end{cases}
\end{equation}
Rarefaction bands are the $2.5$--$97.5$th percentiles over $200$ random piece orders (without
replacement). They vanish at $k = N$ and are not intervals for the total. Chao2 intervals are
log-normal, with $\hat f_0 = \hat S_{\mathrm{Chao2}} - S_{\mathrm{obs}}$,
\begin{equation}
\label{eq:chao2ci}
\Big[\, S_{\mathrm{obs}} + \hat f_0 / C, \;\;
        S_{\mathrm{obs}} + \hat f_0\, C \,\Big],
\qquad
C = \exp\!\Big(z_{1-\alpha/2}
      \sqrt{\ln\big(1 + \widehat{\mathrm{var}}(\hat S)/\hat f_0^{2}\big)}\Big),
\end{equation}
and percentile bootstrap over $B = 1999$ resamples of pieces with replacement
\citeSI{SI:davison1997,SI:efron1993,SI:manly2007}, which agree except at small $N$, where the
analytic interval is preferred. Both estimators are the $q = 0$ case of Chao et al.\
\citeSI{SI:chao2014}. Six-class curves saturate at $m = 3$, so curves use exact subgroups,
$m = 3$--$6$ (Figure~\ref{fig:genome-periods}b). Pieces are treated as independent, so pieces by
one composer are not kept together and the comparison between periods is descriptive.

\begingroup
\makeatletter
\renewcommand\@biblabel[1]{[S#1]}
\patchcmd{\@bibitem}{\the\value{\@listctr}}{S\the\value{\@listctr}}{}{\typeout{SI bibitem patch FAILED}}
\makeatother
\bibliographystyleSI{unsrt}
\bibliographySI{references_SI}
\endgroup